\documentclass[11pt,a4paper]{article}

\usepackage[utf8]{inputenc}
\usepackage[T1]{fontenc}
\usepackage{amsmath,amssymb,amsthm}
\theoremstyle{plain}
\newtheorem{theorem}{Theorem}[section]
\newtheorem{proposition}[theorem]{Proposition}

\theoremstyle{definition}

\theoremstyle{remark}
\newtheorem{remark}[theorem]{Remark}
\usepackage{graphicx}
\usepackage[colorlinks=true,linkcolor=blue,citecolor=blue,urlcolor=blue]{hyperref}
\usepackage{booktabs}
\usepackage{algorithm}
\usepackage{algpseudocode}
\usepackage{xcolor}
\usepackage[margin=1in]{geometry}
\usepackage{caption}
\usepackage{subcaption}
\usepackage{cite}
\usepackage{float}
\usepackage{enumitem}
\usepackage{microtype}
\usepackage{tikz}
\usetikzlibrary{arrows.meta,positioning,calc,decorations.pathreplacing,shapes.geometric,fit,backgrounds}
\usepackage{authblk}
\usepackage[numbers,square]{natbib}

\definecolor{accentblue}{HTML}{1A6FA3}
\definecolor{accentteal}{HTML}{1E9A8E}
\definecolor{accentamber}{HTML}{D4860A}
\definecolor{accentred}{HTML}{C0392B}
\definecolor{lightgray}{HTML}{F5F5F5}
\definecolor{navyblue}{HTML}{0D1B3E}

\tikzset{
  node/.style={circle,draw,minimum size=20pt,font=\small,thick},
  nodeS/.style={node,fill=accentblue!30,draw=accentblue},
  nodeT/.style={node,fill=accentamber!30,draw=accentamber},
  strong/.style={line width=1.8pt,accentred},
  weak/.style={dashed,gray!70},
  arr/.style={-{Stealth[length=5pt]},accentblue,thick},
}

\title{\textbf{Quantum Variational Approaches to the Maximum Independent Set Problem at Utility Scale}\\\large }
\author[1]{Kalyan Dasgupta\thanks{Corresponding author, \texttt{kalyand1@in.ibm.com}}}
\author[1]{Sumanta Mukherjee}
\author[2]{Dhriti Verma}
\author[1]{Surya Shravan Kumar Sajja}
\author[1]{Abhishek Singh}
\author[2]{Dzung Phan}
\author[2]{Jayant Kalagnanam}

\affil[1]{IBM Research, Bangalore, India}
\affil[2]{IBM Research, Yorktown Heights, NY, USA}

\date{}

\begin{document}
\maketitle

\begin{abstract}
Near-optimal solutions to Maximum Independent Set on dense graphs sit in local optima
that greedy correction and maximality extension cannot escape, even with repeated
restarts. We encode near-optimal seeds as a uniform quantum superposition on ancilla
qubits and evolve them under an excitation-preserving variational ansatz that holds the
search inside the feasible Hamming-weight subspace. A preprocessing stage of spectral
reordering and distance-based sparsification, together with history-guided
post-processing of the sampled bitstrings, takes the method to 200 nodes. The ansatz
entangles the seed branches and the bond dimension of the simulated state grows with
circuit depth, which places deeper circuits outside the reach of exact matrix product
state simulation at the bond dimensions available to us. Larger instances therefore need
quantum hardware. Measured on the data register alone, the superposition behaves as a
classical mixture over the seeds, so coherence between the branches has to be created and
then looked for. We do this with a CRZ phase layer and post-selection on the ancilla
qubits, which brings the branches into interference and exposes the cross terms.
The idea is to see if interference between near-optimal seeds
widens the range of independent sets the circuit returns. We measure the cross terms rather
than assume them and place an upper limit on their size, and we trace that limit to the
distance between the vertices that distinguish near-optimal solutions once the graph is
reordered.

Instances that fit within the qubit count of the hardware available to us need none of
the ancilla machinery. Standard VQE with this pipeline recovers the certified optimum for
instances up to 125 nodes, and these run on \texttt{ibm\_marrakesh} with parameters
transferred from noiseless simulation. The ancilla construction is introduced for the
sizes past that point. On a 180-node hard instance where standard VQE and our classical
post-processing stall at size 14, the superposition recovers the certified maximum
of 15, and an ablation over single seeds shows the superposition is responsible. Five
200-node instances are solved to the certified optimum, and on the 400-node
\texttt{brock400-1} benchmark the method reaches size 25 against a certified optimum of
27, with the recovered size improving as depth increases, which points at circuit
expressiveness as the bottleneck. These are the largest hard general-graph instances we
know of where a gate-based variational algorithm optimises the full circuit directly.
\end{abstract}

%=============================================================================
\section{Introduction}
\label{sec:intro}
%=============================================================================

The Maximum Independent Set (MIS) problem asks for the largest subset
of vertices in an undirected graph such that no two selected vertices
are adjacent. It belongs to the canonical NP-hard problems of combinatorics~\cite{garey1979computers}
and arises in numerous practical contexts: frequency assignment in wireless
networks \cite{joo2012local}, protein interaction network analysis~\cite{atias2012comparative},
compiler register allocation, and VLSI design automation \cite{butenko2006clique}.
Through the clique-independence duality, it is also intimately related to the maximum
clique and minimum vertex cover problems. Classical exact solvers based on
branch-and-bound and integer programming scale poorly beyond a few hundred
vertices for hard instances \cite{pardalos1994maximum}, prompting continued
interest in heuristic and metaheuristic methods as well as, more recently,
quantum computing.

Two principal paradigms have emerged for quantum optimization. Quantum annealing,
realized on D-Wave hardware \cite{johnson2011quantum}, encodes combinatorial problems
as Quadratic Unconstrained Binary Optimization (QUBO) instances and searches for
low-energy states by controlling a transverse-field Hamiltonian. Gate-based variational
algorithms, principally the Variational Quantum Eigensolver (VQE)~\cite{peruzzo2014variational}
and the Quantum Approximate Optimization Algorithm (QAOA) \cite{farhi2014quantum},
take a hybrid quantum-classical approach in which a parameterized quantum circuit is
optimized by a classical outer loop. Both families have been applied to graph-structured
combinatorial problems, and each faces a characteristic set of challenges in the
near-term (NISQ) era \cite{preskill2018quantum}: limited qubit counts, short coherence
times, constrained hardware connectivity, and the barren plateau phenomenon in which
gradient signals vanish exponentially in system size \cite{mcclean2018barren}.

Beyond these generic difficulties, the MIS problem presents a structural challenge
that is specific to large, dense graphs. When the target independent set is a
small fraction of the total vertex count (as is the case for the 180-node instance
studied here, where the MIS size is 15 out of 180), the optimal solution occupies
an exponentially small corner of the $2^n$-dimensional Hilbert space. Uniform
superposition initial states, which underpin standard QAOA, assign equal probability
weight to all $2^n$ basis states and thus give the MIS solution an initial amplitude
of only $2^{-n/2}$. Even after thousands of optimization iterations, the variational
circuit may converge to a good but suboptimal energy basin from which escape requires
coordinated multi-qubit changes rather than local perturbations.

These challenges were addressed through a pipeline of techniques that combine
graph preprocessing, variational optimization, and classical post-processing. Graph
reordering places structurally adjacent vertices at nearby qubit indices so that the most
important two-qubit interactions become nearest-neighbour in the circuit and require
no routing overhead. Distance-based sparsification drops long-range graph edges from the
QAOA ansatz, cutting the two-qubit gate count and reducing circuit depth. Energy
evaluation still uses the full Hamiltonian. A two-stage bitstring post-processing
pipeline drives raw quantum measurement outcomes to maximum independent sets. Together,
these techniques are enough to find the exact MIS on every instance that fits within the
qubit count of the hardware available to us. The 64-, 99- and 125-node instances are all
solved by standard VQE with this pipeline, and all three are executed on
\texttt{ibm\_marrakesh} (156 qubits) with the certified optimum preserved. No ancilla
machinery is required at these sizes.

For the 180-node instance, however, these methods encounter a hard barrier. Starting from
any single seed, the greedy correction and maximality extension pipeline consistently
returns solutions of size 14 regardless of the number
of optimization iterations or random restarts. Exact and specialised classical solvers do
reach the size-15 optimum on this instance, as reported in Section~\ref{sec:setup}, so the
barrier is a property of the local search used inside our pipeline and not a claim about
classical MIS algorithms in general. To overcome this barrier, we introduce a
technique that puts quantum superposition and interference to work on the seeds
themselves. Rather than
initializing the variational circuit in the computational zero state or a uniform
superposition, we encode a curated collection of the best classically-found near-optimal
solutions as a uniform quantum superposition using ancilla qubits, and we then evolve this
state under an excitation-preserving variational ansatz built on $fSim(\theta,\varphi)$ gates.
The $fSim$ gate with $\varphi \neq 0$ introduces a ZZ interaction that drives the MPS bond
dimension up with depth, so deep circuits fall outside the exact MPS budget used here and
hardware execution becomes the practical route at scale. The excitation-preserving
structure keeps the circuit within the Hilbert subspace of fixed Hamming weight throughout,
eliminating penalty terms and preserving fine-grained energy differences that distinguish
near-optimal configurations.

Measured on the data register alone, this multi-seed state is a
stochastic mixture over the seeds and is classically tractable, so it serves as the control
against which the quantum contribution is measured. We then introduce an ancilla interference
extension in
which a layer of controlled-RZ phase gates followed by Hadamard on the ancilla and post-selection
creates interference between seed branches, so as to amplify configurations that are
jointly supported by all seeds. The post-selected output retains the cross-term contributions
$\mathrm{Re}[A_k A_{k'}^*]$ that the mixture discards, and it is the non-Clifford variational
ansatz together with the interference post-selection that introduces quantum correlations
across the branches. We measure these contributions directly and find them small at the seed
compositions used here, and we identify the condition on the graph and the reordering under
which they can be larger.

We validate the pipeline across several benchmark instances. Within reach of the available
hardware, the 125-node \texttt{c125-9} instance is solved both in simulation
and on IBM Quantum hardware \texttt{ibm\_marrakesh}, confirming that the pipeline runs end
to end on a real device. Beyond that reach, the 200-node QOBLIB suite is solved to
optimality in simulation with the ancilla superposition. For the 400-node \texttt{brock400-1} instance, the excitation-preserving
ansatz reaches near-optimal IS sizes that standard VQE cannot achieve; the required circuit
depth scales with instance density and provides empirical evidence that hardware execution
will be necessary at this scale. Across the gate-based literature we surveyed, these are the
largest hard general-graph MIS instances to which a variational algorithm has been applied with
the full circuit optimised directly. The comparison is confined to gate-based algorithms, since
neutral atom platforms encode the independence constraint physically through the blockade radius
and operate on unit-disk graphs, which is a different setting.

Encoding several near-optimal solutions simultaneously as a quantum superposition
and optimizing over all of them in a single circuit run is not, by itself, new.
What makes it work here is that the excitation-preserving ansatz, a circuit built on
$\mathrm{fSim}(\theta,\varphi)$ gates, keeps all evolution confined to the
fixed-Hamming-weight subspace. Without this constraint the seeds immediately mix
into infeasible states, and the penalty term has to do the job of feasibility
enforcement, which is precisely what breaks the normal VQE on dense instances. With
it, the energy landscape retains the resolution needed to distinguish near-MIS
configurations from each other, and the shared optimizer finds parameters that
improve all seeds at once. Interference between the seed branches is then a separate
question, and one we answer by measurement rather than by assumption.

The interference extension came out of asking what is actually lost when we measure
only the data register. Tracing out the ancilla leaves a stochastic mixture with no cross-seed
correlations and no advantage over running the seeds independently. A CRZ phase layer before
the final ancilla Hadamard and post-selection recovers the cross-terms
$\mathrm{Re}[A_k A_{k'}^*]$ that the mixture discards. Whether the mixture and the post-selected output
differ in practice turns on whether two seed branches place amplitude on the same bitstring,
and quantifying that condition is part of what this work contributes.

The paper is organized as follows. Section~\ref{sec:related} surveys related work
on classical MIS solvers, quantum annealing, gate-based variational algorithms,
and neutral atom approaches. Section~\ref{sec:formulation} states the Ising Hamiltonian
formulation of MIS. Section~\ref{sec:methods} develops the full algorithmic framework,
covering graph preprocessing, the variational circuit constructions, the bitstring
post-processing pipeline, and the ancilla superposition and interference techniques.
Section~\ref{sec:theory} provides theoretical foundations of each component.
Section~\ref{sec:setup} describes the experimental setup and benchmark instances.
Section~\ref{sec:results} presents results across all instances, including the ablation
study, hardware validation, and scaling experiments. Section~\ref{sec:conclusion} summarises the
findings and identifies open directions.

%=============================================================================
\section{Related Work}
\label{sec:related}
%=============================================================================

Classical algorithms for MIS span exact solvers, approximation algorithms, and
heuristics. The best-known exact algorithms achieve time complexities of
$O(1.2^n)$ \cite{fomin2009measure} to $O(1.1996^n)$ ~\cite{xiao2013exact},
making them tractable only for small instances. Polynomial-time approximation
algorithms achieve ratios of $O(n/(\log n)^2)$ \cite{boppana1992approximating},
while inapproximability results show that a ratio of $n^{1-\varepsilon}$
requires NP-hard computation for any $\varepsilon > 0$ \cite{hastad1996clique}.
In practice, heuristic methods such as simulated annealing ~\cite{kirkpatrick1983optimization},
Tabu search, and evolutionary algorithms have been used on the DIMACS benchmark
suite \cite{johnson1996cliques}, which remains the standard reference for hard
instances.

Quantum annealing approaches to MIS encode the problem as a QUBO via Lucas's
Ising formulation \cite{lucas2014ising}. Studies using D-Wave hardware have
demonstrated MIS solutions on small instances after minor-embedding the problem
graph into the Chimera or Pegasus hardware topology \cite{choi2008minor,pelofske2019solving}.
The overhead introduced by chain-based embedding limits the effective problem
size and motivates interest in gate-model approaches, which are topology-agnostic
in principle.

QAOA was introduced by Farhi et al.\ \cite{farhi2014quantum} for binary
combinatorial optimization, with the central result that $p$ layers converge to
the optimal solution as $p \to \infty$. Its approximation properties for MaxCut
have been analyzed in detail \cite{wang2018quantum, farhi2020quantum}, and
application to maximum k-vertex cover has been demonstrated on small instances
\cite{cook2020quantum}. Bravyi et al.\ identified symmetry-induced obstructions
to QAOA performance on MIS~\cite{bravyi2020obstacles}, and the problem of
barren plateaus with large constraint penalties has been noted independently by
several groups. Warm-starting QAOA (initializing near a classically-found
solution rather than in the uniform superposition) was proposed by Egger et al.~\cite{egger2021warm} and
further developed by Tate et al.\ \cite{tate2023warm} using SDP relaxations.

VQE \cite{peruzzo2014variational} originated in quantum chemistry for estimating
molecular ground state energies but has since been applied widely to combinatorial
problems through hardware-efficient ansätze \cite{kandala2017hardware}. Symmetry-preserving
or excitation-preserving circuits, which conserve the number of excitations
(Hamming weight) in the quantum state, have been proposed as a way to embed feasibility
constraints directly into the circuit structure, avoiding expensive Lagrange
penalty terms \cite{gard2020efficient, hadfield2019quantum}. Such circuits are
built from fermionic simulation ($fSim$) gates \cite{arrazola2022universal} and
are well-suited to problems where the optimal solution has a known
or bounded cardinality.

The impact of qubit ordering on circuit efficiency has been studied in the context
of reducing SWAP overhead on hardware-constrained devices \cite{itoko2020optimization, li2019tackling}.
Bandwidth minimization of sparse matrices by the Cuthill-McKee algorithm \cite{cuthill1969reducing}
and its reverse variant motivated spectral reordering approaches that use the Fiedler vector of
the graph Laplacian ~\cite{fiedler1973algebraic} for a spectral embedding. Spectral sparsification
theory ~\cite{spielman2011spectral} provides a rigorous framework for approximating graphs by
sparse subgraphs while preserving spectral properties.

Classical post-processing of quantum outputs has been studied by
Shaydulin et al. ~\cite{shaydulin2019multistart}, who showed that multi-start heuristics
improve solution quality. Local search applied to decoded bitstrings,
in the spirit of the GRASP method ~\cite{festa2002randomized}, is now standard
practice in the field. A parallel line of work exploits the native structure of Rydberg
atom arrays, where the physical blockade radius implements the independence
constraint. Ebadi et al.\ demonstrated MIS solutions on 289-node unit-disk
graphs \cite{ebadi2022quantum}, and subsequent hardware improvements have extended
this approach to larger instances \cite{scholl2021quantum}.

The present work builds on warm-starting and excitation-preserving ideas in a specific way.
Rather than initializing near a single classically-found solution, we prepare a uniform
quantum superposition over multiple distinct near-optimal solutions, so that a single
circuit execution optimizes over all seed branches simultaneously. The excitation-preserving
ansatz makes this viable and without it, the seeds mix into infeasible states and the
penalty term has to carry all feasibility enforcement. Prepared in this way and measured on
the data register alone, the multi-seed state is equivalent to a stochastic mixture over the
seeds and remains classically simulable, so we treat it as the control arm throughout. The
step that goes beyond the mixture is the ancilla interference post-selection, which retains
the cross terms that tracing out the ancilla discards, and it is the comparison between the
two arms that isolates its effect. The bitstring post-processing pipeline
extends the local search approach of ~\cite{shaydulin2019multistart} with a two-sided
suppression mask and a maximality-augmentation step that together recover 2
to 3 additional vertices over raw quantum outputs. Unlike Rydberg-based approaches, the
framework places no geometric restriction on the problem graph and targets gate-based
processors directly.

%=============================================================================
\section{Problem Formulation}
\label{sec:formulation}
%=============================================================================

Let $G = (V, E)$ be an undirected graph on $n = |V|$ vertices. A subset $S \subseteq V$
is \emph{independent} if no two vertices in $S$ share an edge. The MIS problem seeks
the independent set of maximum cardinality,
\begin{equation}
  S^* \in \operatorname*{argmax}_{\substack{S \,\subseteq\, V \\ S \text{ independent}}} |S|.
  \label{eq:mis}
\end{equation}
To solve MIS with a quantum algorithm, the combinatorial structure is encoded in an Ising Hamiltonian whose ground state corresponds to the maximum independent set.

Associated with each vertex $i \in V$ a binary decision variable $x_i \in \{0, 1\}$, where $x_i = 1$ denotes that vertex $i$ belongs to the independent set and $x_i = 0$ denotes exclusion. The binary variable is related to an Ising spin variable $z_i \in \{-1, +1\}$ by the mapping
\begin{equation}
  x_i = \frac{1 - z_i}{2},
  \label{eq:mapping}
\end{equation}
so that $z_i = +1$ (qubit state $|0\rangle$) corresponds to $x_i = 0$ (vertex excluded) and $z_i = -1$ (qubit state $|1\rangle$) corresponds to $x_i = 1$ (vertex included). Substituting $x_i = (1-z_i)/2$ into the QUBO objective
$-\sum_i x_i + P\sum_{(i,j)\in E} x_i x_j$ and expanding gives the MIS Ising Hamiltonian
\begin{equation}
  H_{\mathrm{MIS}} = -\sum_{i \in V} z_i
                    + P \sum_{(i,j) \in E} \frac{(1-z_i)(1-z_j)}{4},
  \label{eq:ising}
\end{equation}
where the first term rewards including vertices ($z_i = -1$, i.e.\ $x_i=1$) and the second penalizes adjacent pairs: the factor $(1-z_i)(1-z_j)/4$ equals 1 when both vertices are in the set ($z_i=z_j=-1$) and 0 otherwise. Expanding and collecting terms, this reduces to
\begin{equation}
  H_{\mathrm{MIS}} = \sum_{i \in V} h_i\, z_i
                    + \sum_{(i,j) \in E} J_{ij}\, z_i z_j + C,
  \label{eq:ising2}
\end{equation}
with $h_i = \frac{1}{2} - \frac{P}{4}\deg(i)$, $J_{ij} = P/4$, and $C = P|E|/4 - n/2$. Choosing $P \geq 2$ is sufficient to make every constraint violation energetically unfavourable; a complete proof with an explicit energy gap formula is given in Sec.~S1 of the Supplemental Material. In practice, we use $P = 2$, which provides a comfortable margin for all three benchmark instances. The energy $\langle H_{\mathrm{MIS}} \rangle$ is a dimensionless quantity representing the expectation value of the Ising Hamiltonian; lower energies correspond to larger feasible independent sets, with the global minimum achieved at the maximum independent set. In the quantum circuit, each vertex is represented by a single qubit; the computational basis state $|x_0 x_1 \cdots x_{n-1}\rangle$ encodes the candidate solution in which $x_i = 1$ (qubit in state $|1\rangle$) denotes inclusion in the independent set. Table~\ref{tab:notation} collects the notation used throughout.

\begin{table}[!htbp]
\centering
\caption{Notation and terminology used throughout.}
\label{tab:notation}
\begin{tabular}{@{}ll@{}}
\toprule
Symbol / Term & Definition \\ \midrule
$n, |V|$ & Number of vertices (qubits) \\
$|E|$ & Number of graph edges \\
$P$ & Penalty weight for violated edges in $H_{\mathrm{MIS}}$ \\
$p$ & Number of ansatz layers (circuit depth parameter) \\
$k$ & Sparsification bandwidth (edges with $|\pi(i)-\pi(j)| \leq k$ retained) \\
$\chi$ & MPS bond dimension (entanglement capacity of tensor-network state) \\
Barren plateau & Region of parameter space with exponentially vanishing gradients \\
Heavy-hex & IBM's qubit connectivity: hexagonal lattice with additional ``heavy'' edges \\
fSim gate & Fermionic simulation gate; preserves excitation (Hamming weight) \\
EMA & Exponential moving average (for bitstring filtering heuristic) \\
\bottomrule
\end{tabular}
\end{table}

%=============================================================================
\section{Methodology}
\label{sec:methods}
%=============================================================================

Figure~\ref{fig:pipeline} summarises the full pipeline.
Steps~1 and~2 are about vertex reordering and distance-based sparsification. These steps
are applied universally before any circuit is constructed.
Step~3 runs VQE or QAOA with a CVaR objective and an optimizer; step~4 drives the
resulting bitstring samples to valid maximal independent sets via greedy repair and
stepwise extension.
When steps 1--4 are not enough, we bring in the ancilla in step~5.
This step encodes multiple near-optimal seeds as a quantum superposition under an
excitation-preserving ansatz, and further, adds a CRZ interference layer that creates
cross-term contributions absent in any classical mixture.
Each component is detailed in the subsections below.

\begin{figure}[t]
  \centering
  \includegraphics[width=0.9\linewidth]{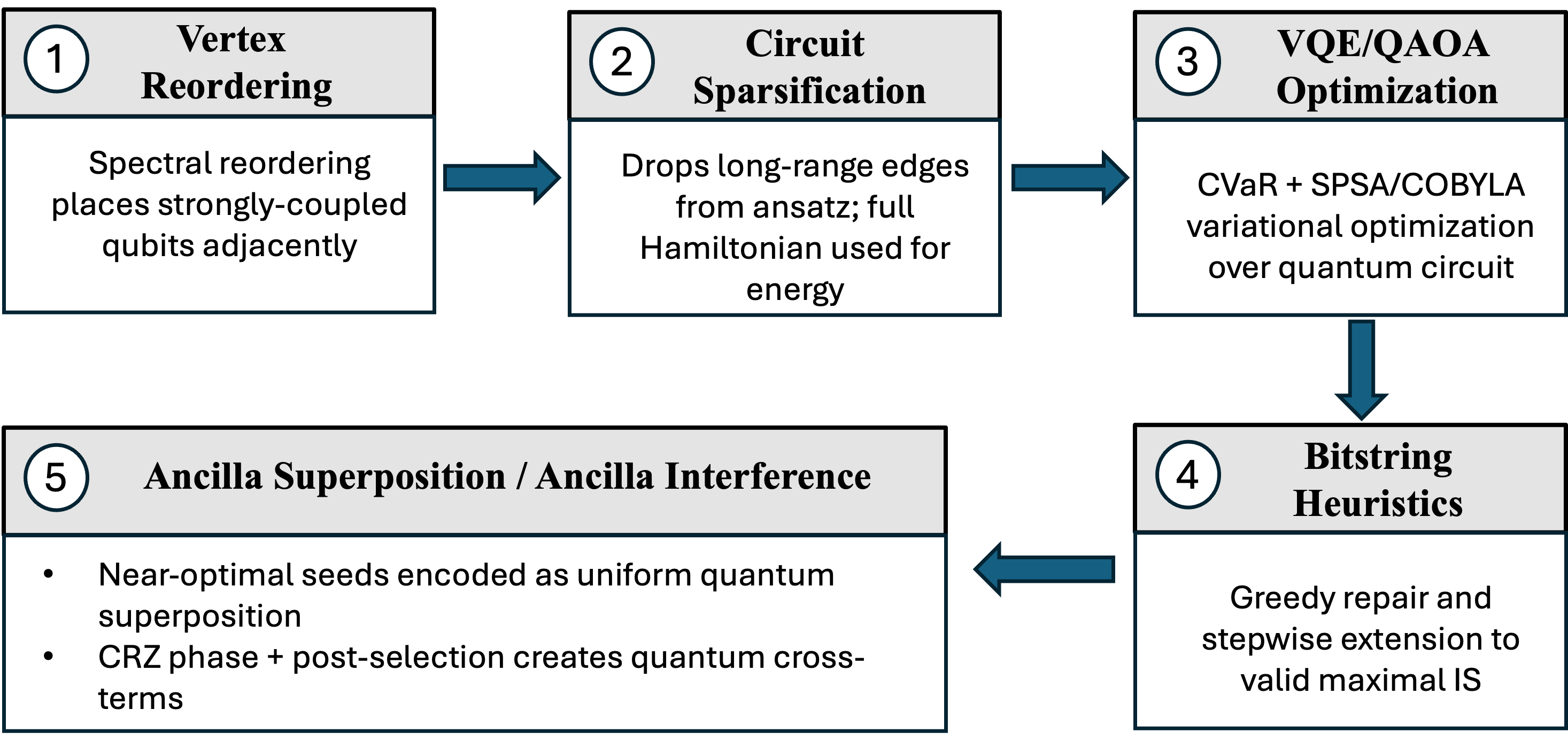}
  \caption{Overview of the six-step pipeline. Steps~1--4 are applied to all
  instances; steps~5 engage only when standard VQE and bitstring heuristics
  cannot recover the MIS.}
  \label{fig:pipeline}
\end{figure}

\subsection{Vertex Reordering}
\label{subsec:reorder}

The interaction graph of the MIS Ising Hamiltonian (\ref{eq:ising2})
is isomorphic to the original graph $G$: every edge $(i,j) \in E$
gives rise to a ZZ interaction $J_{ij} z_i z_j$ between qubits $i$ and
$j$, which must be implemented as a two-qubit gate in the variational
circuit. When vertex indices are assigned arbitrarily to qubit positions,
structurally adjacent vertices may be mapped to qubits far apart in
the qubit register, requiring long SWAP chains on hardware with limited
connectivity, or equivalently adding many layers of entangling gates
even in simulation.

The goal of vertex reordering is to find a permutation $\pi$ of the vertex set such that, after relabeling $i \mapsto \pi(i)$, strongly coupled qubit pairs have small index distance $|\pi(i) - \pi(j)|$. The benefit of reordering is not restricted to QAOA. Although the EfficientSU2 ansatz used for VQE entangles qubits only through a fixed nearest-neighbour CNOT chain (so reordering does not change the circuit structure), it does change the \emph{semantic alignment} between the circuit's entanglement structure and the problem's coupling structure. When strongly coupled qubits are adjacent after reordering, the nearest-neighbour CNOT gates directly entangle the most energetically important pairs, maximizing the leverage each variational parameter has over the energy landscape. Reordering therefore improves the effective expressibility per layer of EfficientSU2, allowing the same solution quality to be reached with fewer layers $p$.

Three reordering strategies were implemented and compared.

\textbf{Spectral reordering} uses the Fiedler vector \cite{fiedler1973algebraic}. The weighted graph Laplacian $L_{ij} = -|J_{ij}|$ for $i \neq j$ and $L_{ii} = \sum_{j \neq i} |J_{ij}|$ is constructed, and vertices are sorted by their coordinates in the eigenvector $\mathbf{v}_2$ corresponding to the second-smallest eigenvalue $\lambda_2$ (the algebraic connectivity). This one-dimensional spectral embedding places vertices on the real line so that strongly connected clusters map to contiguous intervals, globally minimizing the weighted sum of squared index distances $\sum_{(i,j)\in E} |J_{ij}| (\pi(i)-\pi(j))^2$.

\textbf{Coupling-weighted BFS} starts from the highest-degree vertex and performs a breadth-first traversal, visiting neighbours in descending order of a composite coupling score $\mathrm{score}(i,j) = |J_{ij}| \cdot \sqrt{d_i d_j}$ that combines interaction strength with the geometric mean of the endpoint degrees. This ensures that strongly and multiply-connected pairs are placed contiguously, but relies on local neighbourhood structure rather than global graph topology.

\textbf{Reverse Cuthill-McKee (RCM)} \cite{cuthill1969reducing} is the classical algorithm for sparse matrix bandwidth minimization. A BFS is initiated from the peripheral node (lowest degree), visiting neighbours in ascending degree order at each level, and the resulting ordering is then reversed. RCM directly minimizes the maximum index distance $\max_{(i,j)\in E}|\pi(i)-\pi(j)|$ (the matrix bandwidth).

Figure~\ref{fig:reorder} illustrates the effect of reordering on a small example. Table ~\ref{tab:reorder_compare} and the analysis in Section~\ref{subsec:reorder_compare} compare the three strategies quantitatively across all benchmark instances.

\begin{figure}[!htbp]
  \centering
  \includegraphics[width=0.9\linewidth]{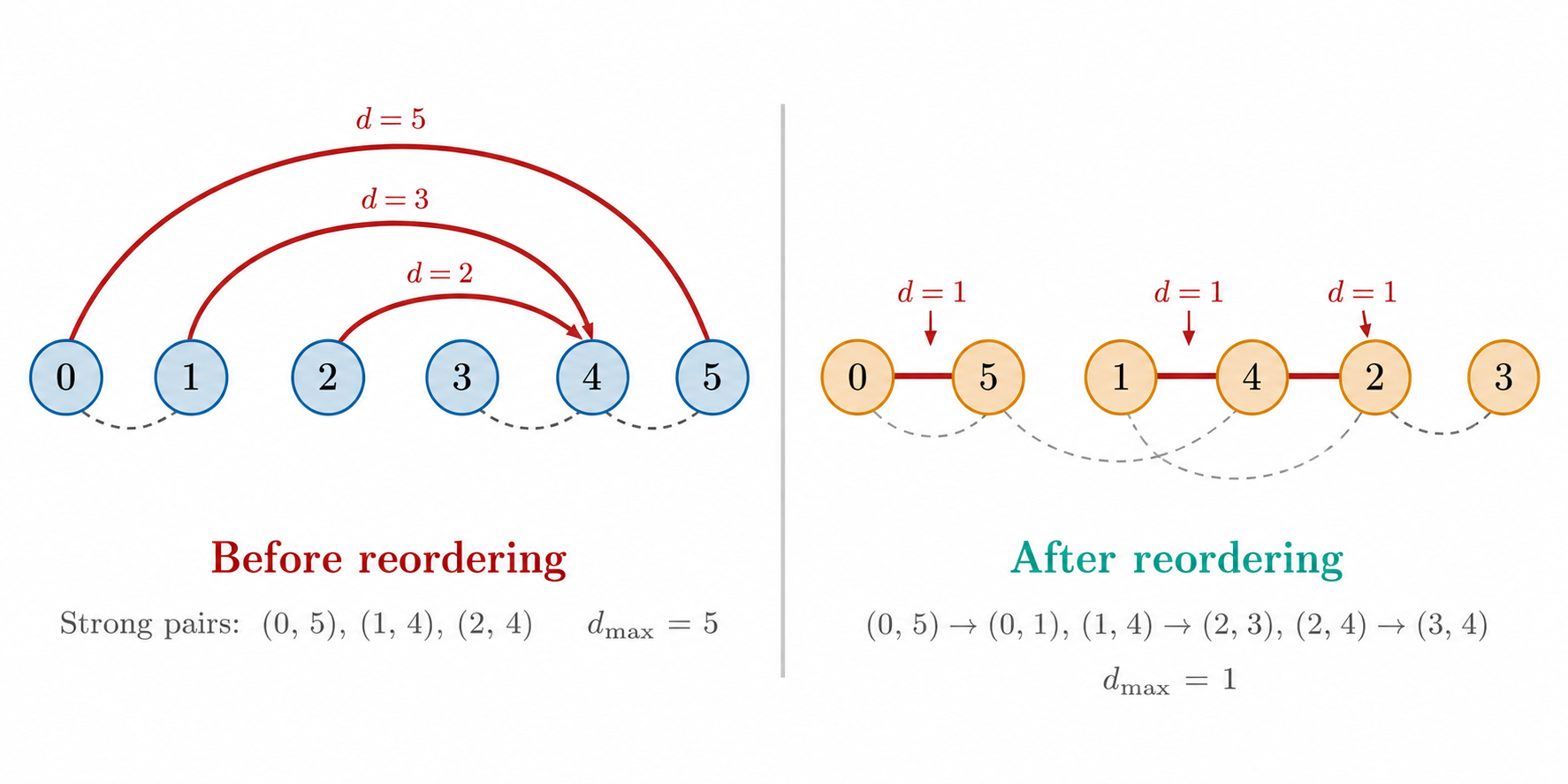}
  \caption{Effect of vertex reordering on a 6-node example. \textbf{Left:} Strong pairs span distances 2, 3, and 5. \textbf{Right:} After Fiedler reordering, all become nearest neighbours ($d=1$). Spectral reordering is consistently superior on dense instances (Table~\ref{tab:reorder_compare}).}
  \label{fig:reorder}
\end{figure}

\subsection{Distance-Based Sparsification}
\label{subsec:sparse}

An important structural distinction between VQE and QAOA determines whether sparsification is necessary. The EfficientSU2 ansatz used for VQE entangles qubits only through a fixed nearest-neighbour CNOT chain, regardless of the problem graph; its entanglement structure is independent of the Ising coupling matrix. Sparsification of the Hamiltonian, therefore, has no effect on VQE circuit depth. QAOA, by contrast, implements the cost unitary $e^{-i\gamma H_C}$ by inserting one RZZ gate per Ising coupling: every edge $(i,j) \in E$ with $|i-j| > 1$ adds a long-range two-qubit gate that inflates circuit depth. For QAOA, sparsification is therefore essential to keep the circuit depth within the MPS simulator and hardware limits.

After spectral reordering, the index distance $|i - j|$ between two coupled qubits correlates with the structural importance of their interaction. We exploit this through \emph{distance-based sparsification}: a bandwidth parameter $k$ is chosen, and only interactions with $|i - j| \leq k$ are used in the QAOA cost unitary, while the complete Hamiltonian is retained for energy evaluation throughout the optimization. Formally, the circuit coupling set is $J^{(k)}_{\mathrm{sparse}} = \{J_{ij} : |i-j| \leq k\}$, and the energy objective is evaluated using the full $J_{\mathrm{full}}$. This asymmetry between circuit construction and energy evaluation is the key design principle. The QAOA circuit operates in a reduced entanglement structure that keeps circuit depth tractable, while the classical optimizer receives accurate energy feedback from the complete Hamiltonian.

\textbf{Gradient bias from sparsification.}
The circuit prepares states reachable by the sparse coupling structure while the optimizer
always receives energy feedback from the complete objective, which raises the question of how
much the final solution quality suffers from the restricted state space. Write
$H_{\mathrm{full}} = H_{\mathrm{sparse}} + \delta H$, where
$\delta H = \sum_{|\pi(i)-\pi(j)|>k} J_{ij} Z_i Z_j$ collects all couplings dropped by the
bandwidth cutoff $k$, and define $\delta = \|\delta H\|_1$, the sum of the absolute values of
the dropped coupling coefficients. Since each dropped coupling has $|J_{ij}| = P/4$ and there
are $(1-\mathrm{coverage}(k)) \cdot |E|$ dropped edges,
\begin{equation}
  \delta \;\leq\; (1-\mathrm{coverage}(k)) \cdot |E| \cdot \frac{P}{4}.
\end{equation}
Let $\boldsymbol{\theta}^*$ be the parameters found by optimizing on the sparse circuit and let
$E_{\mathrm{full}}^*$ be the true ground-state energy of $H_{\mathrm{full}}$. The energy of any
state $|\psi\rangle$ under the full Hamiltonian satisfies
$\langle\psi|H_{\mathrm{full}}|\psi\rangle = \langle\psi|H_{\mathrm{sparse}}|\psi\rangle
+ \langle\psi|\delta H|\psi\rangle$, and the second term is bounded in absolute value by
$\|\delta H\|_1 = \delta$ for any normalized state, since each dropped term contributes at most
its coefficient. Applying this to both the achieved state $|\psi(\boldsymbol{\theta}^*)\rangle$
and the true ground state, and assuming the optimizer reaches within $\epsilon$ of the optimum
on the sparse landscape,
\begin{equation}
  E_{\mathrm{full}}(\boldsymbol{\theta}^*) \leq E_{\mathrm{full}}^* + \epsilon + 2\delta.
  \label{eq:sparse_bound}
\end{equation}
The factor of two arises because $\delta H$ shifts both the achieved energy upward by at most
$\delta$ and the true optimum downward by at most $\delta$ in the worst case, and both errors
accumulate in the same direction.

For the 64-node instance at $k=3$, coverage is $21\%$, so $\delta$ is large. For the 99-node
and 180-node instances at $k=1$, coverage is $2.3\%$ and $1.6\%$ respectively, making $\delta$
very large relative to the full Hamiltonian weight. Sparsification is therefore a circuit-depth
tool rather than an approximation tool. The bound~(\ref{eq:sparse_bound}) is loose in practice,
and the empirical results show that accurate solutions are still found because the
full-Hamiltonian energy evaluation keeps the optimizer correctly oriented even when the circuit
explores only a restricted region of the state space.

\subsection{Circuit Depth Analysis}
\label{subsec:depth}

After spectral reordering to bandwidth $B$, any two coupled qubits $(i,j)$ satisfy
$|\pi(i)-\pi(j)| \leq B$. On a linear qubit chain, all pairs at index distance $d$ can be
entangled in a single parallel layer of two-qubit gates using odd-even decomposition. The QAOA
cost unitary $e^{-i\gamma H_C}$ implements one $\mathrm{RZZ}(2\gamma J_{ij})$ gate per
coupling, and after grouping couplings by index distance $d=1,2,\ldots,B$, each group can be
executed in one parallel layer. The total number of parallel layers per QAOA repetition is
therefore at most $B$, giving an overall circuit depth of $O(B \cdot p)$ for $p$ repetitions
against $O(|E| \cdot p)$ for the unsparsified, unordered circuit. For the instances studied
here, $B \leq 3$ after sparsification while $|E|$ ranges from 543 to 5277, which is the source
of the reduction measured below.

Figure~\ref{fig:depth} shows two-qubit circuit depths for the 64-node instance. On the MPS simulator, vanilla QAOA at $p=1$ requires two-qubit depth of 60; after reordering and sparsification to $k=2$, this falls to 10 (6$\times$ reduction). On FakeSherbrooke (IBM's 127-qubit heavy-hex topology), the reduction is from 562 to 30 due to eliminated SWAP routing.

The MPS simulator's bond dimension grows exponentially with $k$ and $p$, limiting experiments to $k \leq 3$ for the 64-node instance. $k=4$
was explored separately for hardware parameter transfer (see Section~\ref{subsec:hw_setup}). For the denser 99- and 180-node instances, only $k=1$ is tractable at which point QAOA reduces to nearest-neighbour gates and offers no advantage over VQE. QAOA was therefore used only for the 64-node instance.

Figure~\ref{fig:coverage} shows edge coverage vs.\ bandwidth $k$. At $k=3$, the 64-node instance retains 21\% of edges; the 99- and 180-node instances retain only 2.3\% and 1.6\% at $k=1$, explaining their poor QAOA diversity.

\begin{figure}[!htbp]
  \centering
  \includegraphics[width=\linewidth]{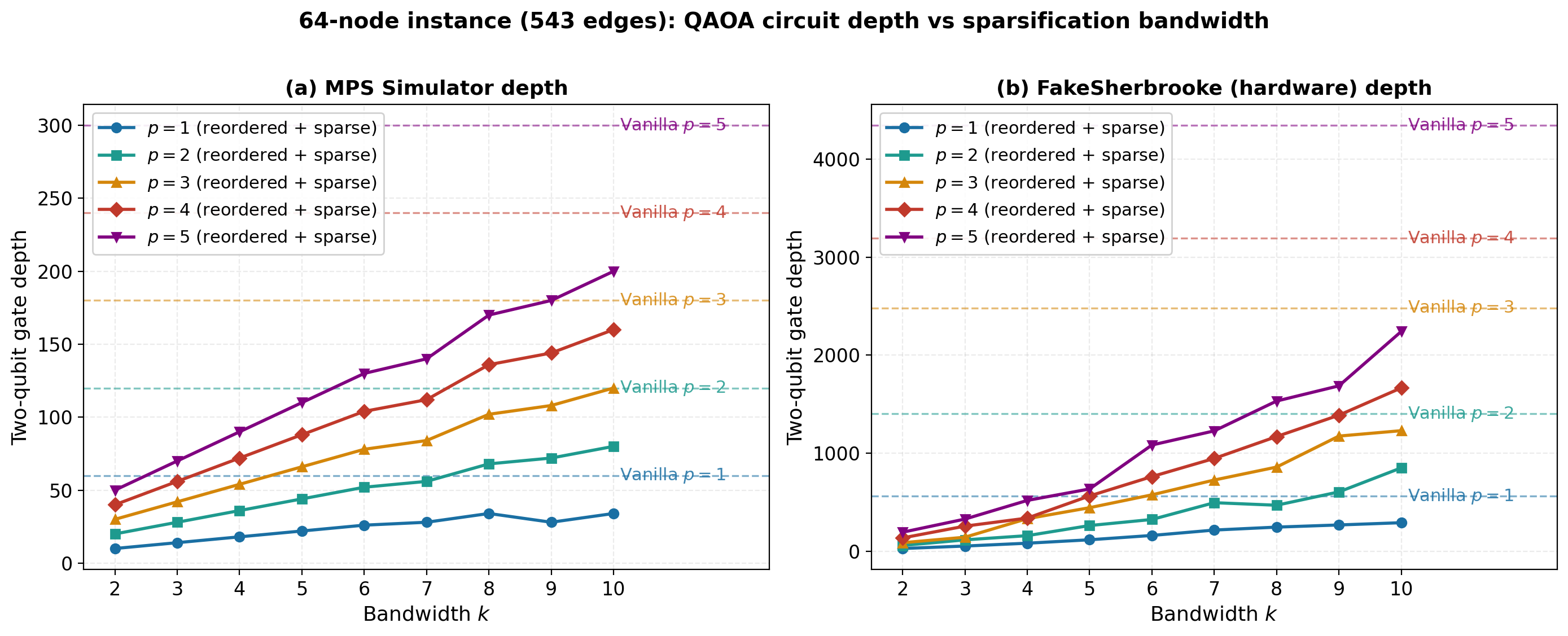}
  \caption{QAOA circuit depth analysis for the 64-node instance (543 edges).
    \textbf{(a)} MPS simulator two-qubit gate depth as a function of sparsification bandwidth $k$, for $p = 1$--$5$ (solid curves). Dashed horizontal lines show the corresponding vanilla (no reordering, no sparsification) depths at $p=1$--$5$, from bottom to top.
    \textbf{(b)} Compiled depth on FakeSherbrooke (IBM Sherbrooke 127-qubit heavy-hex hardware), including full SWAP routing and basis-gate decomposition. The hardware depth is far more sensitive to $k$ because every long-range gate requires a SWAP chain on the heavy-hex topology.
    VQE with EfficientSU2 is not shown because its depth is independent of sparsification. The ansatz entangles only adjacent qubits regardless of the problem graph.}
  \label{fig:depth}
\end{figure}

\begin{figure}[!htbp]
  \centering
  \includegraphics[width=\linewidth]{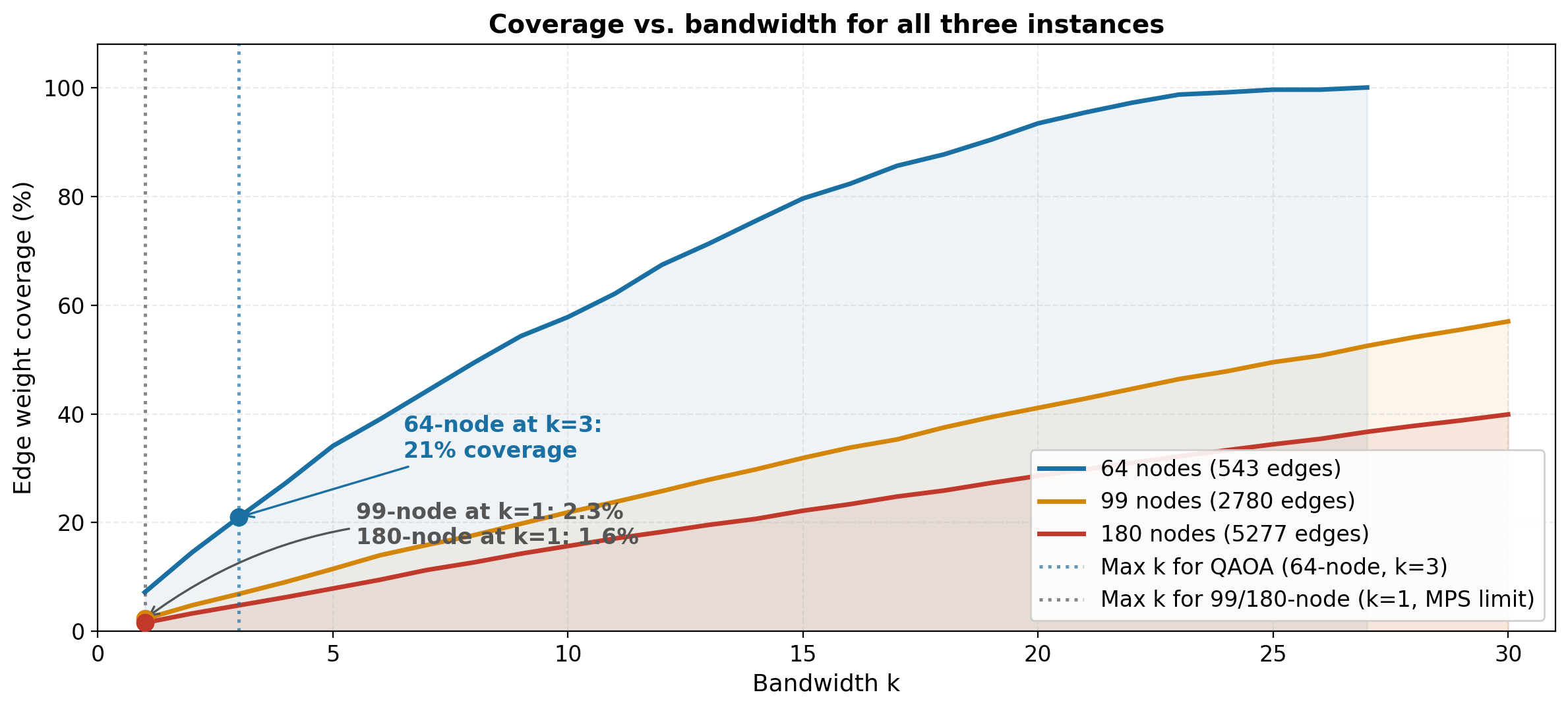}
  \caption{Edge weight coverage as a function of bandwidth $k$ for all three instances. Coverage measures what fraction of the total Ising coupling weight is retained in the QAOA circuit after sparsification. The 64-node instance (blue, 543 edges) has the highest density relative to its size and reaches 21\% coverage at $k=3$, the maximum bandwidth explored in QAOA experiments. The 99-node (amber, 2780 edges) and 180-node (red, 5277 edges) instances are far denser, and their coverage is only 2.3\% and 1.6\%, respectively, at $k=1$, the maximum the MPS simulator could handle. The dotted vertical lines mark these experimental limits. At $k=1$, the QAOA cost unitary contains only nearest-neighbour gates and is structurally equivalent to VQE, removing the motivation to pursue QAOA for the two larger instances.}
  \label{fig:coverage}
\end{figure}

\subsection{Reordering Strategy Comparison}
\label{subsec:reorder_compare}

Table~\ref{tab:reorder_compare} shows the number of edges retained at selected bandwidths for all three reordering strategies across the three benchmark instances. The results reveal a density-dependent pattern that motivates the choice of spectral reordering for the larger instances.

\begin{table}[!htbp]
\centering
\caption{Edges retained after reordering and sparsification for three reordering strategies across all three benchmark instances. Bold indicates the highest number of edges retained for each instance and bandwidth. A higher edge count at the same $k$ means better coverage of the Ising coupling structure in the QAOA circuit.}
\label{tab:reorder_compare}
\begin{tabular}{@{}lrrrrrrrrrrrr@{}}
\toprule
& \multicolumn{4}{c}{64-node (543 edges)}
& \multicolumn{4}{c}{99-node (2780 edges)}
& \multicolumn{4}{c}{180-node (5277 edges)} \\
\cmidrule(lr){2-5}\cmidrule(lr){6-9}\cmidrule(lr){10-13}
Method & $k{=}2$ & $k{=}5$ & $k{=}7$ & $k{=}10$
       & $k{=}2$ & $k{=}5$ & $k{=}7$ & $k{=}10$
       & $k{=}2$ & $k{=}5$ & $k{=}7$ & $k{=}10$ \\ \midrule
Spectral  & 79  & 185          & 240          & \textbf{314}
          & \textbf{134} & \textbf{319} & \textbf{442} & \textbf{609}
          & \textbf{175} & \textbf{416} & \textbf{594} & \textbf{831} \\
BFS       & 72  & 145          & 184          & 230
          & 107          & 271          & 386          & 544
          & 131          & 313          & 428          & 593 \\
RCM       & \textbf{99} & \textbf{192} & \textbf{242} & 309
          & 109          & 287          & 405          & 568
          & 125          & 308          & 427          & 631 \\ \bottomrule
\end{tabular}
\end{table}

RCM is competitive on the sparse 64-node instance but spectral reordering dominates on the denser 99- and 180-node instances, retaining 23--34\% more edges at tight bandwidths. The Fiedler vector's global optimization captures long-range connectivity that local methods (RCM, BFS) miss as density increases. We adopt spectral reordering throughout.

\subsection{Variational Quantum Eigensolver}
\label{subsec:vqe}

VQE \cite{peruzzo2014variational} is a hybrid quantum-classical algorithm that minimizes the energy expectation value $E(\boldsymbol{\theta}) = \langle\psi(\boldsymbol{\theta})|H|\psi(\boldsymbol{\theta})\rangle$ over the parameters $\boldsymbol{\theta}$ of a parameterized quantum state $|\psi(\boldsymbol{\theta})\rangle$ prepared by the ansatz circuit. The classical optimizer receives the expectation value from quantum measurements at each iteration and proposes updated parameters. We employ the \emph{Efficient SU(2)} ansatz \cite{kandala2017hardware}: each of $p$ repetition layers applies single-qubit $R_Y(\theta)R_X(\phi)$ rotations to every qubit, followed by a layer of CNOT gates connecting adjacent qubits in a linear chain. The total parameter count is $2n(p+1)$, growing linearly in system size for fixed $p$.

Two classical optimizers were evaluated. COBYLA (Constrained Optimization BY Linear Approximation) \cite{powell1994direct} is a derivative-free sequential
linear programming method that constructs local linear approximations to the objective. It converges reliably on small parameter spaces but tends to stall in high-dimensional landscapes beyond a few hundred parameters. SPSA (Simultaneous Perturbation Stochastic Approximation) \cite{spall1992multivariate}
estimates the gradient using two objective evaluations at simultaneously perturbed parameter values, with perturbation directions drawn randomly. This reduces the gradient estimation cost from $O(n)$ to $O(1)$ function evaluations per step and is suited to shot-noise-dominated quantum objectives, where individual energy measurements carry substantial statistical uncertainty. In our experiments, SPSA consistently outperformed COBYLA on larger instances, where the parameter count reaches $2 \times 180 \times 3 = 1080$.

For all instances except Phase~3 of the 180-node pipeline, we replace the standard expectation value $\langle H \rangle$ with the \emph{Conditional Value at Risk} (CVaR) objective \cite{barkoutsos2020improving} with $\alpha \in [0.1, 0.2]$. Rather than minimizing the average energy over all measurement outcomes, CVaR minimizes the average energy of the lowest-energy $\alpha$-fraction of samples, concentrating the gradient signal on the best bitstrings and reducing the influence of high-energy outliers. This helps most when the measurement distribution is broad and the MIS configurations are rare: CVaR effectively focuses the optimizer on the tail of the distribution where good solutions reside. For the excitation-preserving VQE phase of the 180-node pipeline
(Phase~3, described in Section~\ref{subsec:ancilla}), the ancilla superposition initialization already places the circuit near good solutions, so the distribution is concentrated and the standard expectation value is used instead.

\subsection{Quantum Approximate Optimization Algorithm}
\label{subsec:qaoa}

QAOA \cite{farhi2014quantum} takes the equal superposition $\vert+\rangle^{\otimes n}$ as its initial state and applies $p$ alternating layers of cost and mixer unitaries:
\begin{equation}
  |\psi(\boldsymbol{\gamma},\boldsymbol{\beta})\rangle
  = \prod_{\ell=1}^{p} e^{-i\beta_\ell H_M} e^{-i\gamma_\ell H_C}\,|+\rangle^{\otimes n},
  \label{eq:qaoa}
\end{equation}
where $H_C = H_{\mathrm{MIS}}$ is the cost Hamiltonian and $H_M = \sum_i X_i$ is the standard transverse-field mixer. We use the standard, unconstrained transverse-field mixer rather than a feasibility-preserving mixer from the Quantum Alternating
Operator Ansatz framework~\cite{hadfield2019quantum}. Constrained mixers for MIS in this framework typically use multi-controlled gates over each vertex's full neighbourhood to restrict mixer moves to the feasible subspace~\cite{tomesh2024tradeoffs}. For
dense graphs such as the 99- and 180-node instances studied here (average degree 56 and 59 respectively), the resulting gate count
per mixer layer scales as $O(|E|)$ rather than $O(n)$, and the wide-control multi-qubit gates induce entanglement far exceeding
the MPS-tractable regime characterised in Section~\ref{subsec:theory_reorder}. We therefore use the standard unconstrained mixer as the QAOA baseline; the excitation-preserving ansatz of Section~\ref{subsec:ancilla} separately achieves a weaker but computationally tractable form of feasibility bias (Hamming-weight conservation rather than exact edge-level feasibility) using only nearest-neighbour two-qubit gates.

Each cost layer implements $e^{-i\gamma H_C}$ as a sequence of $\text{RZZ}(2\gamma J_{ij})$ two-qubit gates for each coupled pair and $\text{RZ}(2\gamma h_i)$ single-qubit rotations for each vertex field; the mixer layer applies $\text{RX}(2\beta)$ to every qubit. We use the same sparsified coupling set $J^{(k)}_\mathrm{sparse}$ for circuit construction as in the VQE case. Gates are sorted in ascending order of index distance $|i-j|$ within each cost layer to maximally exploit parallelism on linear-connectivity hardware.

\subsection{Bitstring Heuristics}
\label{subsec:heuristics}

The measurement distribution produced by the optimized variational circuit concentrates probability weight on low-energy Ising configurations, but raw quantum outputs alone are insufficient to identify a maximum independent set: bitstrings may still violate the independence constraint, or represent feasible but non-maximal independent sets whose size falls short of the optimum. The post-processing pipeline described in this section is one of the most important components of the overall framework. It converts the statistical output of a finite-shot quantum experiment into a rich population of high-quality independent sets by combining a history-guided suppression mechanism with greedy correction and stepwise extension to maximality.

Classical post-processing of quantum outputs is well-established: multi-start local search improves QAOA solutions \cite{shaydulin2019multistart}, and greedy repair addresses infeasibility \cite{bravyi2020obstacles,shaydulin2019multistart}. Our approach differs in three ways: (i) EMA-based history accumulation across sampling rounds rather than independent restarts, (ii) two-sided suppression masks addressing both over- and under-selection, and (iii) a maximality augmentation step that none of the prior works include. A related technique in quantum chemistry \cite{robledo2025chemistry} uses per-qubit statistics to restore particle-number symmetry in noisy bitstrings; we apply the same principle to combinatorial solution quality in noiseless simulation.

The pipeline operates in two phases. During variational optimization, $n_{\mathrm{shots}}$ measurement shots per iteration (1000--4000 in our experiments) are used to estimate the energy expectation value. After convergence, the optimized circuit parameters are fixed and the circuit is sampled repeatedly in successive blocks. After each block, the top-$M$ bitstrings are extracted, the suppression weights are
updated, and the greedy correction pipeline is applied. The specific
sampling protocol is described in Section~\ref{subsubsec:modes} (Mode~B).  This separation between optimization shots and post-optimization sampling shots is deliberate. The optimization requires only enough shots to estimate gradient directions reliably, while the post-processing benefits from a much larger and more diverse sample.

\subsubsection{History-Guided Suppression Weights}
\label{subsubsec:ema}

The heuristic maintains a \emph{suppression weight} $w_i \in [0,1]$ for each qubit $i$, initialized to $0.5$ (no prior information). The weight $w_i$ tracks how frequently qubit $i$ is in the $|0\rangle$ state (not selected) among the best observed bitstrings: a high value of $w_i$ means qubit $i$ is almost always absent from good independent sets and is therefore a candidate for suppression.

After each batch of measurements, the top-$M$ bitstrings are selected according to a quality score that rewards large independent set size and penalizes violated independence constraints. For each qubit $i$, the \emph{zero-frequency} over the top-$M$ pool $T$ is
\begin{equation}
  z_i^{(t)} = \frac{1}{|T|}\sum_{s \in T} \mathbf{1}[s_i = 0],
  \label{eq:zfreq}
\end{equation}
where $s_i$ denotes the value of bit $i$ in bitstring $s$ (in qubit-index space, not raw string position; see the endianness remark below). The suppression weight is then updated via an exponential moving average (EMA) \cite{hunter1986exponentially}:
\begin{equation}
  w_i \;\leftarrow\; \alpha\, z_i^{(t)} + (1-\alpha)\, w_i,
  \label{eq:ema}
\end{equation}
where $\alpha \in (0,1)$ controls the adaptation rate. The weights are interpreted through a two-sided threshold scheme with an upper
limit $\theta^+$ and a lower limit $\theta^-$, with $0 < \theta^- < \theta^+ < 1$
(typical values $\theta^- = 0.25$, $\theta^+ = 0.75$). Qubits with $w_i > \theta^+$
are placed in the \emph{upper suppression mask} $\mathcal{M}^+$: they are almost
always absent from good solutions and are candidates for removal from a sampled
set. Qubits with $w_i < \theta^-$ are placed in the \emph{lower suppression mask}
$\mathcal{M}^-$: they are almost always present in good solutions and are candidates
for addition to the set, provided doing so does not violate independence. Qubits
with $\theta^- \leq w_i \leq \theta^+$ are in an uncertain regime and are left
unchanged. This two-sided scheme addresses both failure modes of the sampler:
over-selection (too many $|1\rangle$ states causing independence violations) and
under-selection (too few $|1\rangle$ states leaving the set smaller than it could be).
A convergence analysis of the EMA weights to the stationary measurement distribution
is given in Sec.~S2 of the Supplemental Material.

\medskip
\noindent\textbf{Remark on the Qiskit bitstring convention.}
Qiskit returns measurement results in \emph{little-endian} order, where the rightmost character of the bitstring corresponds to qubit 0 and the leftmost to qubit $n-1$. All heuristic logic operates in qubit-index space. The conversion between the raw Qiskit string and qubit-index representation is handled exclusively at the boundary by helper functions.

\subsubsection{Sampling Modes}
\label{subsubsec:modes}

The history-guided heuristic can be deployed in two distinct operational modes, depending on how the QAOA or VQE circuit interacts with the correction loop.

\textbf{Mode A (warm-started iterative optimization).} The full optimize, sample, and correct cycle is repeated multiple times. After each round, the converged parameters $\gamma^*, \beta^*$ are carried forward as the starting point for the next round's optimizer (a warm start~\cite{egger2021warm,tate2023warm}), rather than restarting from a random initial point. The suppression weights $w_i$ are updated from the sampled bitstrings of the current round and inform the correction applied before the next round begins. Corrected bitstrings are \emph{not} fed back into the quantum circuit; only the classical parameters are transferred. Figure~\ref{fig:modeA} illustrates this pipeline.

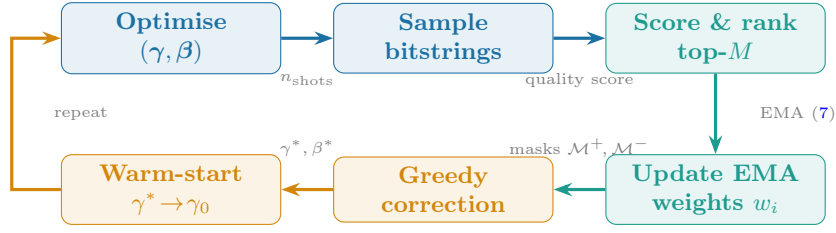
\begin{figure}[!htbp]
\centering
\begin{tikzpicture}[
  blk/.style={rectangle, rounded corners=4pt,
              draw=#1, fill=#1!10, text=#1,
              minimum width=2.9cm, minimum height=0.8cm,
              font=\footnotesize\bfseries, align=center},
  arr/.style={-{Stealth[length=7pt,width=5pt]}, line width=1.2pt, #1},
  lbl/.style={font=\tiny, gray},
]
\node[blk=accentblue]  (A1) at (0,   0)    {Optimise\\$(\boldsymbol{\gamma},\boldsymbol{\beta})$};
\node[blk=accentblue]  (A2) at (3.6, 0)    {Sample\\bitstrings};
\node[blk=accentteal]  (A3) at (7.2, 0)    {Score \& rank\\top-$M$};
\node[blk=accentteal]  (A4) at (7.2,-2.0)  {Update EMA\\weights $w_i$};
\node[blk=accentamber] (A5) at (3.6,-2.0)  {Greedy\\correction};
\node[blk=accentamber] (A6) at (0,  -2.0)  {Warm-start\\$\gamma^*\!\to\!\gamma_0$};
\draw[arr=accentblue] (A1.east)--(A2.west);
\node[lbl] at (1.8,-0.55) {$n_{\mathrm{shots}}$};
\draw[arr=accentblue] (A2.east)--(A3.west);
\node[lbl] at (5.4,-0.55) {quality score};
\draw[arr=accentteal] (A3.south)--(A4.north);
\node[lbl,anchor=west] at (7.65,-1.0) {EMA~\eqref{eq:ema}};
\draw[arr=accentteal] (A4.west)--(A5.east);
\node[lbl] at (5.4,-1.45) {masks $\mathcal{M}^+\!,\mathcal{M}^-$};
\draw[arr=accentamber] (A5.west)--(A6.east);
\node[lbl] at (1.8,-1.45) {$\gamma^*,\beta^*$};
\draw[arr=accentamber] (A6.west)--++(-0.65,0) |- (A1.west);
\node[lbl,anchor=east] at (-0.7,-1.0) {repeat};
\end{tikzpicture}
\caption{Mode A: warm-started iterative optimization. After each round the
  converged parameters initialize the next optimizer call.
  Only classical parameters are transferred between rounds; corrected
  bitstrings are not fed back into the quantum circuit.
  The parameter labels $\gamma^*, \beta^*$ follow QAOA notation; for VQE
  with the EfficientSU2 ansatz the same diagram applies with $\boldsymbol{\theta}^*$
  denoting the rotation angles $\{R_Y(\theta), R_X(\phi)\}$ of each layer.}
\label{fig:modeA}
\end{figure}

\textbf{Mode B (single optimization, block-wise sampling).} The variational circuit is optimized \emph{once} to convergence. The converged parameters $\gamma^*, \beta^*$ are then held fixed, and the circuit is sampled repeatedly in successive \emph{blocks} of $B$ shots each. After each block, the top-$M$ bitstrings are extracted, the zero-frequency statistics (\ref{eq:zfreq}) are computed, and the suppression weights are updated via the EMA (\ref{eq:ema}). The greedy correction is applied to all bitstrings in the block using the current weights. This continues until the weights stabilize or a fixed number of blocks has been processed. Figure~\ref{fig:modeB} illustrates this pipeline.

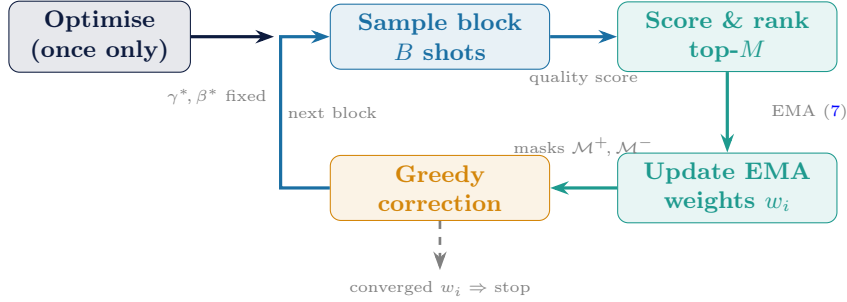
\begin{figure}[!htbp]
\centering
\begin{tikzpicture}[
  blk/.style={rectangle, rounded corners=4pt,
              draw=#1, fill=#1!10, text=#1,
              minimum width=2.9cm, minimum height=0.8cm,
              font=\footnotesize\bfseries, align=center},
  arr/.style={-{Stealth[length=7pt,width=5pt]}, line width=1.2pt, #1},
  darr/.style={-{Stealth[length=7pt,width=5pt]}, line width=1.0pt, #1, dashed},
  lbl/.style={font=\tiny, gray},
]
\node[blk=navyblue, minimum width=2.4cm]
      (B0) at (0,0) {Optimise\\(once only)};
\draw[arr=navyblue] (B0.east)--++(1.1,0);
\node[lbl,anchor=north] at (1.55,-0.55) {$\gamma^*\!,\beta^*$ fixed};
\node[blk=accentblue]  (B1) at (4.5, 0)    {Sample block\\$B$ shots};
\node[blk=accentteal]  (B2) at (8.3, 0)    {Score \& rank\\top-$M$};
\node[blk=accentteal]  (B3) at (8.3,-2.0)  {Update EMA\\weights $w_i$};
\node[blk=accentamber] (B4) at (4.5,-2.0)  {Greedy\\correction};
\draw[arr=accentblue] (B1.east)--(B2.west);
\node[lbl] at (6.4,-0.55) {quality score};
\draw[arr=accentteal] (B2.south)--(B3.north);
\node[lbl,anchor=west] at (8.75,-1.0) {EMA~\eqref{eq:ema}};
\draw[arr=accentteal] (B3.west)--(B4.east);
\node[lbl] at (6.4,-1.45) {masks $\mathcal{M}^+\!,\mathcal{M}^-$};
\draw[arr=accentblue] (B4.west)--++(-0.65,0) |- (B1.west);
\node[lbl,anchor=east] at (3.8,-1.0) {next block};
\draw[darr=gray] (B4.south)--++(0,-0.65)
      node[below,lbl,align=center]{converged $w_i \Rightarrow$ stop};
\end{tikzpicture}
\caption{Mode B: single optimization with block-wise sampling. The circuit is optimized once and parameters are fixed. Successive blocks of $B$ shots are drawn; after each block the EMA suppression weights are updated and greedy correction is applied. No re-optimization takes place.
As with Mode A (Figure~\ref{fig:modeA}), the parameter notation $\gamma^*, \beta^*$ follows QAOA convention; for VQE with EfficientSU2, the parameters are the layer rotation angles $\boldsymbol{\theta}^*$, and the same pipeline applies without modification.}
\label{fig:modeB}
\end{figure}

Mode A is preferable when the quantum runtime budget allows repeated circuit optimization, since the parameters adapt to the evolving suppression history. Mode B is more suitable when quantum runtime is constrained: a single optimization pass is performed, and the classical correction loop extracts as much information as possible from repeated sampling at fixed parameters. In the experiments reported here, we use Mode B for all three benchmark instances, with the $50000$-shot post-optimization sample divided into blocks.

\subsubsection{Greedy Bitstring Correction}
\label{subsubsec:greedy}

Once the suppression masks have been formed, each sampled bitstring is
corrected in two complementary passes. Greedy repair of infeasible solutions
is widely used in combinatorial optimization \cite{festa2002randomized} and
has been applied to quantum bitstring decoding in several recent works
\cite{shaydulin2019multistart,festa2002randomized}. The key distinction of
our approach is that the two-sided mask structure addresses both over- and
under-selection, and the correction order is determined by the learned weights
$w_i$ rather than by a fixed rule such as degree or random order.

The selected qubit set for a bitstring $s$ is $S = \{i : s_i = 1\}$. In
\textbf{Pass~1} (upper mask), qubits in $\mathcal{M}^+ \cap S$ that are
currently selected are sorted in descending order of $w_i$ and removed from
$S$ one at a time, stopping as soon as $S$ is a valid independent set; qubits
with high $w_i$ that do not cause any violation are left in place, preserving
set size. In \textbf{Pass~2} (lower mask), qubits in $\mathcal{M}^- \setminus S$
that are not yet selected are sorted in ascending order of $w_i$ and
tentatively added to $S$, but only if doing so creates no edge between the
new qubit and the current members of $S$; this grows the independent set
toward what the historical evidence suggests is correct.

\begin{algorithm}[!htbp]
\caption{Two-threshold greedy bitstring correction}
\label{alg:greedy}
\begin{algorithmic}[1]
\Require Selected set $S$, weights $\{w_i\}$, adjacency $\{N(i)\}$,
         thresholds $\theta^- < \theta^+$
\Ensure $S$ is an independent set, aligned with historical evidence
\State \textit{// Pass 1: remove over-selected qubits}
\State $\mathcal{M}^+ \gets \{i : w_i > \theta^+\}$
\State $\mathrm{order}^+ \gets \mathrm{sort}(\mathcal{M}^+ \cap S,\;
       \text{key} = w_i,\; \text{descending})$
\For{each $i$ in $\mathrm{order}^+$}
    \If{$S$ is independent} \textbf{break} \EndIf
    \State $S \gets S \setminus \{i\}$
\EndFor
\State \textit{// Pass 2: add under-selected qubits}
\State $\mathcal{M}^- \gets \{i : w_i < \theta^-\}$
\State $\mathrm{order}^- \gets \mathrm{sort}(\mathcal{M}^- \setminus S,\;
       \text{key} = w_i,\; \text{ascending})$
\For{each $i$ in $\mathrm{order}^-$}
    \If{$N(i) \cap S = \emptyset$}
        \State $S \gets S \cup \{i\}$
    \EndIf
\EndFor
\State \Return $S$
\end{algorithmic}
\end{algorithm}

Pass~1 ensures feasibility by removing the most historically absent qubits
until no independence constraint is violated. Pass~2 then grows the set by
including qubits that are historically almost always selected, provided they
are compatible with the current set. Together the two passes produce a
feasible independent set that is more consistent with the accumulated evidence
than the raw quantum output.

\subsubsection{Stepwise Extension to a Maximal Independent Set}
\label{subsubsec:stepwise}

The corrected bitstring is feasible but not necessarily \emph{maximal}. We extend it by repeatedly adding vertices $v$ with $N(v) \cap S = \emptyset$ until no such vertex exists. The selection rule is minimum degree \cite{halldorsson1997greedy}, which tends to produce larger maximal sets; a formal lower bound on the resulting set size is given in Sec.~S3 of the Supplemental Material. The complete pipeline of greedy correction followed by stepwise extension recovers 2--3 additional vertices over raw quantum outputs across all benchmark instances.
At each step, the \emph{candidate set} of vertices compatible with the current independent set $S$ is
\begin{equation}
  \mathcal{C}(S) = \bigl\{v \notin S \;\big|\; N(v) \cap S = \emptyset\bigr\},
  \label{eq:candidates}
\end{equation}
where $N(v)$ is the neighborhood of $v$. A single vertex is selected from $\mathcal{C}(S)$ and added to $S$; the process repeats until $\mathcal{C}(S) = \emptyset$, at which point $S$ is maximal.

\begin{algorithm}[H]
\caption{Stepwise extension to a maximal independent set}
\label{alg:stepwise}
\begin{algorithmic}[1]
\Require Feasible independent set $S$, adjacency $\{N(v)\}$,
         selection rule \textsc{Select}
\Ensure Maximal independent set $S$
\While{$\mathcal{C}(S) \neq \emptyset$}
    \State $v^* \gets \textsc{Select}\bigl(\mathcal{C}(S), S\bigr)$
    \State $S \gets S \cup \{v^*\}$
\EndWhile
\State \Return $S$
\end{algorithmic}
\end{algorithm}

\subsection{Excitation-Preserving Ansatz with Ancilla Superposition}
\label{subsec:ancilla}

For the 180-node instance, the combined effect of standard VQE optimization and bitstring
post-processing yields multiple independent sets of size 14 but cannot break through to size
15, the known MIS cardinality. The fundamental difficulty as we see is that configurations of
size 14 and size 15 differ by a single vertex, and the energy gap between them is small
relative to the landscape features accessible to a general-purpose variational ansatz
initialized from the zero state.

The technique we introduce exploits the quantum superposition principle to address this
limitation directly. Suppose the classical post-processing has identified $m$ independent
sets $s^{(1)}, \ldots, s^{(m)}$ of size $k$ (in our case $k = 14$, $m = 4$). Instead of
initializing the variational circuit at any single one of these solutions, we prepare the
uniform superposition
\begin{equation}
  |\Psi_0\rangle = \frac{1}{\sqrt{m}} \sum_{j=1}^{m} |s^{(j)}\rangle,
  \label{eq:superpos}
\end{equation}
where each $|s^{(j)}\rangle$ is the computational basis state encoding solution $s^{(j)}$.

\textbf{Why ancilla qubits are needed.}
There are practical difficulties in preparing the superposition (\ref{eq:superpos}) directly
on the $n$ data qubits. The primary difficulty is that preparing an arbitrary superposition of $m$
computational basis states requires a mechanism to ``address'' each state individually.
A Hadamard gate on all qubits produces an equal superposition over \emph{all} $2^n$ basis
states, not just the $m$ desired seeds. Ancilla qubits solve this problem by acting as an
\emph{index register}: each ancilla configuration $|j\rangle$ labels one seed state, and
controlled operations write the corresponding seed onto the data register only when the
ancilla holds that label. The result is a superposition over exactly the $m$ seeds, with
the ancilla register serving as a ``which-seed'' tag entangled with the data.

\textbf{Construction overview.}
The circuit uses $a = \lceil\log_2 m\rceil$ ancilla qubits (2 ancillas for $m=4$ seeds). The preparation proceeds in three conceptual steps:
\begin{enumerate}[noitemsep]
  \item \textbf{Create ancilla superposition:} Apply Hadamard gates to all $a$ ancilla qubits, producing the uniform superposition $(1/\sqrt{m})\sum_{j=1}^{m}|j\rangle_{\mathrm{anc}}$ over all index values.
  \item \textbf{Controlled seed loading:} For each seed index $j$, apply multi-controlled X gates (generalized Toffoli gates) that flip the data qubits corresponding to the vertices in $s^{(j)}$, \emph{conditioned on} the ancilla register being in state $|j\rangle$. This writes seed $s^{(j)}$ onto the data register only in the branch where the ancilla holds index $j$.
  \item \textbf{Final entangled state:} After all seeds are loaded, the joint state is $(1/\sqrt{m})\sum_{j=1}^{m}|s^{(j)}\rangle_{\mathrm{data}} \otimes |j\rangle_{\mathrm{anc}}$. The ancilla and data registers are now entangled: measuring the ancilla would collapse the data register to the corresponding seed, but we do not measure the ancilla during optimization.
\end{enumerate}

\textbf{Implementing the controlled loading.}
To condition on ancilla state $|j\rangle$ with binary representation $j = b_{a-1}\cdots b_1 b_0$, we temporarily flip ancilla qubits where $b_i = 0$ (using X gates), apply the multi-controlled X gates with all-ones control, then unflip the same ancilla qubits. This ``control pattern conversion'' ensures each seed is written only when the ancilla holds its designated index. After all $m$ seeds have been loaded, the ancilla register is left entangled with the data register. During variational optimization, measurements are taken only on the data qubits; the ancilla register is either traced out (ignored) or measured separately for diagnostics. The effective state on the data register alone is the desired superposition $|\Psi_0\rangle$.

Figure~\ref{fig:ancilla_circuit} shows a small illustrative instance of this initialization circuit for $n = 5$ data qubits, $m = 4$ seed states, and $a = 2$ ancilla qubits. The Hadamard gates on qubits $q_5$ and $q_6$ create the equal superposition over the four ancilla basis states; each subsequent barrier-delimited block applies the multi-controlled X gates that load one of the four data basis states.

\begin{figure}[!htbp]
  \centering
  \includegraphics[width=\linewidth]{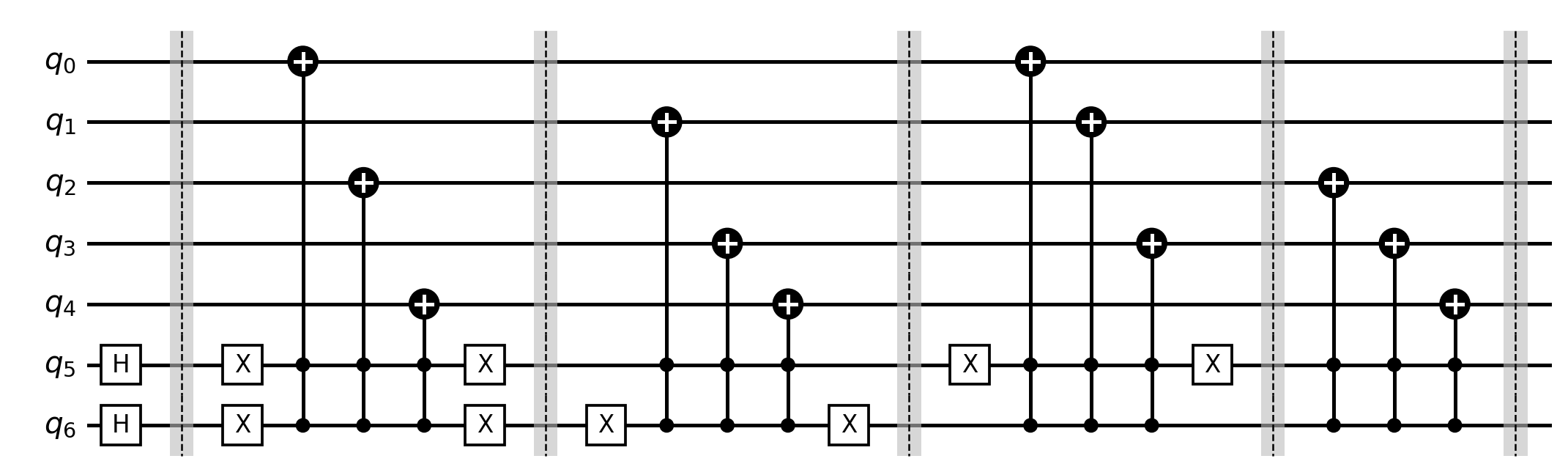}
  \caption{Ancilla superposition initialization circuit for a 5-qubit data register ($q_0$--$q_4$) with $m = 4$ seed basis states encoded using $a = 2$ ancilla qubits ($q_5$, $q_6$). Hadamard gates on the ancilla register create a uniform superposition over all four computational basis states. Each of the four barrier-delimited blocks then coherently writes one seed state onto the data register via multi-controlled X (Toffoli) gates conditioned on the corresponding ancilla state. For the 180-node instance, $q_0$--$q_{179}$ serve as the data register and two ancilla qubits encode the four near-optimal size-14 seed solutions.}
  \label{fig:ancilla_circuit}
\end{figure}

After initialization, the data register which is now in the superposition state $|\Psi_0\rangle$, is evolved under an \emph{excitation-preserving} variational ansatz \cite{gard2020efficient}. For the 180-node instance with $p$ layers, the ansatz applies $p(n-1)$ parametrized $fSim$ gates in a linear nearest-neighbour topology. The ancilla preparation circuit contributes an additional depth of $O(m \cdot k)$ multi-controlled X gates, where $m=4$ is the number of seeds and $k=14$ is the seed cardinality; in practice this amounts to approximately 60 additional two-qubit gate layers after decomposition into native gates. This circuit family is built from parametrized $fSim$ (fermionic simulation) gates of the form
\begin{equation}
  fSim(\theta, \varphi) = \begin{pmatrix} 1 & 0 & 0 & 0 \\ 0 & \cos\theta & i\sin\theta & 0 \\ 0 & i\sin\theta & \cos\theta & 0 \\ 0 & 0 & 0 & e^{-i\varphi} \end{pmatrix},
  \label{eq:fsim}
\end{equation}
which swap excitations between two qubits while preserving the total excitation count.
The gate has two parameters: $\theta$ controls the amplitude of the excitation hop
between the $|01\rangle$ and $|10\rangle$ subspace (the iSWAP-like term), while
$\varphi$ is a \emph{conditional phase} applied to the $|11\rangle$ component, representing
a ZZ-type interaction between the two qubits when both are simultaneously excited.

This interaction term is the crucial element that separates the $fSim$ gate from a simple
excitation-hopping gate. For example, at $\theta = \pi/2$, $\varphi = 0$, the gate
reduces to iSWAP, which is Clifford and efficiently simulable. For generic $(\theta, \varphi)$ with
$\varphi \neq 0$, the gate is non-Clifford and generates entanglement patterns that grow
with circuit depth in a way that quickly exceeds the exact MPS bond dimension available
here.

By Proposition~\ref{prop:bond_dim}, the MPS bond dimension of the $p$-layer
circuit satisfies $\chi_{\max} \leq m \cdot 2^p$, growing exponentially in the number
of layers. For the $m=4$, $p=2$ configuration used here, $\chi_{\max} = 16$, placing
the circuit within the exact MPS-simulable regime. Beyond $p \geq 4$ (where $\chi_{\max} \geq 64$),
exact simulation via MPS requires exponentially growing memory, and approximate simulation via
truncated MPS introduces errors that compromise the optimization landscape.
Hardware execution becomes preferable at that point, since the quantum device applies the full
$fSim(\theta,\varphi)$ unitary without MPS truncation.

Because all seed states have exactly $k = 14$ qubits in the
$|1\rangle$ state (representing inclusion in the independent set), and the $fSim$ gates
conserve Hamming weight, the variational state remains in the $\binom{n}{k}$-dimensional
subspace of $n$-qubit states with exactly $k$ ones throughout the entire optimization.
This has two important consequences. First, the circuit automatically respects the cardinality
constraint; no penalty term is needed to prevent the ansatz from exploring configurations
of the wrong size. Second, because the landscape is restricted to a much smaller
subspace ($\binom{180}{14} \approx 5 \times 10^{18}$ rather than $2^{180}$), the
optimizer can make finer distinctions between near-optimal solutions.

The seed states $s^{(1)}, \ldots, s^{(4)}$ are distinct size-14 independent sets that
cover different parts of the 180-node graph. Their superposition $|\Psi_0\rangle$
encodes all four seeds in a single pure quantum state. The excitation-preserving
unitary $U(\boldsymbol{\theta})$ acts on the data register alone, evolving all four
branches simultaneously under the same variational parameters. The energy optimized
by SPSA is $\langle H\rangle = \frac{1}{m}\sum_{j=1}^m\langle s^{(j)}|U^\dagger H U|s^{(j)}\rangle$,
the average energy across all $m$ seeds; the cross terms vanish because the ancilla
basis states are orthogonal. The optimizer therefore finds parameters that are optimal
across the full seed population in a single circuit execution,
rather than optimal for any individual seed. This parallel search over
multiple carefully chosen initializations is not available to a single-seed
initialization.

The full pipeline for the 180-node instance consists of four phases. In the first phase, standard VQE with the EfficientSU(2) ansatz ($p=2$, SPSA optimizer, CVaR $\alpha \in [0.1,0.2]$, 2000--4000 shots per evaluation, 4000--6000 iterations) is run from the zero initial state. Bitstring post-processing applied to the resulting measurement distribution yields multiple independent sets of size 14, from which a set of diverse solutions are selected as seeds. The diversity criterion is based on Jaccard distance: starting from the lowest-energy size-14 solution, each subsequent seed is chosen to maximize the minimum Jaccard distance to all previously selected seeds, ensuring that the seed population covers distinct regions of the solution space rather than clustering around a single local optimum. The four selected seeds have pairwise Jaccard distances ranging from approximately 0.88 to 1.0 (where 1.0 indicates disjoint vertex sets), confirming broad coverage:
\begin{align*}
  s^{(1)} &= \{2,22,32,38,42,51,95,100,113,121,134,154,163,177\}, \\
  s^{(2)} &= \{8,13,30,48,56,74,95,100,112,122,126,140,146,178\}, \\
  s^{(3)} &= \{17,24,42,89,95,101,104,111,132,140,149,154,168,178\}, \\
  s^{(4)} &= \{20,35,36,46,49,73,93,111,115,135,148,167,173,179\}.
\end{align*}
In the second phase, the ancilla superposition circuit loads these four seed states using two ancilla qubits. In the third phase, VQE with the excitation-preserving ansatz ($p=2$, SPSA, 1000--2000 shots, 2000--4000 iterations) is run from this prepared initial state over the combined $(180 + 2)$-qubit register. In the fourth phase, bitstring post-processing is applied to the data-register measurement distribution. The four phases are illustrated in Figure~\ref{fig:ancilla_sup}.
\begin{figure}[!htbp]
  \centering
  \includegraphics[width=\linewidth]{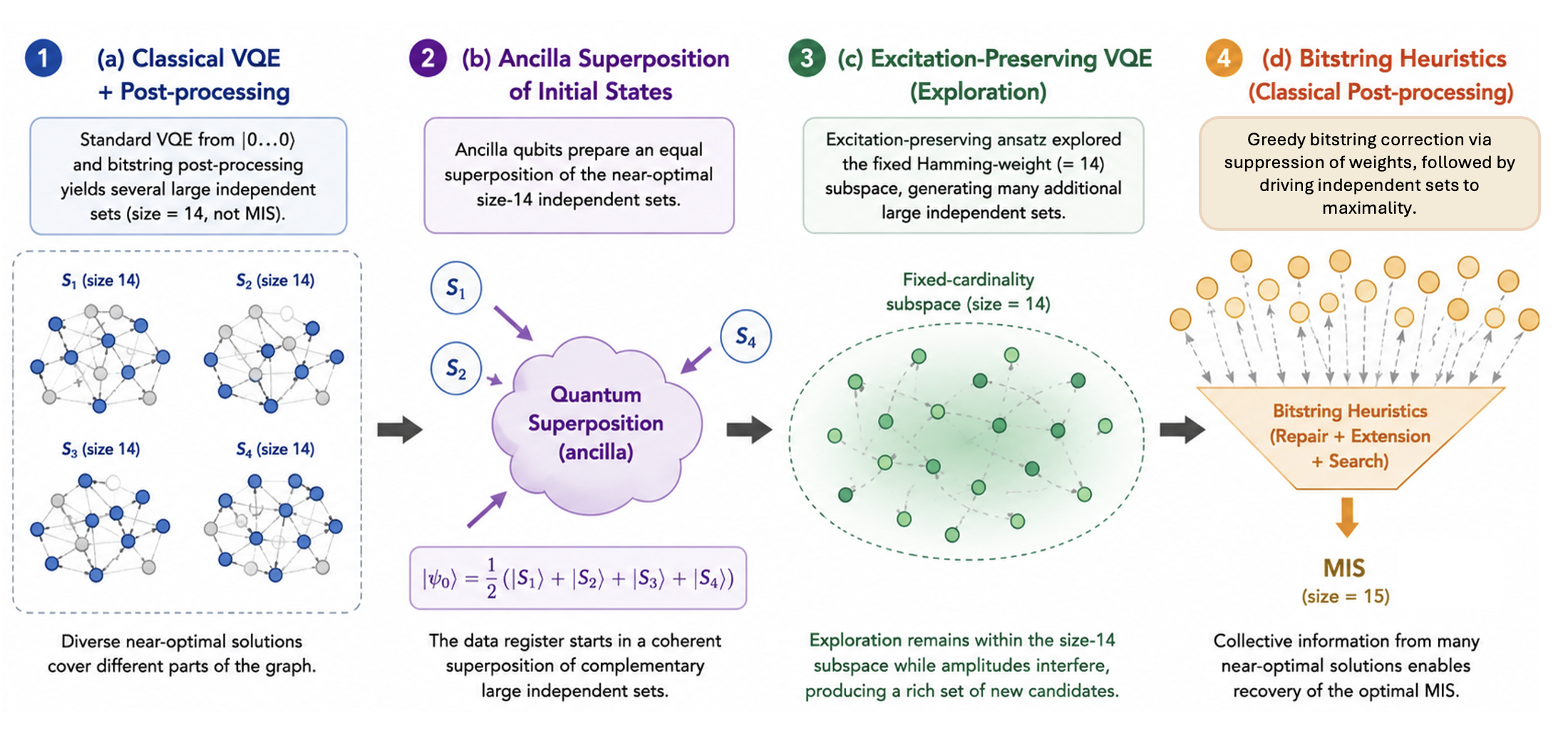}
  \caption{Schematic overview of the four-phase pipeline for the 180-node instance.
  \textbf{(a) Classical VQE and post-processing:} Standard VQE from $|0\ldots 0\rangle$ yields several large independent sets of size 14 (not the MIS). The four graph panels show the four seed solutions $S_1$--$S_4$, each covering different parts of the 180-node graph.
  \textbf{(b) Ancilla superposition of initial states:} Ancilla qubits prepare the data register in the uniform superposition $|\psi_0\rangle = \frac{1}{\sqrt{m}}\sum_{j=1}^m|s^{(j)}\rangle$.
  \textbf{(c) Excitation-preserving VQE (exploration):} The excitation-preserving ansatz evolves $|\psi_0\rangle$ within the fixed Hamming-weight ($=14$) subspace, generating a rich population of new candidate independent sets through parallel variational optimization over all four seed branches under shared parameters.
  \textbf{(d) Bitstring heuristics (classical post-processing):} Greedy correction via suppression weights, followed by stepwise extension to maximality, distils the measurement distribution into the MIS of size 15. Multiple near-optimal solutions collectively contain information about the optimum; parallel search over four seed states uncovers more such candidates, and classical heuristics assemble that information into the MIS.
  }
  \label{fig:ancilla_sup}
\end{figure}

As shown in Proposition~\ref{prop:bond_dim}, the bond dimension of the excitation-preserving
circuit obeys $\chi_{\max} \leq m \cdot 2^p$. For the Phase~3 configuration used
here ($p=2$, $m=4$), this gives $\chi_{\max} = 16$, confirmed numerically with no
truncation at threshold $\varepsilon = 10^{-16}$. The state is therefore efficiently
simulable classically. The ancilla superposition technique enforces cardinality constraints
without penalty terms, finds parameters across all seeds in one run, and scales
with ancilla count. The technique reaches solutions that the bitstring heuristic pipeline
does not reach from any single seed, and provides a foundation for hardware implementations
where deeper ansätze ($p \geq 4$, $\chi_{\max} \geq 64$) lie beyond the exact MPS simulations we can run.

\textbf{Classical tractability note.}
The ancilla superposition initialization itself, which prepares
$(1/\sqrt{m})\sum_j|s^{(j)}\rangle|j\rangle_\mathrm{anc}$ via Hadamard plus
multi-controlled X gates and then measures only the data register, is
\emph{classically tractable}. The effective state on the data register is the mixed
state $\rho = \frac{1}{m}\sum_j |s^{(j)}\rangle\langle s^{(j)}|$, which is a
classical stochastic mixture that could in principle be replicated by running $m$
independent single-seed circuits and pooling results. We make no claim of quantum
advantage for this component. The $fSim(\theta, \varphi)$ ansatz with $\varphi \neq 0$
is the real component that introduces quantum correlations by entangling seed branches
(the measured bond dimension shows this), and deeper circuits
($p \geq 4$, $\chi_{\max} \geq 64$) exceed the exact MPS simulations we were able to run on
the machine described in Section~\ref{sec:setup}.

\subsection{Ancilla Interference}
\label{subsec:interference}

The ancilla superposition construction of Section~\ref{subsec:ancilla} evolves $K$ seed branches simultaneously under shared parameters $\boldsymbol{\theta}$ (with $K = m$, the number of seeds from Section~\ref{subsec:ancilla}). When the data register alone is measured, the result is a mixture of the $K$ individual seed distributions, with no cross-seed correlations. The \emph{ancilla interference} technique makes these cross-seed correlations physically observable by projecting the ancilla register onto a specific superposition state after a final Hadamard layer.

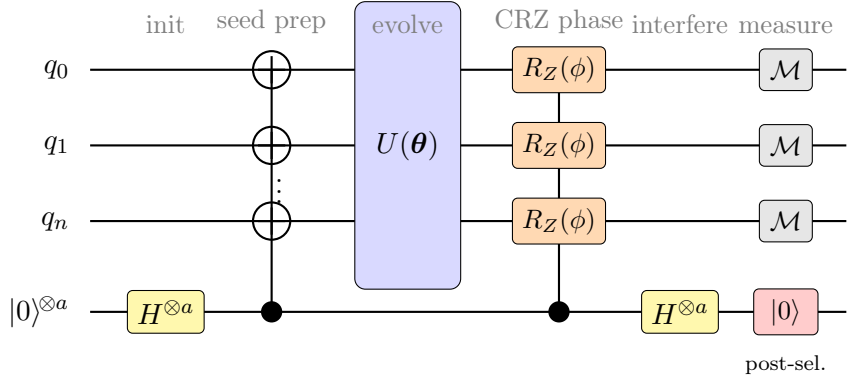
\begin{figure}[!htbp]
\centering
\begin{tikzpicture}[scale=1.0, every node/.style={font=\normalsize}]
  %------------------------------------------------------------
  % Wire y-positions
  %------------------------------------------------------------
  \def\yA{3.2}   % data qubit q_0
  \def\yB{2.2}   % data qubit q_1
  \def\yC{1.2}   % data qubit q_n  (dots between yB and yC)
  \def\yAnc{0.0} % ancilla register

  %------------------------------------------------------------
  % Column x-positions
  %------------------------------------------------------------
  \def\xS{0.0}    % start / labels
  \def\xH{1.0}    % H on ancilla (init superposition)
  \def\xMCX{2.4}  % MCX seed-prep column
  \def\xU{4.2}    % U(theta) evolve
  \def\xCRZ{6.2}  % CRZ phase layer
  \def\xHf{7.8}   % H on ancilla (final)
  \def\xM{9.2}    % measurements
  \def\xEnd{10.0}

  %------------------------------------------------------------
  % Horizontal wires
  %------------------------------------------------------------
  \foreach \y in {\yA,\yB,\yC} {
    \draw[thick] (\xS,\y) -- (\xEnd,\y);
  }
  % ancilla double-wire
  \draw[thick] (\xS,\yAnc) -- (\xEnd,\yAnc);

  % Wire labels
  \node[left=4pt] at (\xS,\yA)   {$q_0$};
  \node[left=4pt] at (\xS,\yB)   {$q_1$};
  \node[left=4pt] at (\xS,\yC)   {$q_n$};
  \node[left=4pt] at (\xS,\yAnc) {$|0\rangle^{\!\otimes a}$};
  % vertical dots between q_1 and q_n
  \pgfmathsetmacro{\xMid}{\xS+0.5*(\xEnd-\xS)/2}
  \pgfmathsetmacro{\yMid}{0.5*\yB+0.5*\yC}
  \node at (\xMid,\yMid) {$\vdots$};

  %------------------------------------------------------------
  % STAGE 1: H on ancilla
  %------------------------------------------------------------
  \node[draw,fill=yellow!40,minimum width=0.7cm,minimum height=0.55cm,
        rounded corners=2pt] at (\xH,\yAnc) {$H^{\otimes a}$};
  \node[above=10pt,gray,font=\small] at (\xH,\yA) {init};

  %------------------------------------------------------------
  % STAGE 2: MCX seed-prep
  % Ancilla = control (filled dot), data qubits = targets (X circle)
  %------------------------------------------------------------
  % vertical control line from ancilla to top data qubit
  \draw[thick] (\xMCX,\yAnc) -- (\xMCX,\yA);
  % control dot on ancilla wire
  \fill (\xMCX,\yAnc) circle (4pt);
  % X-target circles on each data wire
  \foreach \y in {\yA,\yB,\yC} {
    \draw[thick] (\xMCX,\y) circle (7pt);
    \draw[thick] (\xMCX-0.18,\y) -- (\xMCX+0.18,\y);
    \draw[thick] (\xMCX,\y-0.18) -- (\xMCX,\y+0.18);
  }
  \node[above=10pt,gray,font=\small] at (\xMCX,\yA) {seed prep};

  %------------------------------------------------------------
  % STAGE 3: U(theta) — large box spanning data register
  %------------------------------------------------------------
  \node[draw,fill=blue!15,minimum width=1.4cm,minimum height=3.8cm,
        rounded corners=4pt,font=\normalsize]
    at (\xU, 0.5*\yA+0.5*\yC) {$U(\boldsymbol{\theta})$};
  \node[above=10pt,gray,font=\small] at (\xU,\yA) {evolve};

  %------------------------------------------------------------
  % STAGE 4: CRZ — control on ancilla, RZ target on data qubits
  %------------------------------------------------------------
  % vertical line from ancilla up to top data qubit
  \draw[thick] (\xCRZ,\yAnc) -- (\xCRZ,\yA);
  % filled control dot ON the ancilla wire
  \fill (\xCRZ,\yAnc) circle (4pt);
  % RZ gate box on each data wire
  \foreach \y in {\yA,\yB,\yC} {
    \node[draw,fill=orange!30,minimum width=0.7cm,minimum height=0.55cm,
          rounded corners=2pt,font=\small] at (\xCRZ,\y) {$R_Z(\phi)$};
  }
  \node[above=10pt,gray,font=\small] at (\xCRZ,\yA) {CRZ phase};

  %------------------------------------------------------------
  % STAGE 5: H on ancilla (final)
  %------------------------------------------------------------
  \node[draw,fill=yellow!40,minimum width=0.7cm,minimum height=0.55cm,
        rounded corners=2pt] at (\xHf,\yAnc) {$H^{\otimes a}$};
  \node[above=10pt,gray,font=\small] at (\xHf,\yA) {interfere};

  %------------------------------------------------------------
  % STAGE 6: Measurements
  % Data: meter symbol; Ancilla: post-select box
  %------------------------------------------------------------
  \foreach \y in {\yA,\yB,\yC} {
    \node[draw,fill=gray!20,minimum width=0.7cm,minimum height=0.55cm,
          rounded corners=2pt,font=\small] at (\xM,\y) {$\mathcal{M}$};
  }
  \node[draw,fill=red!20,minimum width=0.85cm,minimum height=0.55cm,
        rounded corners=2pt,font=\small] at (\xM,\yAnc) {$|0\rangle$};
  \node[below=4pt,font=\scriptsize] at (\xM,\yAnc-0.3) {post-sel.};
  \node[above=10pt,gray,font=\small] at (\xM,\yA) {measure};

\end{tikzpicture}
\caption{Circuit diagram for the ancilla interference technique ($K$ seeds,
  $a$ ancilla qubits). \textbf{Init:} $H^{\otimes a}$ creates a uniform
  superposition over $K$ ancilla basis states. \textbf{Seed prep:} MCX
  gates (ancilla as control, data qubits as targets) coherently write each
  seed $|s^{(k)}\rangle$ onto the data register conditioned on the ancilla
  state $|k\rangle$. \textbf{Evolve:} The shared variational unitary
  $U(\boldsymbol{\theta})$ acts on the data register, evolving all $K$
  branches simultaneously. \textbf{CRZ phase:} Controlled-$R_Z$ gates
  (control on ancilla wire, $R_Z$ target on data qubits) apply
  seed-specific phases. \textbf{Interfere:} A final $H^{\otimes a}$ on
  the ancilla mixes the seed branches, creating cross-term interference.
  \textbf{Measure:} Data qubits are measured; ancilla is post-selected on
  $|0\rangle^{\otimes a}$, retaining only the constructively interfering
  component.}
\label{fig:interf_circuit}
\end{figure}

\textbf{Construction.}
Figure~\ref{fig:interf_circuit} shows the full circuit.
After the variational evolution $U(\boldsymbol{\theta})$, an optional layer of controlled-RZ (CRZ) phase gates applies seed-specific phases to selected data qubits. For seed $k$, a phase $\phi_{kj}$ is applied to data qubit $j$ conditioned on the ancilla being in state $|k\rangle$:
\begin{equation}
  |\Phi'\rangle = \frac{1}{\sqrt{K}}\sum_{k=0}^{K-1} e^{i\hat{\Phi}_k} U(\boldsymbol{\theta})|s^{(k)}\rangle \otimes |k\rangle_\mathrm{anc},
\end{equation}
where $\hat{\Phi}_k = \sum_j \phi_{kj}\hat{n}_j$ is a diagonal phase operator and $\hat{n}_j$ is the occupation number of qubit $j$. The CRZ gate is Hamming-weight preserving: conditioned on the ancilla state $|k\rangle$, it applies a phase $e^{i\phi_{kj}}$ to data qubit $j$ only when that qubit is in $|1\rangle$, leaving $|0\rangle$ unchanged and flipping no qubit, so the excitation-preserving structure is maintained throughout.

Hadamard gates are then applied to all $a = \lceil\log_2 K\rceil$ ancilla qubits. Post-selecting on the ancilla measuring the all-zeros state $|0\rangle^{\otimes a}$, the data-register state is
\begin{equation}
  |\Psi_\mathrm{interf}\rangle \;\propto\; \sum_{k=0}^{K-1} e^{i\hat{\Phi}_k} U(\boldsymbol{\theta})|s^{(k)}\rangle.
  \label{eq:psi_interf}
\end{equation}
The probability of measuring data-register outcome $\mathbf{x}$ in this post-selected branch is
\begin{equation}
  P_\mathrm{interf}(\mathbf{x}) = \frac{1}{K}\left|\sum_{k=0}^{K-1} A_k(\mathbf{x})\right|^2
  = \left[\frac{1}{K}\sum_k |A_k(\mathbf{x})|^2 + \frac{2}{K}\sum_{k<k'}\mathrm{Re}\!\left[A_k(\mathbf{x})A_{k'}^*(\mathbf{x})\right]\right],
  \label{eq:p_interf}
\end{equation}
where $A_k(\mathbf{x}) = \langle\mathbf{x}|U(\boldsymbol{\theta})|s^{(k)}\rangle$ is the amplitude of bitstring $\mathbf{x}$ from seed $k$ (with $\phi_{kj}=0$ for clarity). The first term is the classical mixture contribution; the second term contains the \emph{cross-terms} $\mathrm{Re}[A_k A_{k'}^*]$ that are absent in the stochastic mixture. Configurations $\mathbf{x}$ for which all seed amplitudes are in phase contribute constructively; those for which seeds partially cancel contribute destructively. The theoretical analysis of when good independent sets benefit from constructive interference is given in Section~\ref{subsec:theory_interference}.

\textbf{Practical implementation.}
CRZ phase gates are added selectively: a phase gate is applied to data qubit $j$ for seed $k$ only if qubit $j$ appears in at least one other seed. Qubits unique to a single seed would contribute only a global phase to that branch, producing no cross-term interference; restricting to shared qubits preserves all dominant interference channels while keeping the circuit compact. Each shared qubit receives one phase parameter per seed branch in which it participates, so the total number of additional phase parameters equals the sum over seeds of the number of vertices that seed shares with at least one other seed. For the 180-node instance with $K=4$ seeds of size 14, the inter-seed vertex overlap is small, resulting in a modest number of additional phase parameters relative to the variational $\theta$ parameters of the main ansatz.

The post-selection rate (probability that the ancilla measures $|0\rangle^{\otimes a}$) is exactly $1/K$ for uniform initial amplitudes and $\phi=0$, as shown in Remark~\ref{prop:postsel_rate}. For $K=4$, this means approximately 25\% of shots survive post-selection. To maintain effective sample statistics comparable to the non-interference runs, the total shot count is scaled by $K$.

\textbf{Bond dimension with ancilla interference.}
For $K=2$ ($a=1$ ancilla qubit, 2 seeds), each CRZ is a standard 2-qubit controlled-phase gate, diagonal in the computational basis, and adds no bond dimension growth; the ancilla Hadamard acts only on the single ancilla qubit. The bond dimension bound $\chi_{\max} \leq K \cdot 2^p$ from Proposition~\ref{prop:bond_dim} (where the proposition uses $m = K$) therefore applies unchanged, confirmed empirically with $\chi_{\max} = 4$ at $p=1$. For $K=3$ ($a=2$ ancilla qubits, 3 seeds), the formula gives $\chi_{\max} \leq 3 \cdot 2^p$; the two ancilla qubits require Toffoli-type decompositions for multi-bit conditioning, but the bond dimension remains tractable for the 400-node system. For $K=4$ ($a=2$ ancilla qubits, 4 seeds), the formula gives $\chi_{\max} \leq 4 \cdot 2^p$, confirmed empirically as $\chi_{\max} = 16$ at $p=2$ for the 180-node system. The additional Schmidt channel ($K=4$ vs $K=3$) pushes the bond dimension beyond the tractability threshold for the 400-qubit circuit, making $K \leq 3$ the practical limit for classical simulation of 400-qubit interference circuits.

%-----------------------------------------------------------------------------
\subsection{Channel-Resolved Measurement of the Cross Terms}
\label{subsec:channel_resolved}
%-----------------------------------------------------------------------------

Post-selection keeps one ancilla outcome and throws the rest away. Retaining every outcome
instead turns the cross terms into a quantity that can be read straight off the counts, and
the construction is easiest to follow for two seeds and one ancilla. The two channels then
carry
\begin{equation}
  P(\mathbf{x},0) = \tfrac14\big|A_0(\mathbf{x})+A_1(\mathbf{x})\big|^2 ,
  \qquad
  P(\mathbf{x},1) = \tfrac14\big|A_0(\mathbf{x})-A_1(\mathbf{x})\big|^2 ,
  \label{eq:chan_twoseed}
\end{equation}
so that their sum and their difference separate the two parts of the output,
\begin{align}
  P(\mathbf{x},0)+P(\mathbf{x},1)
    &= \tfrac12\Big(|A_0(\mathbf{x})|^2+|A_1(\mathbf{x})|^2\Big) = P_\mathrm{stoch}(\mathbf{x}) ,
    \label{eq:chan_sum}\\
  P(\mathbf{x},0)-P(\mathbf{x},1)
    &= \mathrm{Re}\big[A_0(\mathbf{x})A_1^*(\mathbf{x})\big] .
    \label{eq:chan_diff}
\end{align}
The sum returns the stochastic mixture and the difference returns the cross term, so their
ratio
\begin{equation}
  V(\mathbf{x}) = \frac{P(\mathbf{x},0)-P(\mathbf{x},1)}{P(\mathbf{x},0)+P(\mathbf{x},1)}
  = \frac{2\,\mathrm{Re}\big[A_0(\mathbf{x})A_1^*(\mathbf{x})\big]}
         {|A_0(\mathbf{x})|^2+|A_1(\mathbf{x})|^2}
  \label{eq:chan_V}
\end{equation}
is the cross term expressed as a fraction of the mixture. It lies between $-1$ and $1$,
reaches those values when the two branches carry equal amplitude in phase or out of phase,
and vanishes for any bitstring that only one branch reaches, since the product $A_0A_1^*$ is
then zero.

Writing $n_0(\mathbf{x})$ and $n_1(\mathbf{x})$ for the measured counts in the two channels
and $N(\mathbf{x})$ for their total, the same quantity is estimated by
\begin{equation}
  \hat{V}(\mathbf{x}) = \frac{n_0(\mathbf{x})-n_1(\mathbf{x})}{N(\mathbf{x})} .
  \label{eq:chan_Vhat}
\end{equation}
When the cross term vanishes the two channels are equally likely, so $n_0(\mathbf{x})$
follows a binomial distribution with $N(\mathbf{x})$ trials and probability one half.
Measuring the cross term at one bitstring is therefore a test of whether a coin is fair, and
no averaging over bitstrings is involved.

For four seeds and two ancillas the same construction gives three contrasts rather than one,
since the four channels admit three ways of splitting into two pairs,
\begin{equation}
  \begin{array}{lll}
    c_A: & + \text{ on } \{00,10\}, & - \text{ on } \{01,11\} \\
    c_B: & + \text{ on } \{00,01\}, & - \text{ on } \{10,11\} \\
    c_C: & + \text{ on } \{00,11\}, & - \text{ on } \{01,10\}
  \end{array}
  \label{eq:chan_patterns}
\end{equation}
and each splitting isolates a different combination of the six pairwise cross terms. The
fourth combination, with a plus sign on all four channels, sums the channels and returns the
stochastic mixture. This is the ancilla marginal identity of
Section~\ref{subsec:theory_interference}, and it carries no cross term information.

Writing $D_c(\mathbf{x})$ for the signed sum of counts under pattern $c$, so that
$D_c = n_+ - n_-$ and $\hat{V}_c = D_c/N$, the statistic is
\begin{equation}
  T = \sum_{\mathbf{x}} \sum_{c} \frac{D_c(\mathbf{x})^2}{N(\mathbf{x})} ,
  \label{eq:chan_stat}
\end{equation}
whose expectation in the absence of cross terms is the number of terms, and which is shifted
upward by $\sum_{\mathbf{x},c} N(\mathbf{x})\,V_c(\mathbf{x})^2$ when contrasts are present.
Many bitstrings are observed only a few times, so the null distribution of $T$ is obtained by
drawing $200$ multinomial replicates of the observed $N(\mathbf{x})$ with equal probability on
the four patterns rather than by assuming a chi-squared form.

The construction was checked against the identical seed limit, where all branch amplitudes
coincide and the ancilla Hadamards return every shot to the $|0\rangle^{\otimes a}$ channel.
With four identical seeds and $\phi = 0$ all $20\,000$ shots were returned in that channel and
the other three were empty.

% %-----------------------------------------------------------------------------
% \subsection{Use of Generative AI Tools}
% \label{subsec:ai_tools}
% %-----------------------------------------------------------------------------
% 
% Generative AI tools were used at several points in this work. Claude
% (Anthropic)~\cite{claudeai} assisted with the literature survey, with the design of the
% adversarial evolutionary graph generator behind the 99-node benchmark instance, and with
% language polishing and drafting support during preparation of the manuscript. ChatGPT
% (OpenAI)~\cite{chatgpt} assisted with the design of the two-layer CSP and spin-glass
% encoding behind the 180-node instance, as noted in Section~\ref{sec:setup}. IBM~Bob
% \cite{ibmbob} assisted with code development for the pipeline implementation. The research
% question, the algorithmic design, the experiments and the interpretation of the results
% are the authors' own. The authors verified every generated design and every line of
% generated code, ran all experiments reported here, revised the manuscript throughout, and
% take full responsibility for its content.

%=============================================================================
\section{Experimental Setup}
\label{sec:setup}
%=============================================================================

The experimental evaluation spans benchmark instances ranging from 64 to
400 nodes. Three core instances, a 64-node random graph, a 99-node random
graph, and the 180-node primary challenge instance, are used for algorithm
development and ablation studies and are described below. Additional instances
(the QOBLIB 200-node suite, \texttt{c125-9}, and \texttt{brock400-1}) are
introduced in Section~\ref{sec:results}. Key properties of the core instances are
summarized in Table~\ref{tab:instances}.

\begin{table}[!htbp]
\centering
\caption{Benchmark instance properties. Average degree $\bar{d} = 2|E|/|V|$
  and maximum degree $d_{\max}$ are computed from the graph edge lists.
  The 99-node instance is nearly regular ($d_{\min}=55$, $d_{\max}=58$);
  the 180-node instance is similarly dense ($d_{\min}=54$, $d_{\max}=64$).}
\label{tab:instances}
\begin{tabular}{@{}lccccccl@{}}
\toprule
Instance & $|V|$ & $|E|$ & $\bar{d}$ & $d_{\min}$ & $d_{\max}$ & MIS size & Source \\ \midrule
64-node  & 64  & 543  & 16.97 & 6  & 24 & 10 & QOBLIB benchmark \\
99-node  & 99  & 2780 & 56.16 & 55 & 58 & 7  & Hard random instance \\
180-node & 180 & 5277 & 58.63 & 54 & 64 & 15 & Custom dense instance \\ \bottomrule
\end{tabular}
\end{table}

The first, a 64-node graph, is drawn directly from the QOBLIB benchmark library ~\cite{koch2025qoblib}, a standardized repository of MIS instances designed for quantum algorithm evaluation. With $543$ edges and average degree $17.0$ it is relatively sparse, and its known MIS of size 10 represents $15.6$\% of the vertex set.

The second, a 99-node instance with $2780$ edges and an average degree $56.2$, is a hard, dense graph; its known MIS of size 7 represents only 7.1\% of vertices. It was produced by an adversarial evolutionary graph generator (designed with assistance from Claude.ai ~\cite{claudeai}). Seed graphs of a fixed vertex count $n=99$ are drawn from a mix of Erdős–Rényi derived constructions with edge probability $p \in [0.48, 0.52]$ and a planted decoy clique of size $c_1 \in [12, 18]$, random-regular graphs, and compositions of these joined at an interface. Candidates are scored on a composite hardness measure combining CPLEX optimality gap, branch-and-bound node count, runtime, LP and cut-relaxation weakness, reduction resistance, near-optimal solution density and symmetry. The pool is then evolved over $T = 8$ generations with archive size $M = 8$ through degree-preserving edge rewiring (2-switches), single edge add/remove, vertex near-cloning, shielding of the certified optimum, and a root-gap-amplifying motif insertion. Running on a single CPLEX 22.1.2.0 thread (Apple M4 Max, three seeds) certifies the optimum in 1.10–1.32s after 230–258 branch-and-bound nodes and 14.8k–15.7k simplex iterations, and ReduMIS recovers the same value 7 on all three seeds.

The third, a $180$-node instance with 5277 edges and average degree $58.6$, was constructed for the large-scale regime adversarially via a two-layer CSP and spin-glass encoding (designed with assistance from ChatGPT~\cite{chatgpt}). The conflict graph is built over $q = 30$ variables with $k = 6$ values each from a 13-regular CSP layer ($\lambda = 3$ forbidden value-pairs per constraint), an 11-regular signed spin-glass layer contributing $\mu = 2$ frustrated value-pairs per edge, and intra-variable cliques over each variable's $k$ value-vertices. An evolutionary loop with populations of $5$ and $10$ generations then mutates the encoding while scoring candidates against CPLEX, ReduMIS, and KaMIS, plus a spin-landscape ruggedness term. The intent is to inherit difficulty from frustrated Ising models and tight CSPs rather than engineer it on raw edges. Under the same single-threaded CPLEX 22.1.2.0 setup (three independent seeds), this instance requires 52.4--53.9 seconds and 34{,}100--42{,}400 branch-and-bound nodes to certify the MIS of size 15, consuming 1.4--1.8 million simplex iterations. ReduMIS recovers the optimum on all three seeds within a 60-second budget. The higher certification cost compared to the 99-node instance (230--258 B\&B nodes, 1.1--1.3 seconds) reflects the intended difficulty of the two-layer construction.

Both generators score candidates only against classical solvers and classical relaxation
measures, namely CPLEX, ReduMIS and KaMIS together with LP and cut relaxation weakness,
reduction resistance and landscape ruggedness. The selection criteria therefore reflect
classical difficulty alone and contain no term that rewards any property of the quantum
pipeline.

The increasing density from the first to the second and third instances, reflected in both the
average degree and the lower MIS ratio means harder QAOA circuits
(more RZZ gates, lower coverage at any given bandwidth $k$) and harder bitstring decoding
(the optimal solution occupies a smaller corner of the search space). All three instances
are expressed in DIMACS graph format. The MIS cardinalities were certified by CPLEX
\cite{ibm_cplex} and independently reproduced by ReduMIS, as reported above; the complete
MIS population of each instance was subsequently enumerated using the CPLEX solution pool in
exhaustive enumeration mode, independently verified by a Bron--Kerbosch search
\cite{bron1973} on the complement graph (Section~\ref{subsec:phase3_sensitivity}).

All simulations were conducted using the Qiskit framework ~\cite{qiskit2023} with the
Qiskit Aer MPS simulator backend under noiseless conditions. The MPS simulator represents
the quantum state as a one-dimensional tensor network with a tunable bond dimension
$\chi$, rather than storing the full $2^n$-dimensional state vector. A full state-vector
simulation of 180 qubits would require $2^{180} \approx 10^{54}$
complex amplitudes, which is classically intractable by many orders of magnitude.
The MPS simulator is exact for circuits with sufficiently low entanglement entropy and
provides a controlled approximation otherwise; the spectral reordering and sparsification
steps described in Sections~\ref{subsec:reorder} and ~\ref{subsec:sparse} directly reduce
circuit entanglement by concentrating interactions between neighbouring qubits, making
the MPS representation effective in this setting.

For context, gate-based quantum algorithm literature on MIS
have generally been confined to instances of 10--50 qubits. The symmetry
protection analysis of Bravyi et al.~\cite{bravyi2020obstacles} considers
instances of up to 20 qubits; the iterative quantum algorithm of Brady and
Hadfield~\cite{brady2024iterative} reaches graphs of up to 26 vertices; and
the multi-controlled gate decomposition study of Tomesh et
al.~\cite{tomesh2024tradeoffs} examines QAOA-MIS tradeoffs in the range of
10--30 qubits. The HyDRA-MIS decomposition framework of Xu et
al.~\cite{xu2025qaoa} handles larger effective graphs by breaking them into
subproblems of at most 25 qubits, explicitly noting that simulating QAOA
circuits beyond 50 qubits exceeds their capabilities. Similarly, the
quantum-informed recursive optimization (QIRO) algorithm of Finžgar et
al.~\cite{finzgar2024qiro} solves MIS instances with hundreds of variables
by using fixed shallow ($p=1$) QAOA circuits only to extract local
correlations that guide a classical recursive reduction procedure, rather
than optimizing a large variational circuit directly. The quantum-classical
hybrid qReduMIS approach~\cite{schuetz2025qreducmis} employs up to 231
physical qubits on Rydberg atom-array hardware, but operates on unit-disk
graphs whose independent-set structure is physically encoded in the
blockade radius, which is a fundamentally different regime from general
graph MIS on gate-based processors. MPS-based QAOA simulations for MaxCut
have reached 512 qubits for 3-regular graphs with shallow
circuits~\cite{mps_juliQAOA2025}, but MaxCut is an unconstrained problem
that benefits from graph regularity and does not require enforcing
independent set constraints. Against the works cited above, the MPS-VQE simulations presented here, which solve MIS to optimality on instances
up to 200 nodes and reach 92.6\% of the certified optimum on a 400-node instance, are the
largest we are aware of for gate-based variational quantum algorithms applied to hard
general-graph MIS instances with the full variational circuit optimised directly. The
comparison is drawn against these works and not against the literature as a whole.

Shot budgets during the variational optimization phase were tuned to balance statistical
accuracy against per-iteration runtime, which grows with both system size and circuit
depth. For the 64- and 99-node instances shot counts of 2000--4000 per energy evaluation
were used. For the 180-node instance the budget was kept in the range 1000--4000, with
lower values used for Phase~3 to keep per-run runtimes tractable. In all cases, once the
optimizer converged, post-optimization sampling followed Mode~B (Section~\ref{subsec:heuristics},
Figure~\ref{fig:modeB}): $T=10$ iterations were run at the converged parameters, with
$B=50{,}000$ shots drawn per iteration, and the EMA suppression weights and bitstring
heuristic pipeline updated after each iteration ($\alpha=0.4$). The penalty parameter was
set to $P = 2$ for all instances up to 180 nodes, verified to be sufficient by the
penalty bound proved in Sec.~S1 of the Supplemental Material. For the 400-node \texttt{brock400-1} instance, the VQE+CVaR
(Phase~1) experiments required a higher penalty coefficient $P \in \{4, 6\}$ due to the
increased graph density. The excitation-preserving Phase~3 ansatz retains $P = 2$, as
Hamming-weight conservation removes the need for a strong penalty. Table~\ref{tab:configs}
summarizes all experimental configurations.

All simulations were run on an Apple M1 Pro with 32\,GB RAM. Table~\ref{tab:runtimes}
reports wall-clock runtimes at the practically effective iteration counts identified in
the experiments: 6{,}000 iterations for Phase~1 (VQE with EfficientSU2) and 2{,}000 iterations
for Phase~3 (excitation-preserving VQE with ancilla superposition).
Phase~3 converges in fewer iterations than Phase~1 requires at comparable shot counts,
reflecting the much lower initial energy of the ancilla superposition state. Since the seeds are
already near-optimal independent sets, Phase~3 starts far closer to the optimum than Phase~1
which begins from a state with no structural information, and the optimizer has
far less ground to cover. Within Phase~3, runtime scales approximately linearly with shot
count up to 4{,}000 shots, then grows faster at 8{,}000 shots due to increased MPS
simulation cost per circuit evaluation. The bitstring heuristic pipeline adds a further
1.5--2.2 minutes regardless of configuration.

For context, the classical local search (CLS) post-processing step of the pipeline itself runs in sub-second wall-clock time per trial. CPLEX certifies the 180-node MIS classically in 52--54 seconds on a single thread, and ReduMIS recovers it heuristically within 60 seconds. The quantum simulation pipeline is slower, reflecting the overhead of noiseless MPS simulation rather than quantum hardware execution time. The underlying Phase~3 circuit, with $p=2$ and nearest-neighbour $fSim$ gates on 182 qubits, would execute in microseconds on actual quantum hardware; the simulation cost arises from computing the full MPS state vector across thousands of SPSA iterations. The noiseless simulation establishes that the algorithmic pipeline is correct and that the MIS is recoverable from the quantum measurement distribution, as a proof of concept for future hardware implementation. Runtime competitiveness with classical solvers is not claimed.

\begin{table}[!htbp]
\centering
\caption{Summary of experimental configurations. Shots/eval refers to shots per energy evaluation during optimization. $^\dagger$CVaR objective \cite{barkoutsos2020improving} with $\alpha \in [0.1, 0.2]$ used in place of the standard expectation value. Phase~3 of the 180-node instance uses the standard expectation value (ancilla superposition initialization already concentrates the circuit near good solutions).}
\label{tab:configs}
\begin{tabular}{@{}llllc@{}}
\toprule
Instance & Algorithm & Optimizer & Layers $p$ & Shots/eval \\ \midrule
64-node  & VQE (EfficientSU2)$^\dagger$ & COBYLA & 2 & 2000--4000 \\
64-node  & VQE (EfficientSU2)$^\dagger$ & SPSA & 2, 3 & 2000--4000 \\
64-node  & QAOA$^\dagger$ ($k=1$--$3$, $p=2$--$5$) & COBYLA/SPSA & 2--5 & 2000--4000 \\
99-node  & VQE (EfficientSU2)$^\dagger$ & COBYLA & 2 & 2000--4000 \\
99-node  & VQE (EfficientSU2)$^\dagger$ & SPSA & 2 & 2000--4000 \\
180-node & VQE (EfficientSU2)$^\dagger$ & SPSA & 2 & 2000--4000 \\
180-node & VQE (ExcitPres.) + Ancilla & SPSA & 2 & 1000--2000 \\ \bottomrule
\end{tabular}
\end{table}

\begin{table}[!htbp]
\centering
\caption{Wall-clock runtimes at the practically effective iteration counts,
  measured on an Apple M1 Pro (32\,GB RAM). Phase~1 uses VQE with
  EfficientSU2 ($p=2$, 6{,}000 iterations); Phase~3 uses excitation-preserving
  VQE with ancilla superposition ($p=2$, 2{,}000 iterations).
  Bitstring heuristics runtime is independent of configuration.}
\label{tab:runtimes}
\begin{tabular}{@{}llcr@{}}
\toprule
System & Configuration & Shots/eval & Runtime (min) \\ \midrule
\multicolumn{4}{l}{\textit{64-node instance}} \\
& VQE/SPSA, $p=2$, 6K iters     & 2{,}000 & 40.8 \\
& VQE/COBYLA, $p=2$, 3.3K iters & 2{,}000 & 23.9 \\
& QAOA/SPSA, $p=2$, 6K iters    & 2{,}000 & 73.2 \\ \midrule
\multicolumn{4}{l}{\textit{99-node instance}} \\
& VQE/SPSA, $p=2$, 6K iters   & 2{,}000 & 148.0 \\
& VQE/COBYLA, $p=2$, 6K iters & 2{,}000 & 139.0 \\ \midrule
\multicolumn{4}{l}{\textit{180-node, Phase~1 (EfficientSU2, $p=2$, 6K iters)}} \\
& VQE/SPSA & 1{,}000 & 147.9 \\
& VQE/SPSA & 2{,}000 & 303.6 \\ \midrule
\multicolumn{4}{l}{\textit{180-node, Phase~3 (excitation-preserving, $p=2$, 2K iters)}} \\
& EXCP/SPSA & 1{,}000 &  58.9 \\
& EXCP/SPSA & 2{,}000 &  73.9 \\
& EXCP/SPSA & 4{,}000 & 100.6 \\
& EXCP/SPSA & 8{,}000 & 206.4 \\ \midrule
\multicolumn{3}{l}{Bitstring heuristics (all configurations)} & 1.5--2.2 \\
\bottomrule
\end{tabular}
\end{table}

\subsection{Hardware Validation Setup}
\label{subsec:hw_setup}

To assess whether the noiseless simulation results translate to real
quantum devices, the converged VQE configurations for the 64- and
99-node instances, and three QAOA configurations for the 64-node
instance, were executed on the 156 qubit IBM Quantum hardware \texttt{ibm\_marrakesh} (fielding Heron r2 processor) via an open-instance allocation with a 10-minute job time
budget \cite{ibmquantum}. Rather than re-optimizing on hardware, we adopt a
parameter-transfer protocol: the variational parameters obtained from
the noiseless MPS-simulator optimization (Table~\ref{tab:configs}) are
transferred directly to the hardware circuit, which is then transpiled
at optimization level~3 and sampled without further optimization. Hardware circuit execution used the Qiskit Sampler primitive
\cite{qiskit2023}, drawing $B = 20{,}000$ shots per Mode~B iteration. Since no
re-optimization was performed on hardware, the Estimator primitive
was not required. This isolates the effect of hardware noise on the measurement distribution and the downstream bitstring heuristic pipeline, independent of any
hardware-specific optimization difficulties.

The 10-minute time budget constrained each Mode~B iteration to
$B=20{,}000$ shots on hardware, compared to $B=50{,}000$ shots per
iteration in the noiseless setting ($T=10$ in both cases).
Table~\ref{tab:hw_configs} summarizes the transpiled two-qubit gate
depths for each configuration. The VQE circuits, with their
nearest-neighbour EfficientSU2 entanglement structure, transpile to
depths of 65 (64-node) and 100 (99-node). This indicates that
EfficientSU2's nearest-neighbour structure, maps onto the hardware's
native connectivity without SWAP-induced overhead, so that circuit
depth scales linearly with system size even after transpilation.

The 125-node \texttt{c125-9} instance from the QOBLIB benchmark library~\cite{koch2025qoblib} (certified MIS size 34) was also executed on \texttt{ibm\_marrakesh}
using the same parameter-transfer protocol. Converged VQE/SPSA parameters
($p=2$) from noiseless simulation were transferred directly to the hardware
circuit, which transpiles to a two-qubit depth of 126. The same Mode~B
sampling protocol was applied ($T=10$ iterations, $B=20{,}000$ shots per
iteration), with approximately 8 seconds of actual QPU execution time per
iteration, consistent with the 64- and 99-node runs.

The three QAOA configurations, by contrast, transpile to depths of
205--337, reflecting the long-range $RZZ$ couplings discussed in
Section~\ref{subsec:sparse}. Increasing the sparsification bandwidth
from $k=3$ to $k=4$ at fixed $p=3$ increases the noiseless optimization
time from 79.9s to 1477.1s, an 18-fold increase for a single unit of
bandwidth. This is a striking independent confirmation, in a
hardware-targeted setting, of the MPS bond-dimension bound
$\chi \leq 2^{O(kp)}$ derived in Section~\ref{subsec:theory_reorder}:
each additional unit of entanglement bandwidth exponentially increases
the per-iteration simulation cost. A bandwidth of $k=4$ was tractable only as a one-off optimization to generate hardware-transfer parameters. The full Mode-B post-processing protocol ($T=10$ iterations of $B=50{,}000$ shots each) would have been prohibitively slow, hence its exclusion from the main noiseless sweep in Table~\ref{tab:configs}.

\begin{table}[!htbp]
\centering
\caption{Hardware configurations executed on \texttt{ibm\_marrakesh}.
  Optimization time is the noiseless MPS-simulator wall-clock time in
  minutes to obtain the transferred parameters. Total time is the
  end-to-end wall-clock time in minutes for the hardware stage,
  comprising queueing, sampling 20{,}000 shots on the QPU across
  $T=10$ Mode~B iterations, and running the bitstring heuristic
  pipeline. Of this, the IBM Quantum platform reported approximately
  8 seconds of actual QPU execution time per iteration ($\approx 80$
  seconds total across all 10 iterations); the remainder is queueing
  and classical post-processing overhead.}
\label{tab:hw_configs}
\begin{tabular}{@{}llccc@{}}
\toprule
 &  & Two-qubit & Opt.\ time & Total time \\
System & Algorithm & depth & (min) & (min) \\
\midrule
64-node  & VQE/SPSA, $p=2$            &  65 &  36.7 &  3.7 \\
99-node  & VQE/SPSA, $p=2$            & 100 & 140.3 &  3.6 \\
125-node (\texttt{c125-9}) & VQE/SPSA, $p=2$ & 126 & 18.36 & 19.70 \\
64-node & QAOA/COBYLA, $p=3$, $k=4$ & 337 & 24.6  & 11.4 \\
64-node & QAOA/COBYLA, $p=3$, $k=3$ & 205 & 1.3   & 6.4  \\
64-node & QAOA/COBYLA, $p=5$, $k=3$ & 323 & 28.2  & 37.2 \\
\bottomrule
\end{tabular}
\end{table}

%=============================================================================
\section{Results and Discussion}
\label{sec:results}
%=============================================================================

\subsection{Regimes and What Each Instance Requires}
\label{subsec:regimes}

The instances fall into two regimes, and the distinction organises the results that follow.
Instances of at most 156 nodes fit within the qubit count of the hardware available to us
(\texttt{ibm\_marrakesh}, 156 qubits). For these, standard VQE with the EfficientSU2 ansatz,
CVaR, and the bitstring heuristic pipeline recovers the certified MIS, and neither the
ancilla superposition nor the interference circuit is required. The 64, 99 and 125 node
instances are all solved this way, and all three are executed on hardware
(Section~\ref{subsec:hardware}). Beyond 156 nodes the picture changes and the simulations
were all done in MPS simulator. Single seed
initialisation stalls at a suboptimal size on the 180-node instance regardless of iteration
count or restarts, and the ancilla superposition and interference circuits are introduced
for this regime. These instances exceed the available qubit count, 182 qubits for the
180-node system with its ancillas and approximately 402 qubits for the 400-node system,
so they are necessarily studied under MPS simulation. The ancilla techniques are therefore
not offered as an improvement on the hardware validated pipeline. They address problem
sizes that the pipeline cannot reach and that current hardware cannot execute.

\subsection{64-Node Instance}

Figure~\ref{fig:64_results} shows the optimization convergence for the two
VQE configurations and the QAOA run on the 64-node instance.
VQE with COBYLA ($p=2$, Figure~\ref{fig:64_vqe_cobyla}) exhibits smooth
monotonic energy descent over approximately 35 iterations, converging to a
plateau around energy 160. VQE with SPSA ($p=3$, Figure~\ref{fig:64_vqe_spsa})
achieves deeper energy minimization over 4000--6000 iterations. The
stochastic nature of SPSA allows it to escape shallow local minima, and the
$p=3$ circuit provides greater expressibility. QAOA with COBYLA ($p=2$,
$k=3$, Figure~\ref{fig:64_qaoa}) shows the oscillatory
convergence typical of QAOA's alternating cost-mixer structure.

The real difference between VQE and QAOA on this instance is not
solution quality per run, but the \emph{diversity} of solutions recovered from
the measurement distribution. VQE with SPSA recovers 700--800 distinct
independent sets per post-optimization sampling iteration (one Mode~B
iteration with $B=50{,}000$ shots, Section~\ref{sec:setup}), while QAOA
configurations consistently recover only 0--5 independent sets per
iteration.
This disparity has a direct impact on the ability to find the MIS.

We employed two strategies to extract the MIS from the recovered sets.
\textbf{Strategy 1} takes the best single independent set found and drives it
to maximality (single-path extension). \textbf{Strategy 2} drives every
recovered independent set to maximality and reports the largest result. With
Strategy 1, the best results were:
\begin{itemize}[noitemsep]
  \item VQE/COBYLA, $p=2$: $\{2,13,24,35,36,47,54,57\}$, size~8
  \item VQE/SPSA, $p=2$: $\{0,6,15,21,24,33,54,57,63\}$, size~9
  \item VQE/SPSA, $p=3$: $\{0,3,12,30,39,42,48,57,63\}$, size~9
  \item QAOA/COBYLA, $p=5$, $k=3$: $\{4,7,21,31,33,46,51,56\}$, size~8
  \item QAOA/COBYLA, $p=2$, $k=3$: $\{4,7,25,30,33,38,55,56\}$, size~8
  \item QAOA/SPSA, $p=2$, $k=2$: $\{0,3,10,15,24,45,51,60,63\}$, size~9
\end{itemize}
With Strategy 2, which uses the full diversity of recovered sets, VQE
configurations successfully identify multiple maximum independent sets of
size 10:
\begin{align*}
  &\{0, 3, 12, 15, 18, 45, 48, 51, 60, 63\}, \quad
   \{0, 3, 12, 15, 21, 34, 46, 51, 52, 63\}, \\
  &\{0, 7, 12, 18, 30, 33, 45, 51, 52, 63\}, \quad
   \{0, 7, 12, 17, 29, 42, 48, 51, 60, 63\}.
\end{align*}
QAOA, despite reaching size~9 with Strategy~1, cannot reach size~10 with
Strategy~2 because it recovers too few independent sets (0--5 per run) to
provide the diversity needed. With only a handful of sets to drive to
maximality, the chance of landing on a size-10 solution is negligible. This
directly illustrates the importance of solution diversity as a precondition
for the maximality extension step, and explains why VQE with the
hardware-efficient ansatz, which populates the measurement distribution much
more broadly, outperforms QAOA on this instance despite comparable convergence
profiles. No tangible benefit in solution quality was found beyond $p=2$--$3$
for VQE on this instance.

We also evaluated the Hamiltonian-Guided Leverage Embedding (HGLE) approach~\cite{mukherjee2026hgle}, a hybrid pipeline that compresses QAOA measurement samples via leverage-score subspace embedding for noise-robust parameter estimation. On the 64-node instance, HGLE improves QAOA's measurement diversity: HGLE-QAOA recovers 15--150 distinct independent sets per run (versus 0--5 for standard QAOA), and HGLE-VQE with Strategy~2 recovers the same four MIS of
size~10 as the standard pipeline with 1000-1800 sets per run. Despite the improved diversity, HGLE-QAOA reaches only size~9 with Strategy~1, consistent with the structural limitations of QAOA identified above. HGLE did not converge on the denser 99- and 180-node instances, where the high-dimensional landscape appears to exceed the subspace compression assumptions of the method. HGLE was not evaluated on the 200- and 400-node benchmarks introduced in this work.

With CVaR ($\alpha \in [0.1,0.2]$, 6000 iterations, 2000 shots),
the pipeline recovers up to 6 distinct MIS of size~10 per run.
Two independent SPSA runs with different stochastic trajectories
each recovered 6 MIS, with significant overlap across runs
(sets~1, 2, 4, 5 recovered in both), confirming robustness of
the pipeline to the inherent randomness of SPSA. Run~1 recovered:
\begin{align*}
  &\{0,11,12,17,29,42,48,51,60,63\}, \quad
   \{0,3,12,15,21,34,46,51,56,63\}, \\
  &\{0,11,12,18,30,33,45,51,56,63\}, \quad
   \{0,7,12,17,29,42,48,51,60,63\}, \\
  &\{0,7,12,18,30,33,45,51,56,63\}, \quad
   \{0,3,12,15,18,45,48,51,60,63\}.
\end{align*}
Run~2 recovered the same four core sets plus two variants.
CPLEX population-pool enumeration~\cite{ibm_cplex}, independently
verified by a Bron--Kerbosch search~\cite{bron1973} on the complement
graph, certifies that the 64-node instance has exactly 9 distinct MIS,
all sharing a common 4-vertex core $\{0,12,51,63\}$. Across repeated
CVaR runs, all 9 certified MIS were recovered, confirming a 100\% enumeration
rate and demonstrating that the pipeline exhaustively samples the MIS
population for this instance.

\begin{figure}[!htbp]
\centering
\begin{subfigure}[t]{0.32\linewidth}
  \includegraphics[width=\linewidth]{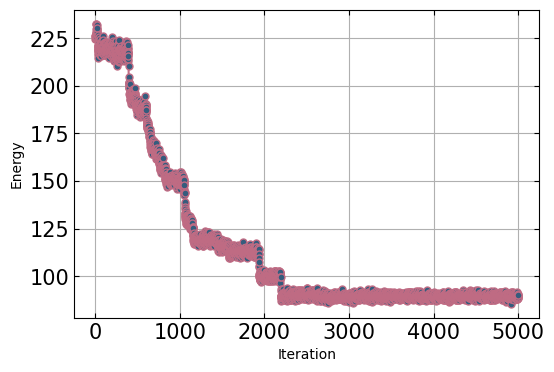}
  \caption{VQE, COBYLA, $p=2$}
  \label{fig:64_vqe_cobyla}
\end{subfigure}\hfill
\begin{subfigure}[t]{0.32\linewidth}
  \includegraphics[width=\linewidth]{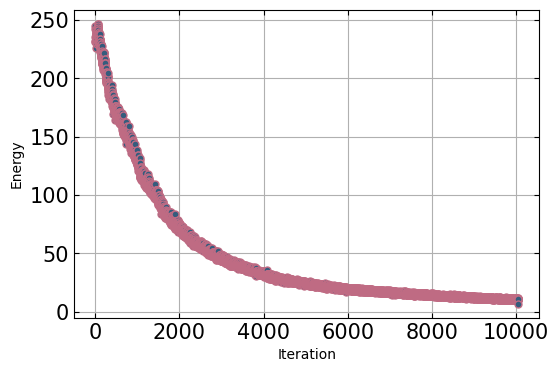}
  \caption{VQE, SPSA, $p=3$}
  \label{fig:64_vqe_spsa}
\end{subfigure}\hfill
\begin{subfigure}[t]{0.32\linewidth}
  \includegraphics[width=\linewidth]{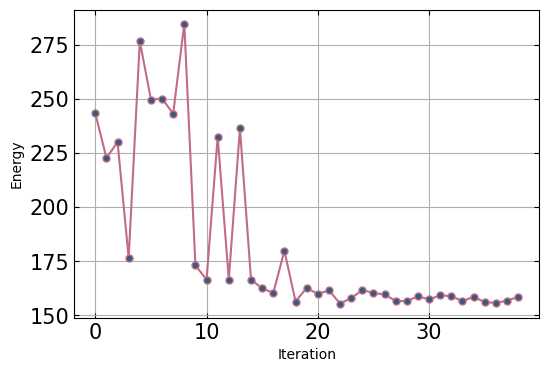}
  \caption{QAOA, COBYLA, $p=2$, $k=3$}
  \label{fig:64_qaoa}
\end{subfigure}
\caption{Optimization convergence on the 64-node instance for three
  configurations. COBYLA-based runs (left and right) plateau early due to
  limitations in high-dimensional parameter space, while SPSA (center) achieves
  continuous descent to near-zero energy over $10000$ iterations, ultimately
  yielding the maximum independent set of size 10 after post-processing.}
\label{fig:64_results}
\end{figure}

\subsection{99-Node Instance}

QAOA was not pursued for the 99-node instance. As detailed in
Section~\ref{subsec:depth} and Figure~\ref{fig:coverage}, the MPS simulator
could not advance beyond $k=1$ for this instance (2780 edges, average degree
56.2), where coverage is only 2.3\%. At $k=1$, the QAOA cost unitary is
structurally equivalent to VQE, removing any motivation to use QAOA. All
experiments therefore used VQE with the EfficientSU2 ansatz.

Figure~\ref{fig:99_results} presents convergence profiles for VQE with COBYLA
and SPSA ($p=2$, 4000 shots). The COBYLA run (Figure~\ref{fig:99_cobyla})
exhibits a characteristic staircase pattern with distinct plateau-and-drop
events, decreasing from $\sim$1330 to $\sim$340 but stalling there, and
recovering on average 60--100 distinct independent sets per
50{,}000-shot sampling iteration.
SPSA (Figure~\ref{fig:99_spsa}) achieves a qualitatively different trajectory:
continuous descent throughout 4000–6000 iterations to approximately 90, and
recovers 150--200 distinct independent sets per sample, roughly twice as many
as COBYLA. As with the 64-node instance, we used both Strategy~1 (best single
set driven to maximality) and Strategy~2 (all sets driven to maximality).

With Strategy~1, the best immediate results were size~6 for both optimizers:
VQE/COBYLA recovers $\{4,11,28,39,61,99\}$ and VQE/SPSA recovers
$\{8,20,25,35,57,97\}$. With Strategy~2, exploiting the full diversity of
recovered sets, both optimizers identify MIS of size~7.
VQE/COBYLA finds two distinct MIS:
\begin{align*}
\{28, 33, 44, 51, 56, 61, 95\}, \quad
\{12, 21, 31, 41, 48, 53, 63\}.
\end{align*}
VQE/SPSA finds four distinct MIS, including two not found by COBYLA:
\begin{align*}
\{28, 33, 44, 51, 56, 61, 95\}, \quad
\{12, 21, 31, 41, 48, 53, 63\}, \\
\{25, 29, 35, 57, 62, 65, 89\}, \quad
\{12, 21, 41, 48, 53, 55, 63\}.
\end{align*}
The broader measurement distribution produced by SPSA (150--200 distinct
independent sets versus 60--100 for COBYLA) directly enables wider coverage
of the MIS solution space. No tangible improvement in solution quality was
found by going beyond $p=2$ for VQE on this instance.

With CVaR ($\alpha \in [0.1, 0.2]$, 6000 iterations, 2000 shots per
evaluation), the diversity of the measurement distribution improves.
 A single SPSA run recovers up to 10 distinct MIS of size~7:
\begin{align*}
&\{4,6,10,43,58,72,92\}, \quad \{19,26,38,61,69,71,85\}, \\
&\{11,14,58,73,91,92,93\}, \quad \{44,45,56,74,78,79,95\}, \\
&\{12,15,60,66,73,88,92\}, \quad \{1,14,44,74,78,93,95\}, \\
&\{12,15,22,49,60,73,88\}, \quad \{1,2,44,78,93,94,95\}, \\
&\{28,33,44,51,56,61,95\}, \quad \{1,14,44,48,74,78,93\}.
\end{align*}
These 10 solutions span structurally diverse regions of the graph, with
many sets sharing no vertices at all. A single run thus samples broadly
from the MIS population; repeated runs with different SPSA trajectories
collectively enumerate a larger fraction of all MIS. This simultaneous recovery of multiple distinct MIS solutions is a
qualitative advantage of the sampling-based quantum approach over
classical exact solvers, which typically return a single optimal
solution. The same enumeration procedure certifies that the 99-node
instance has exactly 56 distinct MIS~\cite{ibm_cplex,bron1973}.
With approximately 10 MIS recovered per run, the pipeline samples
a meaningful fraction of this population in each run, with different
runs covering different subsets.

\begin{figure}[!htbp]
\centering
\begin{subfigure}[t]{0.48\linewidth}
  \includegraphics[width=\linewidth]{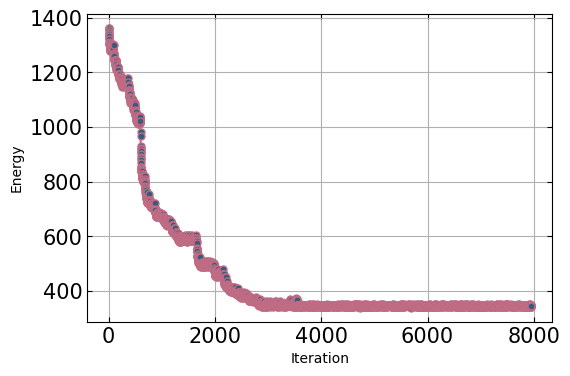}
  \caption{VQE, COBYLA, $p=2$, 4k shots}
  \label{fig:99_cobyla}
\end{subfigure}\hfill
\begin{subfigure}[t]{0.48\linewidth}
  \includegraphics[width=\linewidth]{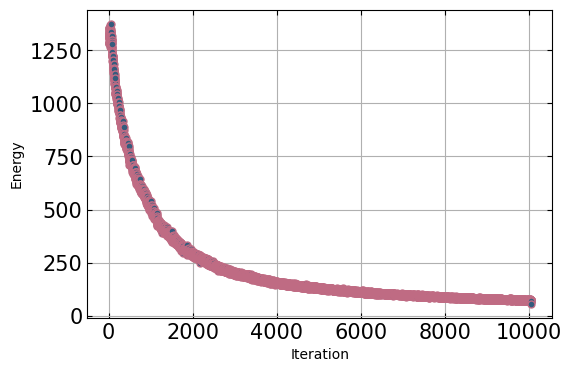}
  \caption{VQE, SPSA, $p=2$, 4k shots}
  \label{fig:99_spsa}
\end{subfigure}
\caption{Optimization convergence on the 99-node instance. COBYLA (left)
  shows plateau-and-drop dynamics and stalls at energy $\approx 340$ over
  $\sim$35 function evaluations. SPSA (right) achieves sustained descent to energy $\approx 90$ over
  4000--6000 iterations, recovering up to 10 distinct maximum independent
  sets of size~7 per run with CVaR post-processing.}
\label{fig:99_results}
\end{figure}

\subsection{180-Node Instance}

The 180-node instance is the central challenge of this study.
Figure~\ref{fig:180_results} shows the convergence of both phases, and
Figure~\ref{fig:180_vqe} the Phase~1 standard VQE run on its own. Starting from an initial
energy of $2632.1$, SPSA brings the energy down over 4000–6000 iterations,
with Figure~\ref{fig:180_vqe} showing the full convergence profile over $10000$ iterations.
Nevertheless, when measurement bitstrings are decoded and post-processed, the best independent
sets recovered plateaus at size 14. The quantum state has converged to a basin associated with size-14
solutions, and neither additional iterations nor perturbation of the parameters can escape
it. The energy landscape in the EfficientSU2 parameter space does not present a gradient path from size-14 to
size-15 solutions.

\begin{figure}[!htbp]
\centering
\begin{subfigure}[t]{0.48\linewidth}
  \includegraphics[width=\linewidth]{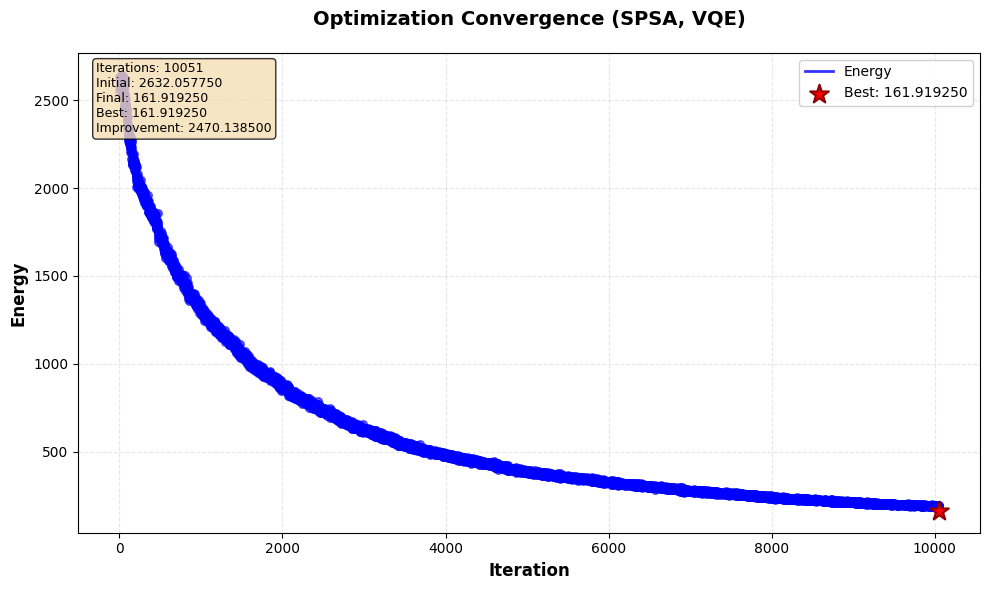}
  \caption{Phase 1: Standard VQE, SPSA, $p=2$}
  \label{fig:180_vqe}
\end{subfigure}\hfill
\begin{subfigure}[t]{0.48\linewidth}
  \includegraphics[width=\linewidth]{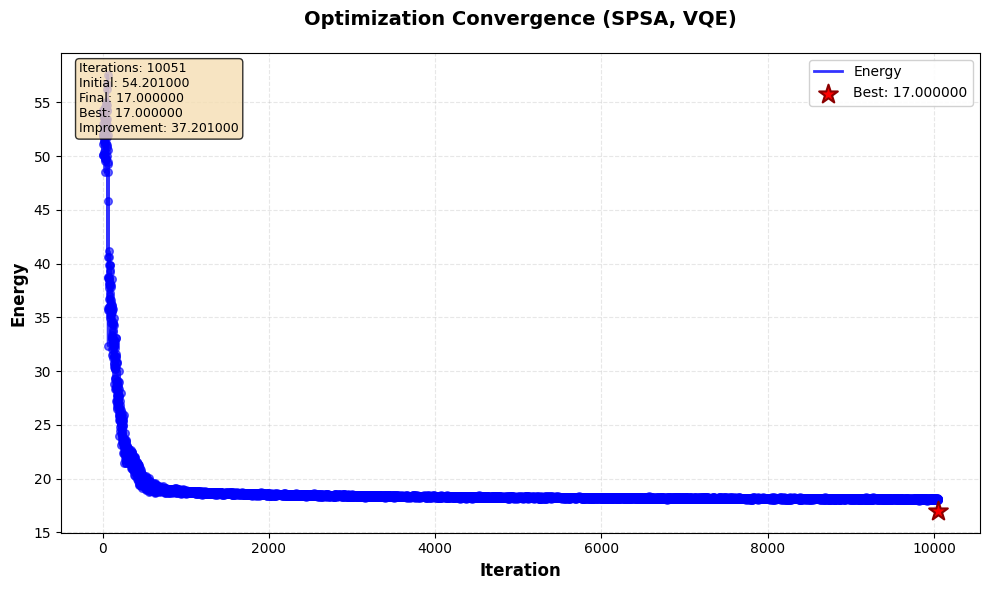}
  \caption{Phase 3: Excitation-preserving VQE with ancilla superposition ($p=2$)}
  \label{fig:180_excp}
\end{subfigure}
\caption{Optimization convergence on the 180-node instance.
  \textbf{Left (Phase 1):} Standard VQE achieves a 94\% energy reduction from
  2632.1 to 161.9 over $10000$ iterations but stalls at independent set size 14.
  \textbf{Right (Phase 3):} The excitation-preserving ansatz initialized from
  the ancilla superposition of four size-14 seed solutions (initial energy
  54.2) converges over 2000--4000 iterations; the size-15 maximum independent
set is subsequently recovered from the measurement distribution by
the bitstring heuristic pipeline.}
\label{fig:180_results}
\end{figure}

Figure~\ref{fig:180_excp} shows the convergence of Phase~3, the
excitation-preserving VQE run ($p=2$, SPSA, 1000--2000 shots) with the
ancilla superposition initialization.
Because the initial state $|\Psi_0\rangle$ already encodes size-14 independent
sets, the initial energy of $54.2$ is a small fraction of the Phase~1
starting point of $2632.1$. This further converges to a low expectation value over 2000--4000 iterations,
and the convergence profile shows continued slow descent rather than premature plateau
behavior. This is expected of the excitation-preserving landscape as the circuit is
confined exactly to the Hamming-weight-14 subspace (the $fSim$ gates conserve Hamming weight).
Every gradient step is directly informative about the relative quality of competing size-14
configurations. The size-15 MIS is subsequently recovered by the stepwise
extension step of the bitstring heuristic. Note that the expectation value energy is
not a direct indicator of MIS presence; it reflects the average quality of the measurement
distribution, guiding the optimizer toward productive regions where MIS-adjacent
configurations are likely.

Post-processing of the Phase~3 measurement distribution yields multiple
independent sets of size~14 and, maximum independent sets of
size~15. Depending on the configuration, between 645 and 2573 unique
independent sets are recovered per run (see Section~\ref{subsec:phase3_sensitivity}),
with 1--2 MIS found per run at the recommended configuration
($p=2$, 2000--4000 iterations, 2000 shots). Different runs with different
seed compositions recover MIS from different structural families of the
optimal solution space. The first MIS confirmed is:
\begin{equation*}
  \text{MIS}_1 = \{4, 12, 42, 50, 56, 71, 95, 112, 121, 127, 139, 146, 159, 168, 177\}.
\end{equation*}
Further runs with varied seed compositions recover additional distinct MIS
from other structural regions of the graph, as detailed in
Section~\ref{subsec:phase3_sensitivity}. Table~\ref{tab:output_sets}
lists the output sets from one representative run. Several recovered
size-14 sets (e.g., sets~2, 4, 5, and 6) contain vertices not present
in any of the input seed states, showing that the evolution reached regions of
the solution space inaccessible to any individual starting point.

\begin{table}[!htbp]
\centering
\caption{Independent sets recovered from Phase~3 measurement distribution after
  post-processing. Sets~1--8 have size 14; the final set has size 15 and is
  the maximum independent set. Sets~2, 4, 5, 6 contain vertices not present in
  any input seed state, indicating exploration beyond the seed
  neighborhood.}
\label{tab:output_sets}
\resizebox{\linewidth}{!}{%
\begin{tabular}{@{}cll@{}}
\toprule
Set & Size & Vertices \\ \midrule
1 & 14 & $\{8, 13, 30, 48, 56, 74, 95, 100, 112, 122, 126, 140, 146, 178\}$ \\
2 & 14 & $\{2, 20, 38, 42, 51, 56, 95, 113, 121, 154, 159, 165, 168, 177\}$ \\
3 & 14 & $\{20, 35, 36, 46, 49, 73, 93, 111, 115, 135, 148, 167, 173, 179\}$ \\
4 & 14 & $\{2, 22, 51, 59, 86, 95, 113, 121, 137, 154, 159, 165, 168, 177\}$ \\
5 & 14 & $\{20, 38, 42, 51, 56, 95, 102, 113, 121, 134, 154, 159, 170, 177\}$ \\
6 & 14 & $\{2, 20, 51, 59, 86, 95, 113, 121, 137, 154, 159, 165, 168, 177\}$ \\
7 & 14 & $\{2, 22, 32, 38, 42, 51, 95, 100, 113, 121, 134, 154, 163, 177\}$ \\
8 & 14 & $\{17, 24, 42, 89, 95, 101, 104, 111, 132, 140, 149, 154, 168, 178\}$ \\
\midrule
\textbf{MIS} & \textbf{15} & $\boldsymbol{\{4, 12, 42, 50, 56, 71, 95, 112, 121, 127, 139, 146, 159, 168, 177\}}$ \\
\bottomrule
\end{tabular}%
}
\end{table}

\subsection{Phase~3 Sensitivity Analysis and Multiple MIS Families}
\label{subsec:phase3_sensitivity}

\textbf{Shot-count sensitivity.}
The number of distinct MIS recovered in a Phase~3 run depends non-monotonically
on the per-evaluation shot count. At 2000 iterations with $p=2$: 1K shots yields
1682 unique independent sets but no MIS; 2K shots yields 2055 unique sets and
1 MIS; 4K shots yields 2082 unique sets but no MIS; 8K shots yields 2573 unique
sets and 1 MIS. This happens because SPSA noise levels vary with shot count and
different shot counts produce different noise levels in the gradient
estimate, leading the optimizer to slightly different parameter basins.
More shots does not guarantee finding the MIS. What matters is whether the
converged parameters place sufficient amplitude on MIS-adjacent configurations.
Increasing the per-iteration shot count $B$ from 50{,}000 to 100{,}000
yields more unique independent sets but does not reliably increase MIS
recovery, confirming that the MIS discovery depends on parameter quality
rather than sampling effort alone.

\textbf{Iteration count and over-optimization.}
Results across Phase~3 configurations peak at approximately 2000--4000
iterations. At 6000 iterations with 2K shots, the number of unique sets
drops sharply to 645 (versus 2055 at 2000 iterations), indicating that
SPSA has drifted into a narrow parameter basin. This over-optimization
pattern also appears in Phase~1, where results peak around 4000--6000
iterations. The practical implication is that an early-stopping criterion
based on solution diversity rather than a fixed iteration budget would be
more appropriate for both phases.

\textbf{Multiple MIS families and iterative enumeration.}
The 180-node instance has multiple structurally distinct MIS families
covering different regions of the graph. The seed composition of the
ancilla superposition determines which family or families are recovered:
when the seeds are concentrated in one structural region, the pipeline
recovers MIS from that region; when the seeds span multiple regions,
multiple MIS families are recovered in a single run.

This seed-family relationship enables an iterative enumeration strategy,
which emerged organically from the experimental observations rather than
by design. In a run using four size-14 seeds that included one set drawn
from a structurally distinct region of the graph, the pipeline returned
two MIS that shared 13--14 vertices with each other:
\begin{align*}
  &\{19,25,63,81,95,97,108,120,133,140,147,155,159,165,178\}, \\
  &\{19,25,43,63,81,95,97,108,120,133,140,155,159,165,178\}.
\end{align*}
The structural similarity of these two solutions suggested they belonged
to the same MIS family and that further members might exist nearby. This
prompted a targeted follow-up experiment. The two known MIS were used
directly as seeds in a new Phase~3 run, together with one additional
size-14 set from the same structural region:
\begin{align*}
  \text{seed}_1 &= \{19,25,43,63,81,95,97,108,120,133,140,155,159,165,178\}, \\
  \text{seed}_2 &= \{19,25,33,51,63,81,95,120,133,140,147,155,165,178\}, \\
  \text{seed}_3 &= \{17,24,42,89,95,101,104,111,132,140,149,154,168,178\}.
\end{align*}
The ancilla superposition encodes all three seeds simultaneously. The
size-15 MIS seed places amplitude in the Hamming-weight-15 subspace,
while the two size-14 seeds place amplitude in the Hamming-weight-14
subspace. The excitation-preserving gates conserve Hamming weight
within each branch independently, so the circuit simultaneously explores
the neighbourhood of the known size-15 solution and the size-14 region
in a single run. The run recovered two MIS, one of
which was a new variant not present among the input seeds:
\begin{align*}
  &\{19,25,63,81,95,97,108,120,133,140,147,155,159,165,178\}, \\
  &\{19,25,43,63,81,95,97,108,120,133,140,155,159,165,178\}.
\end{align*}
This confirms the family hypothesis and demonstrates that the pipeline
can be used iteratively. An observation of structurally similar MIS
solutions provides a clue about the existence of a nearby family, and
targeted reseeding with known solutions from that family enumerates
additional members. This bootstrapping approach converts MIS discovery
into MIS enumeration, and the coverage of the optimal solution space
grows progressively with each targeted run. While classical exact solvers
return one solution, the pipeline presented here returns multiple solutions,
and different runs cover different corners of the MIS space.

\textbf{Certified MIS enumeration.}
To establish a ground truth for the MIS population of the three original
instances (64, 99, and 180 nodes), we used CPLEX's solution pool~\cite{ibm_cplex} with
exhaustive enumeration mode, independently verified by a
Bron--Kerbosch search~\cite{bron1973} on the complement graph.
Both methods agree exactly on counts and sets. The 180-node
instance has precisely 11 distinct MIS, the 64-node instance
has 9, and the 99-node instance has 56. The 11 MIS of the
180-node instance fall into three structurally distinct families.
\textbf{Family~C} comprises a single isolated MIS (MIS$_1$,
the first solution found). \textbf{Family~B} comprises five
tightly clustered MIS sharing a 12-vertex core
$\{19,25,63,95,97,108,120,133,155,159,165,178\}$.
\textbf{Family~A} comprises five more loosely related MIS
sharing only a 3-vertex core $\{23,68,103\}$ in a different
structural region of the graph.

\textbf{Interference circuit MIS recovery.}
The ancilla interference circuit was evaluated on the same 180-node instance using the same
seed compositions and bitstring heuristics pipeline as the stochastic runs above. With
independently optimised SPSA parameters it reached MIS~7--11, all five members of Family~B,
together with MIS~2 from Family~A, six of the 11 certified MIS in total. The stochastic
circuit reached four, spanning Family~C and three of the five members of Family~B. Across
both Phase~3 variants 7 of the 11 certified MIS were reached. Family~A members 3--6 were not
reached in the experiments reported here and may be recoverable with seeds drawn from that
structural region. Table~\ref{tab:all_mis} records which solutions each variant reached, and
Section~\ref{subsec:results_interference} reports a repeated comparison of the two variants
and the variation between runs against which these counts should be read.

\begin{table}[!htbp]
\centering
\small
\caption{All 11 certified maximum independent sets of the 180-node
  instance, verified by CPLEX population pool and Bron--Kerbosch
  enumeration~\cite{ibm_cplex,bron1973}. Sets are grouped by
  structural family. Family~C (MIS~1) is isolated. Family~B
  (MIS~7--11) shares 12-vertex core
  $\{19,25,63,95,97,108,120,133,155,159,165,178\}$.
  Family~A (MIS~2--6) shares 3-vertex core $\{23,68,103\}$.
  $\checkmark$ indicates MIS recovered by the respective Phase~3 variant.}
\label{tab:all_mis}
\resizebox{\linewidth}{!}{%
\begin{tabular}{@{}cclcc@{}}
\toprule
MIS & Family & Vertices & \multicolumn{2}{c}{Recovered} \\
\cmidrule(l){4-5}
 & & & Stochastic & Interference \\ \midrule
1  & C & $\{4,12,42,50,56,71,95,112,121,127,139,146,159,168,177\}$    & $\checkmark$ & \\
\midrule
2  & A & $\{5,6,23,29,35,41,68,73,79,93,103,148,167,170,179\}$         & & $\checkmark$ \\
3  & A & $\{17,23,26,31,61,68,76,88,103,110,120,140,158,169,175\}$     & & \\
4  & A & $\{17,23,29,42,54,68,76,79,89,101,103,129,140,160,170\}$      & & \\
5  & A & $\{17,23,29,42,54,68,76,79,89,101,103,129,160,170,179\}$      & & \\
6  & A & $\{17,23,29,42,54,68,79,89,101,103,129,148,160,170,179\}$     & & \\
\midrule
7  & B & $\{19,25,37,63,95,97,108,118,120,133,147,155,159,165,178\}$   & $\checkmark$ & $\checkmark$ \\
8  & B & $\{19,25,43,63,81,95,97,108,120,133,140,155,159,165,178\}$    & $\checkmark$ & $\checkmark$ \\
9  & B & $\{19,25,43,63,95,97,108,118,120,133,140,155,159,165,178\}$   & & $\checkmark$ \\
10 & B & $\{19,25,63,81,95,97,108,120,133,140,147,155,159,165,178\}$   & $\checkmark$ & $\checkmark$ \\
11 & B & $\{19,25,63,95,97,108,118,120,133,140,147,155,159,165,178\}$  & & $\checkmark$ \\
\bottomrule
\end{tabular}%
}
\end{table}

\subsection{Ablation: Bitstring Heuristics from Single Seed States}
\label{subsec:ablation}

A natural question is whether the bitstring heuristic pipeline alone, applied
directly to each of the original size-14 seed solutions without any quantum
circuit, could discover the size-15 MIS. We conducted this ablation
experiment explicitly: starting from each seed state individually, the full
greedy correction and stepwise extension pipeline was applied with multiple
random restarts, exhausting the classical search from that starting point.

The results are very clear. In the large majority of runs, the heuristic
returns the same size-14 set it started from. The local search landscape
is flat around each seed and the heuristic is trapped. Occasionally the
heuristic escapes to a different size-14 set, but one with substantial
vertex overlap with the starting state. For example, starting from seed $s^{(3)}$ (set~8 in Table~\ref{tab:output_sets}),
the heuristic returns $s^{(3)}$ itself in most runs; in rare cases it escapes to a
nearby set sharing vertices $\{42, 95, 140, 168\}$ with the seed, a local
neighbourhood excursion, not a global escape. Across all the seeds and
all restarts, the size-15 MIS was never found by the bitstring heuristic pipeline
applied to individual seeds.

This result establishes that the barrier at size 14 is an actual local
optimum of the search landscape explored by this pipeline, not a result of an incomplete
heuristic. It is a statement about that local search and not about classical MIS solvers,
which reach the size-15 optimum on this instance (Section~\ref{sec:setup}). The key drawback of single-seed initialization is not the
quality of the seed, but the diversity it can access. Each seed covers one part of the graph, and the bitstring heuristic can only
explore the local neighbourhood of that part. The quantum superposition provides
simultaneous access to all seeds within a single circuit execution. The excitation-preserving
unitary optimizes parameters collectively across all branches, finding
solutions that are better than what any single seed's optimization
produces.
When the ancilla superposition is used, a single run of the pipeline recovers
eight distinct size-14 sets, including several not present among the original
seeds, and the size-$15$ MIS. The classical ablation thus confirms that the
discovery of the MIS is a consequence of the parallel search over
multiple near-optimal initializations, rather than of the heuristic pipeline
alone. However, this leaves open whether the excitation-preserving
circuit structure itself is the key ingredient, or whether the
ancilla superposition initialization specifically is what enables
the size 15 recovery.

A second ablation addresses the role of the ancilla superposition
specifically. To isolate whether the excitation-preserving circuit
structure alone is sufficient, or whether the multi-seed superposition
is essential, we ran the Phase~3 circuit initialized from each seed
state individually, without the ancilla superposition. In all cases, the
circuit returned either the same seed or a size-14 set with
substantial vertex overlap with the starting state, and never
recovered the size-$15$ MIS regardless of the number of iterations
or random restarts of the SPSA parameters.

This result establishes that the excitation-preserving circuit
structure alone is not sufficient and the ancilla superposition
initialization is the essential ingredient. Even under classical
MPS simulation, where the state is exactly representable,
the single-seed circuit cannot escape the local basin that traps
the bitstring heuristic pipeline. The multi-seed quantum superposition breaks
this barrier, and it does so through the collective optimization
mechanism: parameters that are suboptimal for any individual seed
may be collectively optimal across all seeds simultaneously,
accessing regions of the landscape that no single-seed
initialization can reach.

\subsection{Hardware Validation}
\label{subsec:hardware}

The converged VQE configurations for the 64-node, 99-node, and 125-node
(\texttt{c125-9}) instances, and three QAOA configurations for the
64-node instance, were executed on \texttt{ibm\_marrakesh} using the
parameter-transfer protocol described in Section~\ref{subsec:hw_setup}:
noiseless-converged parameters are transferred directly to the transpiled
hardware circuit and sampled without re-optimization, isolating the
effect of hardware noise on the measurement distribution and the
bitstring heuristic pipeline.

\textbf{VQE on hardware.}
Table~\ref{tab:hw_vqe} summarizes the VQE hardware results. For the
64-node instance (hardware depth 65), sampling 20{,}000 shots
($B=20{,}000$, $T=10$, Section~\ref{sec:setup}) recovers 1641 unique
independent sets, including 3 distinct MIS of size 10 and 118 sets of
size 9 (one vertex short of the MIS):
\begin{align*}
  &\{0,3,12,15,21,34,46,51,56,63\}, \quad
   \{0,7,12,17,29,42,48,51,60,63\}, \\
  &\{0,3,12,15,18,45,48,51,60,63\}.
\end{align*}
For the 99-node instance (hardware depth 100), sampling 20{,}000 shots
recovers 217 unique independent sets, including 5 distinct MIS of
size 7 and 56 sets of size 6:
\begin{align*}
  &\{1,2,33,34,85,93,95\}, \quad \{0,16,19,29,39,85,97\}, \\
  &\{0,3,13,69,80,93,97\}, \quad \{0,3,69,80,85,93,97\}, \\
  &\{12,15,60,66,73,88,92\}.
\end{align*}

In both cases, hardware sampling at 40\% of the noiseless shot budget
($B=20{,}000$ vs.\ $B=50{,}000$) recovers approximately half the
noiseless MIS diversity: 3 of up to 6 distinct MIS per run for the
64-node instance, and 5 of up to 10 for the 99-node instance, both
with CVaR (Table~\ref{tab:comparison}). This consistent
$\sim$50\%-at-40\% retention, achieved using parameters optimized
entirely in noiseless simulation with no hardware-aware
re-optimization, indicates that the bitstring heuristic pipeline is
robust to the gate errors, decoherence, and readout errors present on
\texttt{ibm\_marrakesh} at these circuit depths.

\begin{table}[!htbp]
\centering
\caption{VQE hardware validation results on \texttt{ibm\_marrakesh} for all
  three hardware-validated instances ($B=20{,}000$ shots per iteration,
  $T=10$ iterations; hardware depths given in Table~\ref{tab:hw_configs}).
  ``Best (S1)'' is the best single set driven to maximality (Strategy~1);
  ``Best (S2)'' from Strategy~2 (all sets $\to$ maximality)
  (Table~\ref{tab:comparison}). Sets recovered for \texttt{c125-9} on
  hardware is highly variable (1--587 per trial) due to degraded
  measurement diversity under noise.}
\label{tab:hw_vqe}
\begin{tabular}{@{}lccc@{}}
\toprule
Instance & Sets recovered & Best (S1) & Best (S2) \\ \midrule
64-node  & 1641      & 7--8  & 10 (3/run) \\
99-node  & 217       & 4--6  & 7 (5/run)  \\
c125-9   & 1--587    & 33--34 & 34 (4/run) \\
\bottomrule
\end{tabular}
\end{table}

\textbf{c125-9 on hardware.}
The \texttt{c125-9} instance (125 nodes, certified MIS size
34~\cite{koch2025qoblib}) was selected for hardware validation due to
its size being within the available qubit budget ($\leq 156$~qubits).
No ancilla superposition was required: standard VQE with CVaR and SPSA
was sufficient to recover the MIS, consistent with the observation that
the ancilla technique is needed only for instances where single-seed
methods stall. Running standard VQE ($p=2$, CVaR, SPSA) under noiseless
MPS simulation, 148 independent MIS of size~34 were recovered across
post-optimization sampling trials, confirming that the pipeline samples
broadly from the MIS population for this instance. Transferring the
noiseless-converged parameters to \texttt{ibm\_marrakesh} ($B=20{,}000$
shots per iteration, $T=10$ iterations), 4 distinct MIS of size~34 were
recovered. Hardware sampling yields 1--587 unique independent sets per
trial, compared to $\sim$19{,}000 per trial in simulation, reflecting
the degraded measurement diversity at hardware noise levels. The
certified optimal size of 34 is nonetheless preserved exactly,
confirming that the pipeline recovers certified-optimal solutions on
real quantum hardware at the 125-node scale.

\textbf{QAOA on hardware.}
Three QAOA configurations for the 64-node instance were optimised in
noiseless simulation using COBYLA (40--50 iterations), following the same
parameter-transfer protocol: parameters are converged in simulation first,
then transferred directly to the hardware circuit for sampling without
re-optimization on the device. The three configurations are $p=3$ with
$k=4$ (depth 337, CVaR $\alpha=0.2$), $p=3$ with $k=3$
(depth 205, CVaR $\alpha=1.0$), and $p=5$ with $k=3$ (depth
323, CVaR $\alpha=1.0$); hardware depths are given in
Table~\ref{tab:hw_configs}. Each was sampled on \texttt{ibm\_marrakesh}
with $B=20{,}000$ shots per iteration over $T=10$ iterations.
All three recover zero valid independent sets after greedy correction,
regardless of circuit depth or sparsification bandwidth. This is a
stronger result than the noiseless QAOA outcome (0--5 independent sets
per sampling iteration, Section~\ref{sec:results}), and is consistent
with the depth- and sparsification-based arguments of
Sections~\ref{subsec:reorder} and~\ref{subsec:sparse}.

An important caveat, however, is that COBYLA's limited iteration
budget (40--50 iterations, with minimal energy change from the initial
parameters) means the noiseless measurement distribution at the
transferred parameters was already of low quality before hardware
noise is introduced. The zero-recovery outcome on hardware therefore
reflects a \emph{compounding} effect. A measurement distribution that
was already poor under noiseless simulation is further degraded by
gate errors and decoherence across circuit depths of 200--337
two-qubit gates. The result is therefore consistent with the simulator constraints
that limited the QAOA configurations available for parameter
transfer, rather than a definitive statement about QAOA viability
on hardware in general. The relative contributions of optimizer
quality and hardware noise cannot be disentangled from this
experiment alone.

\textbf{QPU usage.}
Across all hardware runs, the IBM Quantum platform reported
approximately 8 seconds of actual QPU execution time per Mode~B
iteration ($\approx 80$ seconds total across $T=10$ iterations per
configuration), against end-to-end wall-clock times of 3.6--37.2
minutes (Table~\ref{tab:hw_configs}) dominated by queueing and
classical post-processing. This is consistent with the bond-dimension
discussion of Section~\ref{subsec:theory_reorder}. The circuits
themselves execute essentially instantaneously on hardware, and the
practical cost of the pipeline lies in the classical optimization and
post-processing stages rather than in quantum execution.

\subsection{200-Node QOBLIB Benchmark Suite}
\label{subsec:200node}

To assess generality beyond the primary 180-node instance, we evaluated the
pipeline on five instances from the QOBLIB benchmark library~\cite{koch2025qoblib}:
\texttt{brock200-1}, \texttt{brock200-2}, \texttt{brock200-3},
\texttt{brock200-4}, and \texttt{c-fat200-1}, all with 200~nodes.
The excitation-preserving ansatz with ancilla superposition and the full
bitstring heuristic pipeline was applied to each instance. The maximum
independent set was recovered for all five instances, demonstrating that
the pipeline generalises reliably to 200-node instances with diverse graph
structures. The \texttt{brock200} family consists of random dense graphs with
certifiably small MIS, while \texttt{c-fat200-1} is a clique-cover graph with
larger MIS. Recovery was verified by comparison against the certified MIS
sizes available in the QOBLIB benchmark.

\subsection{Ancilla Interference: Experimental Evidence}
\label{subsec:results_interference}

The ancilla interference circuit (Section~\ref{subsec:interference}) was
evaluated on the 180-node instance to measure the cross terms it retains and to
test whether they change what the pipeline recovers relative to the stochastic
mixture baseline.

\textbf{Interference diagnostics.}
Two diagnostic tests and a direct measurement of the cross terms were applied to
the post-selected output distribution.

\emph{Test~1 (Stochastic baseline comparison).}
The two arms differ only in the ancilla interference layer. Both use the same four seeds,
the same excitation-preserving ansatz, the same protocol for obtaining SPSA parameters and
the same post-optimisation sample budget, and both pass their measurement outcomes through
the same bitstring heuristics. The stochastic arm traces out the ancilla and samples the
data register, which makes it a mixture over the seed branches. The interference arm applies
the CRZ phase layer and Hadamard on the ancilla and post-selects on $|0\rangle^{\otimes a}$.
Any difference in what the pipeline recovers is therefore attributable to the post-selected
cross terms. Across the configurations tested on the 180-node instance the two arms reached
different members of the certified optimal set, and Table~\ref{tab:all_mis} records which.
The stochastic arm reached four of the eleven and the interference arm six, among them all
five members of Family~B and one member of Family~A that the stochastic arm did not reach.
These are outcomes of single runs per configuration. To test whether such differences follow
from the interference layer or from variation between runs, we repeated the comparison with a
two-seed composition over eight independent parameter initialisations, counting the distinct
near-optimal independent sets recovered. The interference arm averaged 4.1 sets and the
stochastic arm 3.5, a distinct difference but may not be considered significant (paired $t$-test, $p = 0.54$). We
therefore report the family level differences as observations on these runs and not as a
measured effect of interference.

\emph{Test~2 (Phase-sensitivity sweep).}
A scalar offset $\delta$ was added uniformly to all CRZ parameters of seed~0 and swept from
$0$ to $2\pi$ in 24 steps with 20\,000 shots per step, recording both the mean IS size of the
post-selected output and the post-selection rate. For $K=4$ seeds ($p=1$) on the 180-node
instance the mean IS size varies over a range of $\Delta = 0.385$ (peak $9.565$, trough
$9.180$) around a grand mean of $\bar{\mu} = 9.359$. The mean IS size is the
probability-weighted IS size averaged over all post-selected shots, where infeasible
bitstrings contribute zero. Since the excitation-preserving ansatz preserves Hamming
weight~14, every post-selected bitstring has weight~14, valid ones are independent sets of
size~14 and infeasible ones contribute zero to the mean. The grand mean of $9.359$ therefore
implies that about $9.359/14 \approx 67\%$ of post-selected shots are valid IS of size~14.

The observed range is of the same order as the shot noise on the individual points. A test of
the 24 measurements against a constant gives a chi-squared statistic of $26.2$ on 23
degrees of freedom, where
23 is the value expected for a flat curve and 35 would be needed to indicate structure at the
five percent level. The sweep therefore bounds the modulation rather than resolving it, and at
$95\%$ confidence any modulation is below $0.066$ in IS size units, which is $0.47$ percentage
points of the fraction of post-selected shots that are valid independent sets. Sec.~S6 obtains
this limit from the shot noise on the individual sweep points.

A $\delta$-dependence of IS quality would be a direct signature of quantum interference, since
the CRZ phases control the constructive or destructive overlap between seed branches and a
classical stochastic mixture has no such phase degree of freedom. The post-selection rate
stays flat within $\pm 0.9\%$ of $1/K = 0.250$ over the same sweep, with any variation bounded
by $0.9\%$ of $1/K$. The phase-weighted cross terms $\sum_{\mathbf{x}}
e^{i\phi(\mathbf{x})} A_k(\mathbf{x})A_{k'}^*(\mathbf{x})$ integrated over the output remain
small as $\delta$ varies, so the total probability in the ancilla-$|00\rangle$ subspace is
conserved. This does not preclude the cross terms from being significant on individual
bitstrings, and such a redistribution within the valid IS subset would drive a variation in
mean IS size. The sweep is shown in Sec.~S6 of the Supplemental Material together with its
harmonic decomposition.

\textbf{Direct measurement of the cross terms.}
The two tests above bound the effect of the phases on solution quality, which is the quantity
of practical interest, and they do not bound the cross terms themselves. The two questions are
separate. As Eq.~(\ref{eq:cross_sum_zero}) (refer section \ref{subsec:theory_interference}) shows, the cross terms sum to zero over bitstrings,
so an observable averaged over the output responds to the part of them that correlates with
solution quality and not to their size. The size is reached by keeping all ancilla outcomes
rather than discarding the shots that fail post-selection.

Equation~(\ref{eq:cross_measurable}) writes the cross terms as a difference between the
post-selected distribution and the ancilla marginal of the same run, and
Eq.~(\ref{eq:contrast}) gives the normalised contrast that this measures for two seeds. At the
converged parameters, with the optimised phases in place and $400\,000$ shots, $55\%$ of
post-selected shots are valid independent sets of size~14. Across the $2903$ distinct
bitstrings observed, the aggregated statistic lies $0.33$ standard deviations below its
expectation under the hypothesis of no cross terms, so the measured contrast is consistent
with zero and its root mean square value is below $0.011$ at $95\%$ confidence. The statistic
and the null against which it is compared are defined in
Section~\ref{subsec:channel_resolved}, and Sec.~S6 of the Supplemental Material gives its
calibration on this run.

A cross term at a given bitstring requires two seed branches to place amplitude on it. The
ansatz moves excitations along the reordered qubit chain, so two branches meet only when the
vertices where their seeds differ are close on that chain. On this instance they are not. The
closest pairs among the eleven certified optima of Table~\ref{tab:all_mis} differ in one
vertex, MIS~8 and MIS~9 at vertices 81 and 118 and MIS~10 and MIS~11 at the same pair, and
spectral reordering places those two vertices 11 positions apart. Every pair differing in at
most two vertices is separated by at least 10 positions. Reordering minimises the bandwidth of
the Hamiltonian, which is a different objective from placing the distinguishing vertices close
together, so the bound above reflects the geometry of this instance and this reordering.

The mechanism has a second condition that is independent of the first. The branch structure is
created by the controlled preparation, which correlates the ancilla with the data register on
the qubits where the seeds differ. The phase layer acts only on the qubits it targets, which
are the ones the seeds share, so the phases can steer a cross term only on bitstrings that
occupy those qubits, while whether the cross term exists at all is set by the overlap of the
branches. The bound reported above is limited by the overlap of the branches rather than by the reach of
the phase layer.

Enlarging the overlap requires either a permutation chosen to bring the
distinguishing vertices closer or having a circuit with greater depth. The second lies beyond the
exact simulation used here, since each added layer extends the transport reach by one step while
doubling the bond dimension of the seed superposition. The effort of reaching the overlap
grows exponentially. A quantum processing unit pays for
that in terms of circuit depth and noise rather than in representation size (that classical simulators pay).
It also pays for the long range gates of the phase layer in terms of routing rather than in the
cost that simulators pay in terms of high bond dimension. Interference over bitstrings is therefore a mechanism
that hardware can measure at depths simulation cannot reach, and recording the ancilla outcomes costs
nothing beyond a normal run. The measurement reported here is one such instrument.

% \textbf{Fewer trials to MIS.}
% The ancilla interference circuit was compared against the stochastic mixture
% baseline (same seeds, same post-optimization shot budget, but without the
% ancilla interference post-selection layer — the ancilla is traced out and
% only the data register is sampled) on the 180-node instance. The
% interference circuit required fewer sampling trials to first recover a MIS
% than the stochastic mixture baseline. This is consistent with the constructive
% amplification predicted by Proposition~\ref{prop:constructive}, which holds that the
% interference mechanism concentrates probability mass on configurations
% jointly supported by all seeds, which are also the configurations most
% amenable to MIS recovery by the stepwise extension heuristic.

% The test was performed for $p=1$. We note that for larger
% $p$, the circuit depth and bond dimension increase, and the
% comparison was not exhaustively repeated across all settings. The evidence
% presented is therefore indicative rather than definitive regarding the general
% superiority of interference over stochastic mixture as a sampling strategy.

\subsection{400-Node Instance (brock400-1)}
\label{subsec:400node}

The largest instance considered is \texttt{brock400-1}, taken from the QOBLIB
independent set library~\cite{koch2025qoblib} with certified MIS~27. The instance file
used in every experiment reported here has 400 nodes and 20077 edges, giving an average
degree of 100.39. Its average degree is higher than that of the 180-node system (100.39
against 58.63), so each qubit carries more penalty terms, and this has direct consequences
for the required circuit penalty $P$ and the MPS bond dimension. Measured as a fraction of
all possible edges the 180-node instance is the denser of the two ($0.33$ against $0.25$),
and it is the higher average degree rather than the edge fraction that drives the
behaviour reported below.

\textbf{Phase~1: standard VQE baseline.}
As a baseline, standard VQE with SPSA ($p=1$--$2$) was run from the zero state,
recovering 30--40 distinct independent sets per sampling trial. Strategy~1 reaches a best IS size of 23 in most runs; Strategy~2
reaches IS size 24 in some runs, but such occurrences
are infrequent. These Phase~1 seeds serve as the initialisation source for
Phase~3. The contrast with Phase~3 is striking. The excitation-preserving
ansatz with ancilla superposition not only recovers IS size 24 with
higher frequency, but also crosses the IS=25 barrier that
Phase~1 cannot reach.

\textbf{Penalty $P$-scaling for Phase~1 (VQE+CVaR).}
The density of \texttt{brock400-1} requires a higher penalty coefficient $P$ for
the standard VQE+CVaR phase than is needed for the 180-node system:

\begin{itemize}[noitemsep]
  \item $P=2$: no meaningful progress; the penalty coefficient is too weak to enforce independence on this dense graph.
  \item $P=4$: best IS size 22--23, with very few solutions reaching 24.
  \item $P=6$: best IS size reaches 24 with higher frequency.
\end{itemize}

This $P$-scaling behaviour reflects the harder optimization landscape of a
dense graph. For \texttt{brock400-1} with average degree 100.39, a weak
penalty coefficient allows the optimizer to achieve low energy through
configurations with many constraint violations, obscuring the gradient
signal toward feasible high-quality IS. A higher $P$ sharpens the energy
gap between feasible and infeasible configurations, making the optimizer
reliably track toward valid IS.
The excitation-preserving ansatz (Phase~3), by contrast, operates in the
fixed Hamming-weight subspace and uses $P=2$ throughout; feasibility is
enforced structurally by the ansatz rather than by the penalty, so no
$P$-scaling is needed.

\textbf{Ansatz layers $p$ scaling and the IS=25 barrier.}
Using size-24 IS seeds obtained from the Phase~1, $P=6$ runs above, the
excitation-preserving ansatz with ancilla superposition ($P=2$) was applied.
The number of ansatz layers $p$ is a key hyperparameter: more layers increase expressiveness but raise simulation cost and require more iterations to converge.

\begin{itemize}[noitemsep]
  \item $p=1$: recovers IS size~24 but does not reach 25.
  \item $p=2$: recovers IS size~24 with more distinct solutions.
  \item $p=3$: first recovery of IS size~25 (first time size~25 is reached).
\end{itemize}

All three comparisons use the same total post-optimization shot budget to ensure
a fair comparison. The improvement from $p=1$ to $p=3$
mirrors the pattern observed on the 180-node instance, where $p=1$
fails to find the MIS while $p=2$ succeeds. In both cases, more ansatz layers provide greater expressiveness to
distinguish near-optimal configurations and escape local energy basins.
The cost of this depth is measurable. On our single-machine MPS setup the median time per
optimizer step rises from 14.5 seconds at $p=2$ to 27.0 seconds at $p=3$, and the
$p=3$ run that first reached IS=25 took 12.3 hours over 1500 iterations, and its
convergence is shown in Figure~\ref{fig:400_nreps3}. Higher $p$
($\geq 4$) were also attempted, and at that rate each iteration takes long enough to put
the total wall-clock time out of reach without hardware.

\textbf{Attempting IS=26.}
Using the IS=25 solutions as seeds, the pipeline was re-run targeting IS=26.
Despite multiple runs with varied ancilla configurations, $p$,
and optimizer settings, IS=26 was not recovered. The pipeline appears to
plateau at IS=25, achieving 92.6\% of the certified optimum (27). The
cause of this plateau is not fully understood, but is consistent with the
circuit expressibility ceiling discussed in the context of bond dimension
scaling: a circuit of bounded MPS rank cannot represent all Hamming-weight-25
configurations that are adjacent to Hamming-weight-26 configurations.

\textbf{Swap-based seeds and (1,1)/(2,1) exchanges.}
Additional experiments used seed states derived by (1,1)- and (2,1)-swap
operations on the IS=25 solutions, creating diverse neighboring solutions
that differ by 2--3 vertices. These seeds produced more IS=25 solutions per
run but did not extend the barrier to IS=26. The diversity of the IS=25
population can be expanded, but the circuit cannot yet cross to IS=26.

\textbf{Ancilla interference for 400-node.}
The ancilla interference circuit (Section~\ref{subsec:interference}) was
tested with $K=2$ and $K=3$ seeds on the 400-node system. For $K \leq 3$,
the MPS bond dimension remains tractable (Section~\ref{subsec:theory_interference}).
For $K=4$ ($\chi_{\max} \leq 4 \cdot 2^p$), the bond dimension that was tractable
for the 180-node system becomes computationally prohibitive at 400 qubits, since
the per-iteration MPS contraction cost scales with both $\chi_{\max}$ and system size,
making $K \leq 3$ the practical limit. The cost of one added branch is already visible
between the two tractable settings. Holding $p=1$, the iteration count at 1000 and the
shot budget fixed, the time per optimizer step goes from 14.7 seconds at $K=2$ to
38.8 seconds at $K=3$, so a further branch would leave the run beyond what we can
schedule classically. At $K=2$ with IS=25 seeds
($a=1$ ancilla qubit, confirmed $\chi_\mathrm{max}=4$), the interference circuit
remained tractable and was used with the configurations above.

\textbf{Hardware requirement.}
The 400-node system with ancilla overhead requires approximately 402 qubits
(400 data + 1--2 ancilla). No current publicly accessible quantum hardware
provides this qubit count (\texttt{ibm\_marrakesh} has 156 qubits). The
simulation results reported here therefore represent the practical ceiling of
classical MPS-based optimization for this instance, and hardware execution is expected to access deeper circuit regimes (higher $p$)
that are currently outside our exact MPS budget, since the MPS bond
dimension grows as $2^p$ per layer. The $p$-scaling evidence above provides
empirical motivation for this expectation, since $p=3$ reaches IS=25 while
$p=1,2$ plateau at IS=24. Higher $p$ ($\geq 4$) becomes impractical
classically and remains a target on hardware, where circuit depth
does not incur exponential simulation cost.

\begin{figure}[!htbp]
\centering
\includegraphics[width=0.85\linewidth]{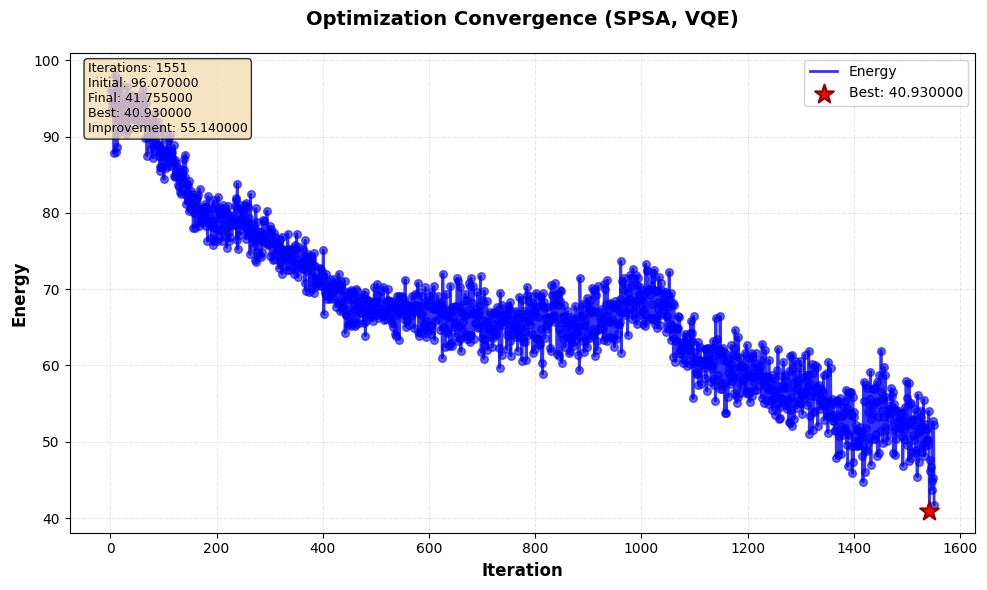}
\caption{Energy convergence of the excitation-preserving ansatz
  ($p=3$, SPSA, 1000 shots per iteration, 1500 iterations)
  on the 400-node \texttt{brock400-1} instance. Energy decreases from 96.07
  to a best of 40.93 over the run. Post-optimization sampling of this
  converged circuit recovers IS size~25 for the first time;
  runs with $p\leq 2$ do not reach IS=25 from the same seeds.
  The certified MIS is 27.}
\label{fig:400_nreps3}
\end{figure}
\subsection{Scope and Limitations}
\label{subsec:limitations}

A few limitations on these claims are worth stating directly.

The ablation of Section~\ref{subsec:ablation} compares the multi-seed superposition against
each seed on its own. It does not compare against a classical population method such as a
genetic or memetic algorithm that recombines the four seeds directly, which is the natural
classical counterpart of what the superposition does, and we have not run that comparison.
The ablation therefore supports the weaker statement that collective optimisation over
several seeds is necessary within this pipeline, and it does not show that a population based
classical search would fail on this instance.

The circuit does not itself produce the size-15 solution. Since the $fSim$ gates conserve
Hamming weight, a circuit initialised on weight-14 seeds stays in the weight-14 subspace and
the final vertex comes from the stepwise extension, as Figure~\ref{fig:ancilla_sup}(d)
indicates. The quantum part is responsible for the quality and the diversity of the size-14
pool.

No runtime or advantage claim is made anywhere in this work. CPLEX certifies the 180-node
optimum in 52 to 54 seconds on a single thread and ReduMIS recovers it within 60 seconds
(Section~\ref{sec:setup}), and the circuits used here are exactly MPS simulable at these
depths and seed counts ($\chi_{\max} \leq m \cdot 2^p$).

The bound on the interference contrast applies at the seed composition and depth used here. A
test in the regime where the branches do overlap needs seeds differing in one or two vertices,
and since the variational parameters are converged for a given seed composition such a test
requires a separate optimisation for each pair rather than a transfer of existing parameters.
We have not carried that out.

\subsection{Comparative Summary}

Table~\ref{tab:comparison} summarizes performance across all instances and
configurations, reporting results for both Strategy~1 (best single set driven
to maximality) and Strategy~2 (all recovered sets driven to maximality). The
table makes several patterns clear. First, Strategy~2 consistently outperforms
Strategy~1 by 1--4 vertices, confirming that solution diversity is as important
as optimization quality. Second, SPSA consistently produces more independent
sets than COBYLA for the same configuration, which directly enables Strategy~2
to be more effective. Third, QAOA recovers far fewer independent sets than VQE
(0--5 vs. 700--800 for the 64-node instance), severely limiting its Strategy~2
performance. Fourth, for instances up to 180 nodes no tangible benefit was found beyond
$p=2$--$3$; the 400-node system requires penalty $P=4$--$6$ for standard VQE+CVaR due
to its higher graph density; the excitation-preserving Phase~3 ansatz
retains $P=2$ (no penalty needed in the fixed Hamming-weight subspace).

\begin{table}[!htbp]
\centering
\caption{Performance summary. ``Sets recovered'' is the approximate number of
  distinct independent sets per post-optimization sampling iteration
  ($B=50{,}000$ shots per iteration, $T=10$ iterations; Section~\ref{sec:setup}).
  ``Best (S1)'' is the best result from Strategy~1 (best single set $\to$
  maximality); ``Best (S2)'' from Strategy~2 (all sets $\to$ maximality).
  MIS is the known maximum independent set size. For Phase~3, the wide
  range in sets recovered (645--2573) reflects sensitivity to shot count
  and iteration budget; see Section~\ref{subsec:phase3_sensitivity}.}
\label{tab:comparison}
\small
\begin{tabular}{@{}llccccc@{}}
\toprule
Instance & Configuration & $p$ & Sets rec. & S1 & S2 & MIS \\ \midrule
64-node  & VQE, COBYLA (CVaR) & 2 & 700--1200 & 8  & 10 & 10 \\
64-node  & VQE, SPSA (CVaR)  & 2 & 700--800  & 9  & 10 & 10 \\
64-node  & VQE, SPSA (CVaR)  & 3 & 700--800  & 9  & 10 & 10 \\
64-node  & QAOA, COBYLA & 2--5 & 0--5  & 8--9 & 9 & 10 \\
64-node  & QAOA, SPSA  & 2   & 0--3   & 9  & 9  & 10 \\
99-node  & VQE, COBYLA & 2 & 60--100   & 6  & 7  & 7  \\
99-node  & VQE, SPSA (CVaR)  & 2 & 150--200  & 6  & 7  & 7  \\
180-node & VQE, SPSA (Phase 1) & 2 & 400--600  & 9--10 & 13--14 & 15 \\
180-node & EXCP + Ancilla (Phase 3) & 2 & 645--2573 & 14 & 15 & 15 \\
\midrule
brock200-1--4, c-fat200-1 & EXCP + Ancilla & 2 & --- & --- & MIS & MIS \\
c125-9  & VQE, SPSA (CVaR) & 2 & 40000 & 22--24 & 34 & 34 \\
\midrule
brock400-1 & VQE, SPSA (Phase 1) & 1--2 & 3000-4000 & 23 & 24 & 27 \\
brock400-1 & EXCP + Ancilla (Phase 3) & 2--3 & 8000-9000 & 25 & 25 & 27 \\
\bottomrule
\end{tabular}
\end{table}

The results collectively support several conclusions. SPSA is the optimizer
of choice for large-scale variational circuits. Its stochastic gradient
estimates prevent premature convergence, and its broader exploration
of the measurement distribution yields more diverse independent
sets. Solution diversity is a key bottleneck. QAOA's failure on the 64-node
instance is not primarily due to worse convergence, but to the negligible
number of independent sets it recovers per run. The two-strategy post-processing
framework (Strategy~1 for quick estimates, Strategy~2 for thorough search)
provides a practical and flexible tool for extracting maximum value from any
variational circuit. The 180-node result demonstrates that the
ancilla superposition technique can break barriers that neither the bitstring
heuristic pipeline alone nor standard VQE can resolve, as confirmed by the ablation study
(Section~\ref{subsec:ablation}). The five 200-node QOBLIB instances
demonstrate that these results generalise. The pipeline recovers the MIS
without requiring larger seed populations or higher circuit depth than the
180-node case. For the 400-node instance, both standard VQE Phase~1 and the
excitation-preserving Phase~3 show that the pipeline continues to make
meaningful progress at scale. Phase~1 already reaches IS=23--24 from a
standard initialisation, and Phase~3 with $p=3$ pushes this
to IS=25. The gap to the certified optimum (IS=27) highlights that graph
density introduces a new bottleneck, requiring higher circuit depth and
gradient quality, that is expected to be resolved with hardware execution.

%=============================================================================
\section{Theoretical Analysis}
\label{sec:theory}
%=============================================================================

The experimental results of Section~\ref{sec:results} raise several analytical questions
that this section tries to address. The diagonal structure of the MIS Hamiltonian and the shallow
circuit depths together keep the gradient signal tractable, avoiding barren plateau conditions
that affects deeper unconstrained circuits. The MPS bond dimension bound marks the boundary
where classical simulation hits the limit and hardware execution becomes necessary. The
bitstring heuristic pipeline is effective primarily because near-optimal seeds sit close to
the MIS in configuration space, and the conditions for this are provided. Finally, the
ancilla post-selection produces a distribution that is different from a classical
mixture over the same seeds, and the interference mechanism gives the detail.

\subsection{Gradient Behavior of Shallow Circuits on Diagonal Hamiltonians}
\label{subsec:theory_gradients}

Standard parameterized quantum circuits face a training obstacle known as the barren plateau.
McClean et al.~\cite{mcclean2018barren} showed that for a circuit drawn from a unitary
2-design, the gradient variance vanishes exponentially in system size,
\begin{equation}
  \mathrm{Var}[\partial_{\theta_k} E] \in O(\frac{1}{2^{n}}),
\end{equation}
making gradient-based optimization exponentially hard for large $n$. Cerezo
et al.~\cite{cerezo2021} established that the severity depends mostly on how
local the cost function is.

\begin{proposition}[Cerezo et al.\ 2021~\cite{cerezo2021}]
\label{prop:cerezo}
For a cost function $C = \sum_i c_i$ where each $c_i$ acts on at most $\ell$ qubits, the gradient variance for a shallow circuit (depth $O(\log n)$) satisfies
\begin{equation}
  \mathrm{Var}[\partial_{\theta_k} C] \;\geq\; \frac{c}{n \cdot 4^\ell},
  \label{eq:cerezo_bound}
\end{equation}
for a circuit-dependent constant $c > 0$.
\end{proposition}

For fixed locality $\ell$, the factor $4^\ell$ is a constant, so the bound~(\ref{eq:cerezo_bound}) decays only as $1/n$, not exponentially in $n$. The guarantee applies in the shallow regime; for deeper circuits it weakens and can eventually become exponential even for local cost functions.

The MIS Hamiltonian (\ref{eq:ising2}) consists entirely of $Z_i$ and $Z_i Z_j$ Pauli terms, so every cost function term acts on at most $\ell = 2$ qubits. This follows directly from the QUBO formulation and holds regardless of how vertices are reordered. Substituting into~(\ref{eq:cerezo_bound}):
\begin{equation}
  \mathrm{Var}[\partial_{\theta_k} C_{\mathrm{MIS}}] \;\geq\; \frac{c}{16\,n}.
  \label{eq:mis_grad_bound}
\end{equation}
The gradient decays only as $1/n$. This is why the sparsified, reordered circuits remain trainable at the depths used in this work ($p \leq 3$), where the shallow-circuit condition holds. A broader comparison of this regime with non-diagonal Hamiltonians such as electronic structure is given in Sec.~S5 of the Supplemental Material.

The excitation-preserving ansatz obtains a further gradient advantage through the constrained-subspace structure. Larocca et al.~\cite{larocca2022} showed that gradient variance is inversely proportional to the dimension of the dynamical Lie algebra (DLA) of the ansatz. For the excitation-preserving circuit, this containment follows from construction: each $fSim(\theta,\varphi)$ gate swaps at most two adjacent excitations, so the generator algebra closes within the algebra of Hamming-weight-preserving operators on $\mathcal{H}_k$. The DLA is therefore a subalgebra of $\mathfrak{su}(\tbinom{n}{k})$, and the dimension that
controls the bound is $\tbinom{n}{k}$ rather than $2^n$,
\begin{equation}
  \mathrm{Var}\!\left[\partial_\theta E\big|_{\mathcal{H}_k}\right]
  \;\in\; \Omega\!\left(\frac{1}{\tbinom{n}{k}^2}\right).
  \label{eq:dla_bound}
\end{equation}
We do not compare Eq.~(\ref{eq:dla_bound}) quantitatively with the $O(2^{-n})$ result for
unitary 2-designs. The two are lower bounds established in different frameworks, and their
ratio is not an improvement factor. The claim we rely on is the weaker structural one, that
the gradient is spread over the Hamming-weight-$k$ subspace rather than over the full Hilbert
space. At finite depth $p$ the DLA has not saturated $\mathfrak{su}(\tbinom{n}{k})$, so
Eq.~(\ref{eq:dla_bound}) is loose at the depths used here.

Adding circuit layers within the constrained subspace does not worsen barren plateaus in the
way deeper unconstrained circuits do. Standard EfficientSU2 circuits approach a unitary 2-design
over the full $2^n$-dimensional space as depth increases, spreading the gradient across an exponentially
larger landscape~\cite{holmes2022}. The $fSim(\theta,\varphi)$ gates propagate excitations
further across the graph with each additional layer but always remain within $\mathcal{H}_k$,
so the DLA grows only within the $\binom{n}{k}$-dimensional space. Each extra layer adds
expressibility within the region that the MIS problem cares about, without the exponential
gradient dilution of an unconstrained circuit. This is the likely reason that $p=1$ yields no MIS on
the 180-node instance, $p=2$ succeeds, and $p=3$ improves further on the 400-node instance.
The improvement with depth reflects actual gain in circuit expressibility rather than noise or
statistical fluctuation.

\begin{remark}
Gradient-free optimizers such as COBYLA and SPSA do not bypass barren plateaus. Arrasmith et al.~\cite{arrasmith2021} showed that when gradient variance vanishes exponentially, the cost function landscape is exponentially flat and gradient-free methods also fail to make progress. Avoiding the plateau is therefore essential regardless of optimizer choice. For the shallow circuits in this work, the locality bound~(\ref{eq:mis_grad_bound}) ensures the landscape retains structure that both gradient-based and gradient-free optimizers can exploit.
\end{remark}

\subsection{Reordering and Sparsification}
\label{subsec:theory_reorder}

\textbf{Circuit depth bound.}
Let $\pi: V \to \{0,\ldots,n-1\}$ denote the vertex-to-qubit permutation after reordering, and let $B = \max_{(i,j)\in E}|\pi(i)-\pi(j)|$ denote the resulting bandwidth of the coupling matrix. After spectral reordering to bandwidth $B$, the QAOA cost unitary can be implemented in at most $B$ layers of parallel two-qubit gates, giving a total circuit depth of $O(B \cdot p)$ rather than $O(|E| \cdot p)$ for the unsparsified circuit (full derivation in Section~\ref{subsec:depth}). Since the Fiedler vector minimizes the weighted sum of squared index distances $\sum_{(i,j)\in E} |J_{ij}|(\pi(i)-\pi(j))^2$, it provides a good polynomial-time approximation to bandwidth reduction on dense graphs, justifying its use where $B \ll |E|$.

\textbf{Sparsification gradient bias.}
Let $J^{(k)}_{\mathrm{sparse}} = \{J_{ij} : |\pi(i)-\pi(j)| \leq k\}$ denote the couplings retained at bandwidth $k$, and let $J_{\mathrm{full}}$ denote the complete coupling set. When the circuit uses $J^{(k)}_{\mathrm{sparse}}$ but energy is evaluated on $J_{\mathrm{full}}$, the gradient bias is bounded by $(1 - \mathrm{coverage}(k)) \cdot |E| \cdot P/4$ times the maximum $ZZ$ correlation gradient norm, where
\begin{equation}
  \mathrm{coverage}(k) = \frac{\sum_{|\pi(i)-\pi(j)|\leq k}|J_{ij}|}
  {\sum_{(i,j)\in E}|J_{ij}|}
\end{equation}
is the fraction of total coupling weight retained at bandwidth $k$. For the 64-node instance at $k=3$, coverage is 21\% (Table~\ref{tab:reorder_compare}), meaning 79\% of coupling weight is absent from the circuit. The asymmetric evaluation strategy (sparse circuit, full Hamiltonian for energy) ensures the optimizer still receives accurate energy feedback. Convergence to $\epsilon$-optimal on the sparse landscape yields at most $\epsilon + 2\delta$ error on the full landscape by a standard perturbation argument, where $\delta$ is the gradient bias magnitude (derivation in Section~\ref{subsec:sparse}).

\textbf{MPS tractability.}
After reordering to bandwidth $k$, the causal cone of any qubit in the QAOA circuit has width $O(k \cdot p)$, bounding the entanglement entropy to $S \leq O(k \cdot p)$ and the MPS bond dimension to $\chi \leq 2^{O(k \cdot p)}$. This grows exponentially in $k$ and $p$ but remains polynomial in $n$ for fixed $k$ and $p$. For the excitation-preserving circuit with $m$ seed states and $p$ ansatz layers, the bond dimension satisfies $\chi_{\max} \leq m \cdot 2^p$ (Proposition~\ref{prop:bond_dim}). For the 99- and 180-node instances, even $k=2$ strains the MPS representation; only $k=1$ is tractable, which reduces the QAOA cost unitary to nearest-neighbour gates and removes the motivation for QAOA on these instances.

\subsection{Bitstring Heuristics}
\label{subsec:theory_heuristics}

The extension success condition identifies when a near-optimal starting point can be driven to the MIS by the stepwise extension of Algorithm~\ref{alg:stepwise}.

\begin{proposition}[Extension success]
\label{prop:extension}
Let $S_0$ be a feasible independent set of size $|S_0| = k-1$, and let $S^*$ be a maximum independent set of size $k$. If there exists a vertex $v^* \in S^* \setminus S_0$ such that $N(v^*) \cap S_0 = \emptyset$, where $N(v^*)$ denotes the neighbourhood of $v^*$, then there exists a selection ordering in the stepwise extension from $S_0$ that reaches a maximal set of size at least $k$.
\end{proposition}

The condition $N(v^*) \cap S_0 = \emptyset$ says that the vertex missing from $S_0$ but present in $S^*$ must be compatible with all current members of $S_0$. The minimum-degree selection rule of Algorithm~\ref{alg:stepwise} provides a practical heuristic that tends to find such orderings, though it does not guarantee selecting $v^*$ at each step.

For the 180-node instance, this result explains why diverse seeds matter. Each seed blocks
a different subset of vertices, and the ancilla superposition searches over all seed
neighbourhoods simultaneously. If the MIS-completing vertex is blocked by one seed but
compatible with another, the parallel search can reach it where a single seed cannot.
The ablation study of Section~\ref{subsec:ablation} tests this directly, and finds that
neither the bitstring heuristics nor the excitation-preserving circuit reaches the MIS
from any seed on its own.

\subsection{Solution Diversity: VQE vs.\ QAOA}
\label{subsec:theory_diversity}

VQE recovers hundreds to thousands of distinct independent sets per post-optimization sampling iteration ($B=50{,}000$ shots per iteration, $T=10$ iterations, Section~\ref{sec:setup}) (up to 700--800 for the 64-node instance and up to $10$ distinct MIS for the 99-node instance with CVaR); QAOA recovers only 0--5. We attribute the difference to how sharply each algorithm's measurement distribution concentrates after optimization, and this subsection sets out that account.

Let $|\mathrm{supp}|_{\mathrm{eff}}$ denote the \emph{effective support} of the output distribution, defined through the participation ratio
\begin{equation}
  |\mathrm{supp}|_{\mathrm{eff}} = \Big(\sum_{\mathbf{x}} P(\mathbf{x})^2\Big)^{-1},
  \label{eq:eff_supp}
\end{equation}
which counts the number of bitstrings that carry appreciable probability. Both circuits have full support in the strict sense, since a generic parameterised unitary leaves no amplitude exactly zero, so the strict count of nonzero amplitudes does not distinguish them. What differs is concentration. A small effective support means measurements repeatedly return the same few bitstrings, and a large one means many different bitstrings can appear.

QAOA starts from the uniform superposition $|+\rangle^{\otimes n}$, where the effective support is $2^n$. As the cost and mixer unitaries concentrate probability onto low-energy states the effective support contracts, and for the shallow, sparsified configurations studied here ($k \leq 3$, $p \leq 5$) the measured distributions are dominated by a handful of bitstrings. VQE starts from $|0\rangle^{\otimes n}$, where the effective support is 1. The parameterized rotations then spread probability across the Hilbert space, and the optimized circuit produces a distribution whose effective support is far larger than the number of shots drawn from it. The contrast is therefore between a distribution that has contracted onto few configurations and one that stays broad.

When drawing $N$ measurement shots from a distribution with effective support $|\mathrm{supp}|_{\mathrm{eff}}$, the expected number of distinct outcomes $D$ follows the coupon-collector formula~\cite{erdos1961coupon},
\begin{equation}
D(N) \approx |\mathrm{supp}|_{\mathrm{eff}} \left(1 - e^{-N/|\mathrm{supp}|_{\mathrm{eff}}}\right).
\label{eq:coupon}
\end{equation}
For VQE, where $N \ll |\mathrm{supp}|_{\mathrm{eff}}$, this gives $D \approx N$, so almost every shot yields a new bitstring. For QAOA, where $N \gg |\mathrm{supp}|_{\mathrm{eff}}$, it gives $D \approx |\mathrm{supp}|_{\mathrm{eff}}$, so the distribution saturates and few distinct outcomes appear.

This account is consistent with the VQE advantage we observe for MIS. The maximality extension strategy requires diverse starting points, and the broad VQE distribution supplies them while the contracted QAOA distribution does not. The probability that at least one starting point leads to the MIS increases with the number of distinct candidates, which is why breadth matters in this pipeline. Equation~(\ref{eq:coupon}) treats sampling as uniform over the effective support, which the true non-uniform distribution only approximates, and $D(N)$ is bounded above by the uniform expression in general. What this subsection offers is therefore a structural explanation of an empirical observation and not a derivation of the diversity gap, and we do not claim it as a bound.

The use of CVaR further amplifies this diversity advantage. By focusing the gradient signal on the best-performing tail of the measurement distribution, CVaR steers the optimizer toward parameter basins where the distribution is both low-energy and diverse, recovering more distinct near-optimal solutions per run than the standard expectation value objective.

\subsection{Ancilla Superposition Initialization}
\label{subsec:theory_ancilla}

The ancilla superposition pipeline (Figure~\ref{fig:ancilla_sup}) has two analytically concrete advantages over single-state initialization.

\textbf{Initial energy advantage.}
The superposition state $|\Psi_0\rangle = \frac{1}{\sqrt{m}}\sum_{j=1}^m |s^{(j)}\rangle|j\rangle_{\mathrm{anc}}$ has energy expectation value
\begin{equation}
  E_0 = \langle\Psi_0|H_{\mathrm{MIS}}|\Psi_0\rangle
      = \frac{1}{m}\sum_{j=1}^m E(s^{(j)}),
\end{equation}
the arithmetic mean of the seed state energies $E(s^{(j)})$. Since the seeds are near-optimal independent sets, their energies are far below that of the zero state $|0\rangle^{\otimes n}$ used in standard VQE, which encodes no structural information about the problem. The superposition initialization therefore places the optimizer in a productive region of the landscape from the first iteration, requiring far fewer optimization steps to reach configurations where the MIS is recoverable from the measurement distribution.

\textbf{Reachability and landscape structure.}
The excitation-preserving ansatz generates unitaries within the $\binom{n}{k}$-dimensional Hamming-weight-$k$ subspace $\mathcal{H}_k$, where $k$ denotes the number of qubits in the $|1\rangle$ state. For sufficient depth $p$, this ansatz is universal within $\mathcal{H}_k$~\cite{arrazola2022universal}, guaranteeing that all size-$k$ independent sets are reachable by the circuit. The size-15 MIS has Hamming weight 15 and lies strictly outside $\mathcal{H}_{14}$: the $fSim$ gates conserve Hamming weight exactly, so the circuit itself cannot reach weight-15 states. The MIS is instead discovered by the stepwise extension step (Section~\ref{subsubsec:stepwise}), which adds one compatible vertex to a recovered size-14 independent set. The role of the circuit is to concentrate probability mass on size-14 configurations that have a compatible MIS-completing vertex available.

\begin{proposition}[Constrained landscape smoothness]
\label{prop:landscape}
The restriction of $H_{\mathrm{MIS}}$ to $\mathcal{H}_k$ eliminates the dominant $O(n^2)$ variance arising from Hamming-weight variation in the full Hilbert space. The residual within-subspace variance is $O(k \cdot \Delta d + k^2)$, where $\Delta d = d_{\max} - d_{\min}$ is the degree spread of the graph.
\end{proposition}

\begin{proof}
The MIS Ising Hamiltonian has two types of terms, single-qubit fields $h_i z_i$ and two-qubit
couplings $J_{ij} z_i z_j$. Within the subspace $\mathcal{H}_k$, every state has exactly $k$
qubits in $|1\rangle$ and $n-k$ qubits in $|0\rangle$. In the $Z$ basis, $z_i = +1$ encodes
vertex $i$ not selected and $z_i = -1$ encodes vertex $i$ selected.

\textbf{Single-qubit terms are constant on $\mathcal{H}_k$.}
Because every state in $\mathcal{H}_k$ has exactly $k$ qubits in $|1\rangle$, the sum
$\sum_i z_i$ equals $(n-k)(+1) + k(-1) = n - 2k$ for every computational basis state in
$\mathcal{H}_k$, and by linearity for every superposition within it. The single-qubit
contribution to the energy is therefore the same constant $\sum_i h_i \cdot \text{const}$ for
all states in the subspace and contributes nothing to the energy differences between states.
All variation in energy across the landscape is determined solely by the two-qubit coupling
terms $\sum_{(i,j)\in E} J_{ij} z_i z_j$.

\textbf{Energy differences are small between nearby states.}
Two states in $\mathcal{H}_k$ that are nearest neighbours in Hamming distance differ by a
single swap, vertex $i$ deselected and vertex $j$ selected, or vice versa. The energy change
from such a swap is
\begin{equation}
  \Delta E = 2J_{ij}(z_j - z_i) + 2\!\sum_{l \neq i,j}(J_{jl} - J_{il})\,z_l.
\end{equation}
The first term captures the direct coupling between the two swapped vertices. The
QUBO-to-Ising conversion (expanding $P\sum_{(i,j)\in E}x_ix_j$ under $x_i=(1-z_i)/2$) gives
$J_{ij} = P/4$ for every edge, so $|2J_{ij}| = 2 \cdot P/4 = P/2$. The second term captures
how the neighbourhood of vertex $j$ differs from the neighbourhood of vertex $i$; each
neighbour contributes at most $P/4$ and there are at most $d_{\max}-d_{\min}$ net differing
neighbours, giving a bound of $2(d_{\max}-d_{\min}) \cdot P/4$. Both bounds are independent of
$n$.

\textbf{Comparison with unconstrained landscape.}
Without the subspace constraint, the circuit can move between states with different numbers of
selected vertices. The energy difference between a state with $k$ selected vertices and one
with $k \pm 1$ selected vertices receives a contribution from the single-qubit field terms of
order $O(n)$, since $\sum_i h_i$ grows with the number of vertices. The constrained landscape
eliminates this large-amplitude variation entirely, leaving only the interaction terms whose
differences are $O(d_{\max})$. The landscape is therefore smoother between neighbouring
states, which reduces the risk of sharp barriers.

Flat regions far from feasible solutions are reduced for a different reason. The
excitation-preserving ansatz eliminates these regions by construction, since the $fSim$ gates
conserve Hamming weight and the circuit never leaves $\mathcal{H}_k$. Wrong-Hamming-weight
states are structurally unreachable rather than merely penalised. Every state visited during
optimization already has exactly $k$ selected vertices and is therefore feasible in the
cardinality sense, keeping the optimizer within the relevant part of the landscape throughout.
\end{proof}

For instances where $k \ll n$ and $\Delta d$ is small, as in the near-regular 99- and 180-node instances, this represents a substantial reduction in energy variance. In full-space VQE, gradients are dominated by large energy differences between states of different Hamming weights ($O(n)$ per unit of weight change), causing the optimizer to spend most of its budget navigating toward the correct cardinality rather than distinguishing between size-$k$ configurations. Within $\mathcal{H}_k$, all Hamming-weight terms are constant and the gradient reflects only the fine-grained differences between size-$k$ configurations: which vertices are selected and whether their selection creates independence violations. The excitation-preserving constraint therefore creates a landscape where the optimizer is directly guided toward feasible, high-quality independent sets.

\textbf{Seed diversity.}
Diverse seeds cover complementary regions of $\mathcal{H}_k$. The superposition explores at least as much of the landscape as the best single seed and strictly more when the seeds are mutually distant in Hamming space, as is the case for the 180-node instance, where the seeds cover structurally different regions of the graph (confirmed by the recovery of MIS solutions with low vertex overlap to any individual seed).

The variational energy optimized in Phase~3 is the average across all $m$ seed branches,
\begin{equation}
  \frac{1}{m}\sum_{j=1}^m E_j(\boldsymbol{\theta})
  = \frac{1}{m}\sum_{j=1}^m \langle s^{(j)}| U^\dagger(\boldsymbol{\theta}) H_{\mathrm{MIS}} U(\boldsymbol{\theta}) |s^{(j)}\rangle.
\end{equation}
The optimal parameters therefore favour configurations that score well across the full seed population, not just in the vicinity of a single seed; a formal statement is given in Sec.~S4 of the Supplemental Material. This regularization effect is confirmed empirically by the ablation study of Section~\ref{subsec:ablation}: applying the excitation-preserving circuit from each size-14 seed individually never yields the size-15 MIS, while the ancilla superposition over multiple seeds succeeds. The discovery of output solutions containing vertices present in none of the seed states confirms that the circuit explores the solution space beyond the reach of any individual seed.

\textbf{MPS bond dimension.}
\begin{proposition}[Bond dimension of the excitation-preserving circuit]
\label{prop:bond_dim}
Let $|\Psi_0\rangle = \frac{1}{\sqrt{m}}\sum_{j=1}^{m}|s^{(j)}\rangle$ be the ancilla superposition of $m$ computational-basis seed states, evolved by $p$ layers of nearest-neighbour $fSim$ gates on a linear qubit chain. The MPS bond dimension satisfies $\chi_{\max} \leq m \cdot 2^p$.
\end{proposition}

\begin{proof}
The bond dimension $\chi$ of an MPS representation measures how much entanglement exists
across any cut that divides the qubit chain into a left and right half. It equals the number
of non-zero singular values in the Schmidt decomposition across that cut. We bound
$\chi_{\max}$ in two steps, first for the initial state before any gates are applied, and then
for how much the gates can increase it.

\textbf{Step 1, bond dimension of the initial superposition.}
Each seed $|s^{(j)}\rangle$ is a classical bitstring, so it is a product state with bond
dimension 1 across every cut. The initial state
$|\Psi_0\rangle = \frac{1}{\sqrt{m}}\sum_{j=1}^m |s^{(j)}\rangle$ places $m$ such product
states into an equal superposition. To bound its bond dimension, consider any cut that splits
the qubit chain into a left half $L$ and a right half $R$. Each seed restricts to a definite
bitstring on $L$, giving a left component $|s^{(j)}_L\rangle$. The Schmidt decomposition across
this cut can have at most as many terms as there are distinct left components. Since there are
$m$ seeds, there are at most $m$ distinct left components, so the Schmidt rank is at most $m$
and $\chi_0 \leq m$. The bound is tight when the seeds differ from one another on the left side
of the cut.

\textbf{Step 2, bond dimension growth per circuit layer.}
Each layer of the excitation-preserving circuit applies two-qubit $fSim$ gates between
neighbouring qubits. A two-qubit gate acting on qubits at positions $i$ and $i+1$ can only
affect the Schmidt decomposition across cuts that pass between those two qubits. The effect on
the Schmidt rank across any such cut is at most a factor of 2, because the gate maps each
Schmidt term into at most two new terms. For a gate that acts only on nearest-neighbour pairs,
cuts that do not separate the gate's two qubits are unaffected. Since the circuit has $p$
layers, no information can travel more than $p$ sites away from any qubit, and the Schmidt rank
across any fixed cut can grow by at most a factor of $2^p$ over the full circuit.

\textbf{Combined bound.}
Multiplying the initial Schmidt rank by the per-layer growth factor,
\begin{equation}
  \chi_{\max} \;\leq\; \chi_0 \cdot 2^p \;\leq\; m \cdot 2^p.
  \label{eq:bond_dim_proof}
\end{equation}
For the parameters used here ($m = 4$ seeds, $p \leq 3$), this gives $\chi_{\max} \leq 32$,
which is within the range that MPS simulation handles exactly.
\end{proof}

\begin{remark}
The bound is nearly tight in practice. For $m=4$ seeds and $p=2$, the theoretical maximum is $\chi_{\max}=16$; numerical measurement at truncation threshold $\varepsilon=10^{-16}$ confirms $\chi=16$ exactly. At $p=3$, the theoretical maximum is $\chi_{\max}=32$; measurement gives $\chi=31$ ($97\%$ of the ceiling), confirming that the optimized $fSim$ angles generate close-to-maximal entanglement and the four seed branches contribute nearly independent Schmidt components. Beyond $p=3$, the circuit operates at the limit of exact MPS simulation and at $p \geq 4$ ($\chi_{\max} \geq 64$) enters the truncation regime, making hardware execution preferable.
\end{remark}

\textbf{Empirical evidence for the $p$-scaling threshold.}
The bond dimension bound has direct empirical consequences for solution quality. For the 180-node instance, $p=1$ yields no MIS; $p=2$ recovers the MIS with post-processing. For the 400-node system (brock400-1, certified MIS~27), $p=1$ finds independent sets of size at most 24; $p=2$ also finds size-24 sets but with 5--7$\times$ more distinct solutions; $p=3$ is required to recover the first size-25 solutions. This monotonic improvement with $p$ is consistent with the bond dimension analysis: each additional layer increases $\chi$ by a factor of two, expanding the reachable region of the Hamming-weight-$k$ subspace and enabling the circuit to distinguish configurations that are energetically indistinguishable at shallower depth.

\textbf{Why increasing $p$ helps for the excitation-preserving ansatz but not for EfficientSU2.}
A standard EfficientSU2 ansatz uses single-qubit rotations interleaved with CNOT entangling layers. Adding layers expands expressivity across the full $2^n$-dimensional Hilbert space, but this is precisely the problem: the MIS solution occupies an exponentially small corner of that space, and deeper EfficientSU2 circuits suffer increasingly severe barren plateaus, making optimization harder rather than easier.

The excitation-preserving ansatz behaves differently in two respects. First, it is initialized inside the good region: the ancilla superposition places all amplitude within the Hamming-weight-$k$ subspace spanned by near-optimal seeds, so additional layers increase expressivity within the relevant near-MIS region rather than spreading it across the full $2^n$-dimensional space. Second, the barren plateau argument does not apply with the same force. Gradients are computed over a $\binom{n}{k}$-dimensional space rather than a $2^n$-dimensional one, and initialization from high-quality seeds keeps the optimization away from the flat plateaus that afflict random initialization in the full space. Increasing $p$ for the excitation-preserving ansatz is therefore a justified move. Each extra layer allows excitations to propagate further across the graph, building correlations that distinguish configurations which are energetically indistinguishable at shallower depth, while the constrained initialization ensures the additional expressivity is directed toward the relevant part of the solution space.

\subsection{Ancilla Interference: Theory}
\label{subsec:theory_interference}

The ancilla interference circuit (Section~\ref{subsec:interference}) produces the post-selected data-register probability
\begin{equation}
  P_\mathrm{interf}(\mathbf{x}) = \frac{1}{K}\left|\sum_{k=0}^{K-1} A_k(\mathbf{x})\right|^2,
  \qquad A_k(\mathbf{x}) = \langle\mathbf{x}|\,e^{i\sum_j \phi_{kj}\hat{n}_j}\,U(\boldsymbol{\theta})\,|s^{(k)}\rangle,
  \label{eq:p_interf_theory}
\end{equation}
where $\hat{n}_j$ is the occupation number operator for data qubit $j$. This differs from the stochastic-mixture distribution $P_\mathrm{stoch}(\mathbf{x}) = \frac{1}{K}\sum_k|A_k(\mathbf{x})|^2$ by the cross-terms
\begin{equation}
  P_\mathrm{interf}(\mathbf{x}) - P_\mathrm{stoch}(\mathbf{x}) = \frac{2}{K}\sum_{k < k'}\mathrm{Re}\!\left[A_k(\mathbf{x})A_{k'}^*(\mathbf{x})\right].
  \label{eq:interf_vs_stoch}
\end{equation}
A configuration $\mathbf{x}$ is constructively amplified when all cross-terms are positive, meaning all pairs of seed amplitudes $(A_k(\mathbf{x}), A_{k'}(\mathbf{x}))$ are in phase. Conversely, configurations where the seed amplitudes partially cancel are destructively suppressed.

Both distributions are normalised over the data register, so the cross terms sum to zero,
\begin{equation}
  \sum_{\mathbf{x}} \left[ P_\mathrm{interf}(\mathbf{x}) - P_\mathrm{stoch}(\mathbf{x}) \right] = 0 .
  \label{eq:cross_sum_zero}
\end{equation}
Interference moves probability between bitstrings and creates no new ones. An observable of the form
$\sum_{\mathbf{x}} f(\mathbf{x}) P(\mathbf{x})$ therefore differs between the two distributions
only through
\[
  \sum_{\mathbf{x}} f(\mathbf{x})\left[P_\mathrm{interf}(\mathbf{x}) -
  P_\mathrm{stoch}(\mathbf{x})\right],
\]
a sum in which positive and negative cross terms cancel.
Such an observable responds to the part of the cross terms that correlates with $f$ and not to
their size. The mean IS size of the post-selected output is of this form, with $f$ the IS size,
which limits what the phase sweep of Test~2 can resolve.

The cross terms can be recovered without averaging by keeping every ancilla outcome. Writing
$a$ for the measured ancilla pattern, the Hadamard layer gives
\begin{equation}
  P(\mathbf{x},a) = \frac{1}{K^2}\left|\sum_{k=0}^{K-1} (-1)^{a\cdot k} A_k(\mathbf{x})\right|^2 ,
  \label{eq:joint_xa}
\end{equation}
and summing over the $K$ patterns removes every cross term through the orthogonality of the
sign factors,
\begin{equation}
  \sum_a P(\mathbf{x},a) = \frac{1}{K}\sum_k |A_k(\mathbf{x})|^2 = P_\mathrm{stoch}(\mathbf{x}) .
  \label{eq:anc_marginal}
\end{equation}
The pattern $a=0$ returns the post-selected distribution, $K\,P(\mathbf{x},0) =
P_\mathrm{interf}(\mathbf{x})$, so that
\begin{equation}
  P_\mathrm{interf}(\mathbf{x}) - P_\mathrm{stoch}(\mathbf{x})
  = K\,P(\mathbf{x},0) - \sum_a P(\mathbf{x},a) ,
  \label{eq:cross_measurable}
\end{equation}
with both terms taken from the same shots of the same circuit. For $K=2$ and one ancilla the
two outcomes are $P(\mathbf{x},0) = \frac{1}{4}|A_0+A_1|^2$ and $P(\mathbf{x},1) =
\frac{1}{4}|A_0-A_1|^2$, whose sum is $P_\mathrm{stoch}(\mathbf{x})$ and whose difference is
$\mathrm{Re}[A_0 A_1^*]$. The normalised contrast
\begin{equation}
  V(\mathbf{x}) = \frac{P(\mathbf{x},0)-P(\mathbf{x},1)}{P(\mathbf{x},0)+P(\mathbf{x},1)}
  = \frac{2\,\mathrm{Re}\!\left[A_0(\mathbf{x})A_1^*(\mathbf{x})\right]}
         {|A_0(\mathbf{x})|^2+|A_1(\mathbf{x})|^2}
  \label{eq:contrast}
\end{equation}
lies between $-1$ and $1$, and it vanishes for any $\mathbf{x}$ that a single branch alone
reaches. When the cross terms vanish the ancilla patterns are equally likely for each
bitstring. Measuring how the shots for a given bitstring divide among the patterns tests the
cross terms without averaging over $\mathbf{x}$.

\begin{remark}[Post-selection rate]
\label{prop:postsel_rate}
For the uniform-phase interference circuit ($\phi_{kj}=0$ for all $k,j$), the probability of the ancilla measuring $|0\rangle^{\otimes a}$ is exactly $1/K$ for any depth and any distinct seeds, since distinct computational basis states are orthonormal and $U(\boldsymbol{\theta})$ is unitary. This is an elementary consequence of
unitarity and seed orthonormality rather than a property of the interference mechanism.
\end{remark}

\begin{proof}
The ancilla register is initialized to the uniform superposition
$(1/\sqrt{K})\sum_k|k\rangle$ and evolved by $U(\boldsymbol{\theta})$ on the data register
only. After the final Hadamard $H^{\otimes a}$ on the ancilla, the joint state is
\begin{equation}
  |\Phi_\mathrm{final}\rangle = \frac{1}{\sqrt{K}}\sum_k U(\boldsymbol{\theta})|s^{(k)}\rangle \otimes H^{\otimes a}|k\rangle.
\end{equation}
The amplitude in the ancilla $|0\rangle^{\otimes a}$ component is
$\frac{1}{K}\sum_k U(\boldsymbol{\theta})|s^{(k)}\rangle$, with squared norm
\begin{equation*}
  \frac{1}{K^2}\sum_{k,k'}\langle s^{(k')}|U^\dagger U|s^{(k)}\rangle
  = \frac{1}{K^2}\sum_{k,k'}\delta_{kk'} = \frac{1}{K},
\end{equation*}
where the last step uses unitarity of $U$ and orthonormality of the seeds.
\end{proof}

The post-selection rate of exactly $1/K$ is confirmed empirically: for $K=4$ seeds the ancilla measures $|00\rangle$ on approximately $25\%$ of shots, flat within $\pm 0.9\%$ of the expected value of $0.250$ across the 180-node interference experiments, confirming that the circuit is correctly prepared.

\begin{proposition}[Constructive interference for jointly good configurations]
\label{prop:constructive}
Let $\mathbf{x}^*$ be a bitstring of interest. If the amplitudes $A_k(\mathbf{x}^*)$ are all in phase, then
\begin{equation}
  P_\mathrm{stoch}(\mathbf{x}^*) \;\leq\; P_\mathrm{interf}(\mathbf{x}^*) \;\leq\; K \cdot P_\mathrm{stoch}(\mathbf{x}^*),
\end{equation}
with the upper bound achieved when all $|A_k(\mathbf{x}^*)|$ are additionally equal in magnitude. Conversely, if $\sum_k A_k(\mathbf{x}^*) = 0$, then $P_\mathrm{interf}(\mathbf{x}^*) = 0$.
\end{proposition}

\begin{proof}
When all $A_k(\mathbf{x}^*)$ have the same argument $\psi^*$, writing
$A_k(\mathbf{x}^*) = |A_k(\mathbf{x}^*)|e^{i\psi^*}$,
\begin{align}
  P_\mathrm{interf}(\mathbf{x}^*) &= \frac{1}{K}\left(\sum_k |A_k(\mathbf{x}^*)|\right)^2.
\end{align}
For the lower bound, since
$\left(\sum_k|A_k|\right)^2 = \sum_k|A_k|^2 + 2\sum_{k<k'}|A_k||A_{k'}| \geq \sum_k|A_k|^2$,
we have $P_\mathrm{interf} \geq \frac{1}{K}\sum_k|A_k|^2 = P_\mathrm{stoch}$. For the upper
bound, Cauchy-Schwarz gives $\left(\sum_k|A_k|\right)^2 \leq K\sum_k|A_k|^2$, so
$P_\mathrm{interf} \leq \frac{1}{K}\cdot K\sum_k|A_k|^2 = K\cdot P_\mathrm{stoch}$, with
equality when all $|A_k|$ are equal. In the fully destructive case, if
$\sum_k A_k(\mathbf{x}^*) = 0$ then $P_\mathrm{interf}(\mathbf{x}^*) = 0$ directly from the
definition.
\end{proof}

\begin{remark}[Interference mechanism and diagnostics]
The phase alignment condition of Proposition~\ref{prop:constructive} is not guaranteed for an arbitrary circuit $U(\boldsymbol{\theta})$. However, the CRZ phase parameters $\boldsymbol{\phi}$ in the interference circuit provide additional degrees of freedom that can be optimized to align the phases of the seed amplitudes on good configurations. In the experiments reported here, the phase parameters were either kept at zero (stochastic mixture baseline) or optimized with SPSA from a zero initialization, corresponding to the minimum-perturbation starting point relative to the baseline.

The phase-sensitivity test (Test~2 of the interference diagnostic) sweeps a scalar offset $\delta$ across the CRZ parameters of one seed and measures how the mean IS size of the post-selected output varies. Such a variation would show that $\mathrm{Re}[A_k A_{k'}^*] \neq 0$. The measured sweep is consistent with a constant and therefore places an upper limit on the cross terms rather than showing them, as reported in Section~\ref{subsec:results_interference}. The limit follows from Eq.~(\ref{eq:cross_sum_zero}), which leaves an averaged observable insensitive to the size of the cross terms.

A departure from $1/K$ can arise only through the diagonal phase operators, since the relevant
inner product
\[
  \langle s_k|\,U^\dagger D_k^\dagger D_{k'} U\,|s_{k'}\rangle
\]
is a phase weighted sum of the per bitstring products and is not pinned to zero, whereas at
$\phi = 0$ it reduces to $\langle s_{k'}|s_k\rangle$ and vanishes. The post-selection rate stays close to $1/K$ in the experiments because the phase weighted cross terms integrated over the output remain small. This does not imply that the cross terms are small on every individual bitstring, since only their sum is constrained. Resolving them therefore calls for the channel resolved measurement of Eq.~(\ref{eq:cross_measurable}) rather than for the conditional distribution alone.

The cross-term $\mathrm{Re}[A_k(\mathbf{x})A_{k'}^*(\mathbf{x})]$ is maximized when both seeds $k$ and $k'$ assign substantial amplitude to the same bitstring $\mathbf{x}$. In an optimized excitation-preserving circuit, large amplitudes concentrate on low-energy, high-quality independent sets, and where the branches reach common bitstrings those are the configurations most likely to extend to the MIS under the stepwise heuristic. The interference mechanism would then amplify the bitstrings that are useful across all seeds, which would align it with the classical post-processing objective. Whether it does so at a level that can be measured is addressed in Section~\ref{subsec:results_interference}.
\end{remark}

%=============================================================================
\section{Conclusions}
\label{sec:conclusion}
%=============================================================================

This paper presents a layered framework for solving the Maximum Independent Set
problem with variational quantum algorithms. The framework comprises of combining
graph reordering, hardware-efficient variational circuits, SPSA/COBYLA-based classical
optimization, and bitstring post-processing into a coherent pipeline. Applied to
benchmark instances spanning 64 to 400 nodes, the framework achieves exact maximum
independent sets on all instances up to 200 nodes, with the 180-node result
requiring the key novel component of ancilla-qubit-mediated preparation of a
uniform quantum superposition over classically-found near-optimal solutions
followed by evolution under a particle-conserving excitation-preserving variational ansatz.
This breaks a hard barrier that neither the bitstring heuristic pipeline alone
nor standard VQE can resolve. We propose a hybrid approach, in which a single
circuit execution simultaneously optimizes over all seed branches under shared
variational parameters.

The superposition arm, measured on the data register alone, is classically tractable
as a stochastic mixture over the seeds, and we use it as the control. The ancilla
interference mechanism is the extension that goes beyond that mixture. A CRZ phase layer followed by
Hadamard on the ancilla and post-selection on $|0\rangle^{\otimes a}$ retains the
cross-term contributions $\mathrm{Re}[A_k A_{k'}^*]$ that the mixture discards.

In the configurations we were able to simulate, the interference
did not reach a level we could measure. A sweep of the CRZ offset leaves the quality of the
post-selected output unchanged to within shot noise, the fraction of shots surviving
post-selection stays at its expected value, and reading the ancilla outcomes directly shows no
contrast on individual bitstrings beyond what counting statistics allow.
Section~\ref{subsec:results_interference} states the limits these measurements place. The
reason is geometric. A cross term needs two seed branches to reach the same bitstring, and on
this instance the vertices that distinguish the optimal solutions lie far apart on the qubit
chain, so the branches seldom meet. Reaching the regime where they do meet calls for circuit
depth beyond what exact simulation allows here, and a device/hardware able to run those depths would be
able to show the effect. The $fSim$ gate with $\varphi \neq 0$ entangles the seed
branches, which the measured bond dimension confirms, and it is the depth at which the branches
begin to reach common bitstrings that decides whether the interference has anything to act on.

We validated the pipeline on five QOBLIB 200-node benchmark instances (all MIS
recovered), on the 125-node \texttt{c125-9} instance both in simulation (148
MIS recovered) and on IBM Quantum hardware \texttt{ibm\_marrakesh} (4 MIS
recovered, certified MIS size 34 confirmed on hardware), and on the 400-node
\texttt{brock400-1} instance. For \texttt{brock400-1}, standard VQE (Phase~1) reaches IS size 24 only
infrequently, while the excitation-preserving ansatz with ancilla superposition
(Phase~3) recovers IS size 24 with higher frequency and first
reaches IS size 25 (92.6\% of certified optimum 27) with $p=3$ ansatz layers;
$p=1,2$ plateau at IS=24. The $p$ (ansatz layers) scaling mirrors the 180-node pattern
($p=2$ needed for MIS) and provides empirical evidence that circuit
expressiveness, not post-optimization shot count, is the limiting factor. The higher penalty coefficient required for Phase~1
($P=4$--$6$ for standard VQE+CVaR, vs.\ $P=2$ for the excitation-preserving
Phase~3 ansatz) reflects the much higher graph density of this instance:
a weak penalty fails to distinguish feasible IS from constraint-violating
configurations in the energy landscape, whereas Phase~3 enforces feasibility
structurally via the Hamming-weight subspace and requires no $P$-scaling.

The specific circuits studied here are efficiently simulable classically via MPS for
shallow depth and small seed counts ($\chi_{\max} \leq m \cdot 2^p$), and we make no claim of
quantum computational advantage for these instances. Rather, we demonstrate that
the ancilla superposition and interference techniques \emph{work}: they provide
principled mechanisms for breaking local optima. Two routes lead to the branch overlap the interference mechanism needs. The reordering fixes
the distance between the vertices that distinguish near-optimal solutions, and choosing a
permutation that brings them together is a classical change we have not attempted. Circuit
depth extends the reach of excitation transport, and this route is closed to us classically
rather than algorithmically, since the bond dimension of the seed superposition grows as
$m \cdot 2^p$, which we confirm at 4, 8, 16, 32 and 113 for $p$ from 1 to 6, and since the long
range controlled RZ gates of the phase layer raise the cost of one sampled shot from under two
seconds to over thirty minutes. As hardware improves, the interference over bitstrings becomes
something that can be shown experimentally at depths we cannot simulate, and whether it then
yields better solutions depends on the structure of the graph at hand, since the mechanism acts
only where the vertices distinguishing near-optimal solutions lie close on the qubit chain.

The 400-node experiments
indicate where hardware becomes necessary for this pipeline. The MPS bond dimension
grows as $2^p$ per layer, placing $p \geq 4$ circuits beyond the exact MPS budget available
to us, and recovering the certified MIS at 400 nodes appears to need depths past that point.
Exact MPS at moderate bond dimension is not the limit of classical simulation in general,
and the statement here is about the method we use.

The mechanism is confirmed by a direct ablation study on the 180-node instance.
Applying the same bitstring heuristic pipeline to each size-14 seed state
individually, with repeated restarts, never yields the size-15 MIS. The
pipeline's local search is trapped at each seed. Only when multiple seeds are combined into a
quantum superposition does the pipeline succeed, confirming that the discovery
of the MIS is a consequence of parallel search over multiple
near-optimal initializations rather than of the classical post-processing
alone. The discovery of output solutions containing vertices present in none
of the seed states further confirms that the pipeline explores the solution
space beyond the reach of any individual seed initialization.

Hardware validation on IBM Quantum hardware \texttt{ibm\_marrakesh} confirms that the
pipeline is robust to realistic noise conditions. Using converged noiseless simulator
parameters transferred directly to the hardware circuit, VQE recovers 3 distinct MIS of
size 10 for the 64-node instance (hardware depth 65) and 5 distinct MIS of size 7 for
the 99-node instance (hardware depth 100), approximately half the noiseless diversity at
40\% of the shot budget with no hardware-aware re-optimization. QAOA, by contrast,
recovers zero valid independent sets across all three hardware configurations tested.
This outcome reflects a compounding effect of shallow classical optimization and
circuit-depth-induced decoherence, consistent with the simulator constraints that
limited the QAOA configurations, rather than a definitive statement about QAOA viability
on hardware in general. The approximately 8 seconds of actual QPU execution time per
Mode-B iteration ($\approx 80$ seconds total across $T=10$ iterations) against 3.6--37.2
minutes of end-to-end wall-clock time is consistent with the bond-dimension analysis,
since the quantum circuit executes essentially instantaneously and the practical cost
lies in classical optimization and post-processing.

The most immediate open question is the $p \geq 4$ regime on \texttt{brock400-1}. Bond dimension exceeds 64
at that depth, putting exact MPS simulation out of reach. The current hardware path of transferring
noiseless parameters directly to \texttt{ibm\_marrakesh} loses roughly half the solution diversity.
Hardware-native re-optimization, in which the variational parameters are updated directly on the
quantum processor is the natural next step, though it increases QPU time and makes
the interplay with error mitigation strategies~\cite{temme2017error} non-trivial to manage.
A second open direction is iterative enumeration. The 180-node instance has multiple structurally
distinct MIS families, and the seed composition of the ancilla superposition determines which family
is recovered. Feeding a discovered MIS back as a seed in a subsequent Phase 3 run is a simple
change to the pipeline that converts single-solution discovery into progressive enumeration of the
solution space. We observed this informally but did not study it systematically. Understanding when
the bootstrapping converges over cycles is a question with practical value for instances
where many MIS exist. The sparsification bandwidth $k$ was set by grid search and held fixed
throughout optimization. An adaptive $k$ could be an interesting line of work. In other words,
couplings dropped early in the circuit can be reintroduced as the optimizer approaches convergence
and finer distinctions between configurations begin to matter. Whether adaptive sparsification
provides a consistent benefit is worth testing.

Finally, all components here, the excitation-preserving ansatz, ancilla superposition, and
interference post-selection, are directly applicable to any problem with an Ising Hamiltonian and
a cardinality constraint on the solution: Maximum Clique, Minimum Vertex Cover, graph colouring, etc.
The 180-node result suggests the framework is most useful precisely when the solution is a small
fraction of the full vertex set and local search heuristics get trapped repeatedly at the same
local optima. Whether that pattern holds on other combinatorial problems is an open question.

%=============================================================================
\section*{Acknowledgements}
%=============================================================================

The authors thank Claude (Anthropic) \cite{claudeai} for assistance with literature survey,
for refining the flow and presentation of this paper, and for assistance in designing the
adversarial evolutionary graph generator used to construct the 99-node benchmark instance.
ChatGPT (OpenAI) \cite{chatgpt} is acknowledged for assistance in designing the CSP and
spin-glass encoding used to construct the 180-node benchmark instance. IBM~Bob \cite{ibmbob}
is acknowledged for assistance with code development and implementation.

%=============================================================================
\section*{Data Availability}
%=============================================================================

The two custom benchmark instances introduced in this work, the 99-node adversarial
instance and the 180-node CSP and spin-glass instance, are available in DIMACS format at
\url{https://github.com/kaldag/dense_graph_instances}. The remaining instances are drawn
from the QOBLIB benchmark library~\cite{koch2025qoblib}. The implementation code for the
variational quantum pipelines, the bitstring heuristic framework and the graph reordering
strategies is maintained in an enterprise repository and is available to researchers upon
reasonable request.

% \input{author_contributions}

% \input{competing_interests}

% ---- Bibliography ---------------------------------------------------------
\bibliographystyle{unsrtnat}
\bibliography{references}

%=============================================================================
% SUPPLEMENTAL MATERIAL
%=============================================================================
\clearpage

% S-prefixed numbering, restarted from the main text counters.
\setcounter{section}{0}
\setcounter{figure}{0}
\setcounter{table}{0}
\setcounter{equation}{0}
\renewcommand{\thesection}{S\arabic{section}}
\renewcommand{\thefigure}{S\arabic{figure}}
\renewcommand{\thetable}{S\arabic{table}}
\renewcommand{\theequation}{S\arabic{equation}}

\begin{center}
  {\Large\bfseries Supplemental Material}\\[0.5em]
  {\large Quantum Variational Approaches to the Maximum Independent Set Problem at Utility Scale}
\end{center}
\vspace{1em}

The material below supports the main paper. It is organised
into six parts. Section~\ref{supp:penalty} gives a complete derivation of the
penalty parameter sufficiency condition. Section~\ref{supp:ema} analyses the
convergence of the EMA suppression weights. Section~\ref{supp:maximality} states
the classical maximality bound that motivates the minimum-degree selection rule.
Section~\ref{supp:collective} formalises the collective optimality property of the
ancilla superposition objective. Section~\ref{supp:diagonal} provides the full
comparison of VQE regimes across diagonal and non-diagonal Hamiltonians, including
the chemistry comparison that contextualises the advantages of the optimization
setting. Section~\ref{sec:S6} covers the ancilla interference measurements. It presents the phase
sweep of the CRZ offset together with its harmonic decomposition and the derivation of the
error on a fitted amplitude, followed by the paired replication of the interference and
stochastic arms and the null calibration of the channel-resolved statistic defined in the
main text.

%=============================================================================
\section{Penalty Parameter Sufficiency}
\label{supp:penalty}
%=============================================================================

The MIS Ising Hamiltonian (Eq.~(3) of the main paper) encodes the
independence constraint through penalty terms weighted by $P$. The following
proposition establishes that $P = 2$ is provably sufficient for all instances
studied.

\begin{proposition}[Penalty sufficiency]
\label{prop:penalty}
For the MIS Hamiltonian with single-qubit fields $h_i = 0.5 - P \cdot d_i/4$,
two-qubit couplings $J_{ij} = P/4$ for all edges $(i,j) \in E$, and
penalty parameter $P \geq 2$, any feasible independent set has strictly
lower energy than any infeasible configuration obtained by adding a
vertex that creates at least one violation. The energy gap is
$\Delta E = cP - 1 \geq P - 1 \geq 1$, where $c \geq 1$ is the
number of violations created. The gap is independent of vertex degree.
\end{proposition}

\begin{proof}
The Ising Hamiltonian arises from the QUBO objective
$-\sum_i x_i + P\sum_{(i,j)\in E} x_ix_j$ via the substitution
$x_i = (1-z_i)/2$. Expanding the penalty term,
\begin{equation}
  P\sum_{(i,j)\in E} x_ix_j
  = \frac{P}{4}\sum_{(i,j)\in E}(1 - z_i - z_j + z_iz_j),
\end{equation}
which contributes two-qubit couplings $J_{ij} = P/4$ and, since
vertex $i$ appears in $d_i$ edges, single-qubit fields $-Pd_i/4$.
Together with the reward term $-x_i = -\tfrac{1}{2} + \tfrac{z_i}{2}$,
the single-qubit field for vertex $i$ is
$h_i = \tfrac{1}{2} - \tfrac{Pd_i}{4}$.

Let $S$ be a feasible independent set and let $v^*$ be a vertex with
$c = |N(v^*) \cap S| \geq 1$ violated neighbours in $S$ and
$d_{v^*} - c$ non-selected neighbours. Adding $v^*$ to $S$ flips
qubit $v^*$ from $z_{v^*} = +1$ to $z_{v^*} = -1$, affecting three
types of Hamiltonian terms.

\textbf{Field term.} The contribution $h_{v^*}z_{v^*}$ changes from
$h_{v^*}(+1)$ to $h_{v^*}(-1)$, giving a change of
\begin{equation}
  \Delta E_{\text{field}} = -2h_{v^*}
  = -2\!\left(\tfrac{1}{2} - \frac{Pd_{v^*}}{4}\right)
  = -1 + \frac{Pd_{v^*}}{2}.
\end{equation}

\textbf{Violated coupling terms.} For each violated neighbour
$u \in N(v^*) \cap S$ (where $z_u = -1$), the coupling
$J_{v^*u}z_{v^*}z_u$ changes from $-J_{v^*u}$ to $+J_{v^*u}$,
contributing $+P/2$ per violated neighbour,
\begin{equation}
  \Delta E_{\text{violated}} = +\frac{cP}{2}.
\end{equation}

\textbf{Non-selected coupling terms.} For each non-selected neighbour
$w \notin S$ (where $z_w = +1$), the coupling $J_{v^*w}z_{v^*}z_w$
changes from $+J_{v^*w}$ to $-J_{v^*w}$, contributing $-P/2$ per
non-selected neighbour,
\begin{equation}
  \Delta E_{\text{non-selected}} = -\frac{(d_{v^*}-c)P}{2}.
\end{equation}

Combining and simplifying,
\begin{align}
  \Delta E &= -1 + \frac{Pd_{v^*}}{2} + \frac{cP}{2}
               - \frac{(d_{v^*}-c)P}{2}
           = cP - 1.
\end{align}
The vertex degree $d_{v^*}$ cancels exactly. For $c \geq 1$ and
$P \geq 2$, $\Delta E = cP - 1 \geq P - 1 \geq 1 > 0$. At $P = 2$,
$\Delta E = 2c-1$, giving gaps of $1, 3, 5, \ldots$ for
$c = 1, 2, 3, \ldots$ violations respectively.
\end{proof}

\begin{remark}
The standard MIS QUBO formulation uses the informal argument that the penalty per
violated edge ($P$) must exceed the reward per vertex ($1$), giving $P > 1$ as
the sufficiency condition at the QUBO level. After QUBO-to-Ising conversion, the
Hamiltonian fields $h_i = 0.5 - Pd_i/4$ contain contributions from both the vertex
reward and the edge penalty, so the QUBO-level argument does not directly carry over
to the Ising picture. The exact Ising-level analysis yields $\Delta E = cP - 1$,
which is independent of vertex degree due to an exact cancellation of the
degree-dependent terms across the field and non-selected coupling contributions.
The true sufficiency condition is $P > 1$ (equivalently, $P > 1/c$ for the worst
case $c = 1$), confirming the QUBO-level intuition while making it exact. The choice
$P = 2$ adopted in this work gives a guaranteed energy gap of $\Delta E = 2c-1 \geq 1$
for any number of violations $c \geq 1$.
\end{remark}

%=============================================================================
\section{EMA Suppression Weight Convergence}
\label{supp:ema}
%=============================================================================

The EMA suppression weights $w_i$ converge to the true zero-frequency
$\mu_i = P(\text{qubit } i \text{ in } |0\rangle \text{ in top-}M \text{ samples})$
of the stationary measurement distribution. The convergence has two components.

The first component is initialization bias. Starting from $w_i = 0.5$, the influence of the
initial value decays as $(1-\alpha)^T$ after $T$ sampling iterations, where
$\alpha \in (0,1)$ is the EMA adaptation rate. The second component is sampling noise.
Each iteration's zero-frequency estimate $z_i^{(t)}$ is computed from $B$ shots,
introducing noise bounded by Hoeffding's inequality,
\begin{equation}
  P(|z_i^{(t)} - \mu_i| > \varepsilon) \leq 2e^{-2B\varepsilon^2}.
\end{equation}
For the post-optimization sampling protocol used in this work
($\alpha = 0.4$, $T=10$ iterations with $B=50{,}000$ shots per iteration in
noiseless simulation and $B=20{,}000$ shots per iteration on hardware), the
initialization bias decays to $(0.6)^{10} \approx 0.6\%$ and the sampling noise
is negligible for $\varepsilon \geq 0.01$. The suppression weights are therefore
reliable to within approximately 0.6\% of their stationary values after the full
10-iteration protocol in both settings.

This concentration result applies strictly when the measurement distribution is
stationary, meaning when the variational parameters are fixed. It therefore holds
exactly during the post-optimization sampling phase (Mode~B after convergence) and
approximately during the final stages of optimization when the parameters have
largely stabilized. During active SPSA optimization the distribution shifts with
each parameter update, and the weights track a moving target. The practical
effectiveness of the masks in this regime is supported empirically rather than by
the concentration bound.

%=============================================================================
\section{Maximality Lower Bound}
\label{supp:maximality}
%=============================================================================

The stepwise extension to maximality (Algorithm~2 of the main paper) produces a
maximal independent set $S$ from any feasible starting set $S_0 \subseteq V$. Any
maximal independent set satisfies the uniform lower bound $|S| \geq n/(d_{\max}+1)$,
where $d_{\max}$ is the maximum vertex degree. The minimum-degree selection rule
yields a stronger instance-specific bound: by the result of
Halld\'{o}rsson and Radhakrishnan~\cite{halldorsson1997greedy}, the minimum-degree
greedy strategy produces a maximal independent set of size at least
\begin{equation}
  |S| \;\geq\; \sum_{v \in V} \frac{1}{d(v)+1},
  \label{eq:hrd_bound}
\end{equation}
where $d(v)$ is the degree of vertex $v$. This sum strictly exceeds
$n/(d_{\max}+1)$ for irregular graphs. For the 64-node instance with
$d_{\min}=6$ and $d_{\max}=24$, bound~(\ref{eq:hrd_bound}) gives a provably
tighter guarantee than the uniform estimate.

%=============================================================================
\section{Collective Optimality of the Ancilla Superposition Objective}
\label{supp:collective}
%=============================================================================

The following proposition formalises the regularisation property of the
multi-seed variational objective.

\begin{proposition}[Collective optimality]
\label{prop:collective}
Let $E_j(\boldsymbol{\theta}) = \langle s^{(j)}| U^\dagger(\boldsymbol{\theta})
H_{\mathrm{MIS}} U(\boldsymbol{\theta})|s^{(j)}\rangle$ denote the energy of
seed $s^{(j)}$ under the variational parameters $\boldsymbol{\theta}$. The
optimal parameters
$\boldsymbol{\theta}^* = \arg\min_{\boldsymbol{\theta}} \frac{1}{m}\sum_{j=1}^m
E_j(\boldsymbol{\theta})$
minimize the average energy across all $m$ seeds simultaneously. These parameters
favour solutions that perform well across the full seed population, penalizing
parameter choices that are optimal for a single seed's narrow local basin but poor
for the others.
\end{proposition}

This collective optimization acts as a regularizer: rather than overfitting to one
seed's local basin, the joint objective finds parameters that perform well across
all seeds. The ablation study of the main paper confirms this mechanism: applying
the excitation-preserving circuit from each seed individually never finds the
size-15 MIS, while the ancilla superposition over multiple seeds succeeds,
consistent with the collective optimization escaping the individual local basins.

%=============================================================================
\section{VQE Regimes: Diagonal vs. Non-Diagonal Hamiltonians}
\label{supp:diagonal}
%=============================================================================

This section places the excitation-preserving VQE of the main paper within the
broader context of VQE applied to different Hamiltonian types. It explains why
the diagonal structure of MIS confers advantages that do not transfer to, for
example, electronic structure problems, and introduces QAOA as a natural
comparator in the diagonal setting.

\subsection{Diagonal vs.\ Non-Diagonal Hamiltonians}

The Variational Quantum Eigensolver (VQE)~\cite{peruzzo2014variational} prepares a
parameterized quantum state $|\psi(\boldsymbol{\theta})\rangle$ and minimizes the
energy expectation value $E(\boldsymbol{\theta}) =
\langle\psi(\boldsymbol{\theta})|H|\psi(\boldsymbol{\theta})\rangle$
via a classical outer optimization loop. The tractability of this procedure depends
critically on the structure of $H$.

A Hamiltonian is \emph{diagonal} in the computational basis if every eigenstate
is a computational basis vector $|z\rangle \in \{0,1\}^n$,
\begin{equation}
  H_{\mathrm{diag}} = \sum_{z \in \{0,1\}^n} c_z\,|z\rangle\langle z|.
\end{equation}
Combinatorial optimization problems yield diagonal Hamiltonians. The
MIS Ising Hamiltonian of the main paper consists entirely of $Z$ and $ZZ$ Pauli
terms, and its ground state is the single bitstring $|z^*\rangle$ encoding the
maximum independent set.

The Quantum Approximate Optimization Algorithm (QAOA)~\cite{farhi2014quantum} is a
natural comparator in the diagonal setting. QAOA alternates a cost unitary
$e^{-i\gamma H_{\mathrm{diag}}}$ with a mixer unitary $e^{-i\beta \sum_i X_i}$.
Because $H_{\mathrm{diag}}$ is diagonal, the cost unitary requires no Pauli
decomposition and the circuit is shallow by construction, using only $2p$ parameters
for $p$ layers regardless of system size. QAOA therefore inherits the locality
advantage of the diagonal Hamiltonian directly. The excitation-preserving VQE of
the main paper goes further by confining evolution to the fixed-Hamming-weight
subspace, which eliminates penalty terms and sharpens energy resolution near the
MIS.

In contrast, electronic structure Hamiltonians are non-diagonal after qubit
mapping. The molecular Hamiltonian in second quantization,
\begin{equation}
  H = \sum_{pq} h_{pq}\,a_p^\dagger a_q
    + \tfrac{1}{2}\sum_{pqrs} g_{pqrs}\,a_p^\dagger a_q^\dagger a_r a_s,
\end{equation}
maps under Jordan-Wigner or Bravyi-Kitaev transformation to a sum of Pauli
strings that includes $X$ and $Y$ operators acting on multiple qubits.
Its ground state is a correlated superposition of many computational basis
states rather than a single bitstring. Representing this state accurately
requires a highly expressive ansatz.

\subsection{Depth Requirements}

For diagonal Hamiltonians, the target ground state is a computational basis
state, a product state with zero entanglement. Entangling layers are needed to
help the circuit distinguish near-optimal configurations, not to represent an
inherently correlated state. In the main paper, $p=2$ VQE layers suffice for the
64-node and 99-node instances; $p=2$ with ancilla superposition solves the
180-node instance.

For quantum chemistry, the correlated ground state requires multi-reference
character that cannot be captured without substantial entanglement. The requirement
for chemical accuracy (sub-mHa) forces deep circuits with $O(N^4)$ parameters for
$N$ spin-orbitals in UCCSD-type ans\"{a}tze~\cite{romero2018}. More recent adaptive
approaches such as ADAPT-VQE reduce the parameter count, but they
retain a deep, problem-dependent circuit and do not resolve the underlying tension between
accuracy and depth. Greater accuracy demands more depth, and greater depth worsens the barren
plateau~\cite{holmes2022}.
This tension is further compounded by noise-induced barren plateaus~\cite{wang2021noise},
whose effect worsens with circuit depth. For the shallow circuits of the main paper
($p \leq 3$ on IBM Marrakesh hardware), the noise-induced contribution is bounded
by the same depth argument that governs the Cerezo bound, so keeping $p$ small
simultaneously controls both sources of gradient suppression.

For the optimization problems studied in the main paper, the depth required
for the smaller and medium instances ($p=2$) is well within the trainable regime.
For larger instances the required depth rises toward $p=3$, which is marginal at
1000 shots per gradient evaluation (see Table~\ref{tab:p_limit} below); the
shot-budget analysis bounds this risk explicitly rather than dismissing it.

\subsection{Gradient Comparison Across VQE Regimes}

Table~\ref{tab:vqe_comparison} summarises the gradient scaling and barren plateau
risk across the four settings discussed in this section.

\begin{table}[h]
\centering
\small
\renewcommand{\arraystretch}{1.3}
\begin{tabular}{@{}p{2.8cm}p{2.8cm}p{2.8cm}p{2.8cm}p{3.2cm}@{}}
\toprule
Property
  & QAOA (diagonal)
  & VQE, optimization ($H$ diagonal)
  & VQE, chemistry ($H$ non-diagonal)
  & Excitation-preserving (main paper) \\
\midrule
Ground state
  & Comp.\ basis state
  & Comp.\ basis state
  & Correlated superposition
  & State in $\mathcal{H}_k$ \\
Cost locality
  & $\ell$-local (cost unitary)
  & $\ell$-local after reordering
  & Global (Jordan-Wigner)
  & $\ell$-local within $\mathcal{H}_k$ \\
Gradient scaling
  & $\Omega(1/\mathrm{poly}(n))$
  & $\Omega(1/\mathrm{poly}(n))$
  & $O(2^{-n})$ at depth
  & $\Omega\!\left(1/\tbinom{n}{k}^2\right)$ \\
Depth required
  & Shallow ($p \leq 3$)
  & Shallow ($p \leq 3$)
  & Deep (chemical accuracy)
  & Moderate; increasing $p$ beneficial \\
Subspace constraint
  & No
  & No
  & No
  & Yes ($\mathcal{H}_k$, fixed Hamming weight) \\
Barren plateau risk
  & Low
  & Low
  & High
  & Reduced by subspace constraint \\
\bottomrule
\end{tabular}
\caption{Comparison of VQE and QAOA regimes across problem types.
  QAOA and standard optimization VQE share the diagonal-Hamiltonian advantage.
  The excitation-preserving ansatz obtains a further reduction via subspace
  confinement.}
\label{tab:vqe_comparison}
\end{table}

\subsection{Gradient Trainability as a Function of Circuit Depth}

The DLA dimension at depth $p$ grows as $\dim(\mathfrak{g}_p) \approx m^2 \cdot 4^p$,
taking the square of the bond dimension $\chi_{\max} = m \cdot 2^p$ as a
proxy for how much of $\mathfrak{su}(\tbinom{n}{k})$ the circuit has explored.
This is an order-of-magnitude estimate; the actual DLA dimension at finite $p$
is smaller, since the algebra has not yet saturated. From the Larocca et al.\
framework, the gradient variance decreases approximately as
\begin{equation}
  \mathrm{Var}[\partial_\theta E] \;\approx\; \frac{C}{m^2 \cdot 4^p}.
  \label{eq:p_variance}
\end{equation}
Training is viable when the shot budget satisfies $B \gtrsim m^2 \cdot 4^p$;
once $m^2 \cdot 4^p$ exceeds $B$ by a wide margin, gradient estimates become too
noisy to drive optimization reliably. This should be read as an engineering
guideline rather than a sharp boundary. Table~\ref{tab:p_limit} classifies
$(p, B)$ pairs for $m=4$ seeds over an extended range.

\begin{table}[ht]
\centering
\small
\renewcommand{\arraystretch}{1.2}
\begin{tabular}{@{}ccccccc@{}}
\toprule
$p$ & $m^2 \cdot 4^p$ & $B = 1{,}000$ & $B = 4{,}000$ & $B = 16{,}000$ & $B = 65{,}000$ \\
\midrule
1 & 64     & trainable   & trainable  & trainable  & trainable  \\
2 & 256    & trainable   & trainable  & trainable  & trainable  \\
3 & 1{,}024  & marginal  & trainable  & trainable  & trainable  \\
4 & 4{,}096  & unreliable & marginal  & trainable  & trainable  \\
5 & 16{,}384 & fails      & fails     & marginal   & trainable  \\
6 & 65{,}536 & fails      & fails     & fails      & marginal   \\
\bottomrule
\end{tabular}
\caption{Gradient trainability as a function of $p$, $m=4$ seeds.
  Classification based on whether $m^2 \cdot 4^p \lesssim B$ (trainable),
  $m^2 \cdot 4^p \approx B$ (marginal), or $m^2 \cdot 4^p \gg B$ (fails).
  All thresholds are order-of-magnitude estimates.}
\label{tab:p_limit}
\end{table}

At $p=4$, $\chi_{\max} = m \cdot 2^4 = 64$ also pushes the circuit into the MPS
truncation regime, making hardware execution necessary. Both thresholds are governed
by the same underlying quantity, the bond dimension $\chi_{\max} = m \cdot 2^p$, so
their alignment at $p=4$ is a structural consequence of using the same proxy for
two different phenomena rather than an independent empirical finding. Hardware
execution becomes both classically necessary and gradient-challenging at the same
depth, which provides a principled motivation for the depth range used in the
main paper.

The shot budget required to maintain trainability scales as $m^2$. At $p=4$,
the threshold $m^2 \cdot 4^p = m^2 \cdot 256$ grows rapidly with $m$: for $m=4$
seeds it is $4{,}096$ (marginal at $B=4{,}000$), for $m=8$ seeds it rises to
$16{,}384$ (requiring $B \approx 20{,}000$ shots per gradient evaluation), and for
$m=16$ seeds to $65{,}536$ (requiring $B \approx 100{,}000$ shots). The quadratic
scaling in $m$ means doubling the seed count costs four times the shot budget,
so the practical operating point is a joint choice over $(p, m, B)$ rather than
any single parameter.

%=============================================================================
\section{Phase Sweep and its Harmonic Decomposition}
\label{sec:S6}

Figure~\ref{fig:s6_fringe} shows the phase sweep reported under \emph{Ancilla Interference,
Experimental Evidence} in the main text. A scalar offset
$\delta$ was added to all CRZ parameters of seed~0 and swept from $0$ to $2\pi$ in 24 steps
with $20\,000$ shots per step.

\begin{figure}[!htbp]
\centering
\includegraphics[width=0.86\linewidth]{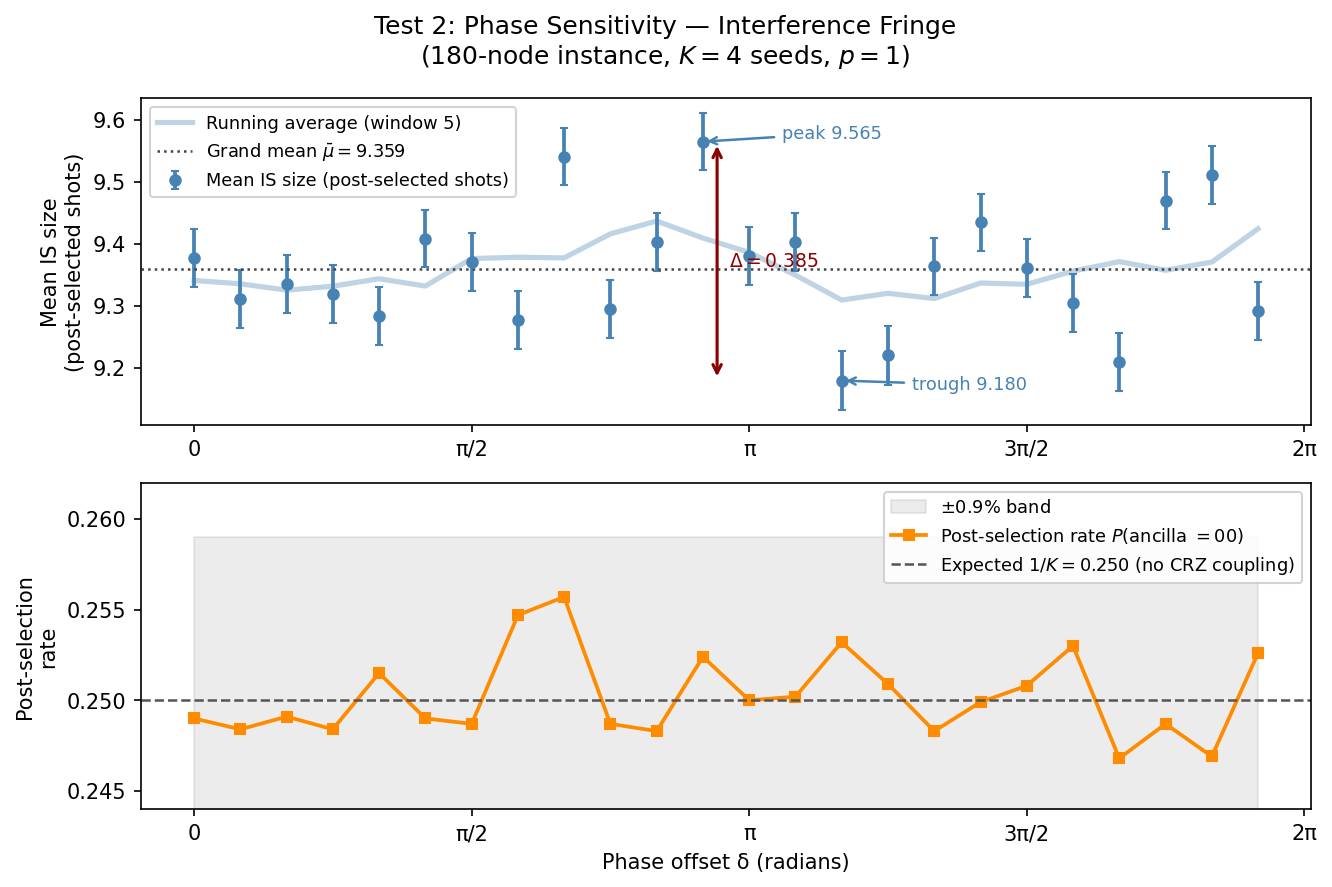}
\caption{Phase sweep on the 180-node instance ($K=4$ seeds, $p=1$, 24 phase steps,
  $20\,000$ shots per step). \emph{Top}, mean IS size of post-selected shots against the CRZ
  phase offset $\delta$, with error bars denoting shot noise and the shaded curve a running
  average over five steps. The scatter is consistent with a constant, with a chi-squared
  statistic of $26.2$ on 23 degrees of freedom. The running average is drawn to guide the eye only, since smoothing
  correlates neighbouring points and the structure it produces is not supported by the raw
  measurements. \emph{Bottom}, post-selection rate over the same sweep, flat about
  $1/K = 0.250$.}
\label{fig:s6_fringe}
\end{figure}

The offset is added to all CRZ parameters of seed~0, which act on four targeted qubits. In the
seed~0 branch a bitstring $\mathbf{x}$ therefore acquires the extra phase $\delta\,m(\mathbf{x})$,
where $m(\mathbf{x})$ counts how many of those four qubits $\mathbf{x}$ occupies and runs from 0
to 4. The terms of $\left|\sum_k A_k(\mathbf{x})\right|^2$ that involve branch~0 carry this
phase while $|A_0(\mathbf{x})|^2$ does not, so the post-selected probability of each bitstring is
a trigonometric polynomial in $\delta$ of degree at most four,
\begin{equation}
  P(\mathbf{x},\delta) = c_0(\mathbf{x}) + \sum_{m=1}^{4}
  \left[\, c_m(\mathbf{x})\cos(m\delta) + s_m(\mathbf{x})\sin(m\delta) \,\right] .
  \label{eq:s6_trigpoly}
\end{equation}
A weighted sum over bitstrings keeps that degree, so a modulation of the mean IS size can appear
only in the first four harmonics.

\paragraph{Error on a fitted amplitude.}
Write $y_i$ for the mean IS size at $\delta_i = 2\pi i/n$ with $i = 0,\ldots,n-1$ and $n = 24$.
Fitting $y_i = \mu + \sum_k \left[A_k\cos(k\delta_i) + B_k\sin(k\delta_i)\right]$ on equally
spaced points gives
\begin{equation}
  A_k = \frac{2}{n}\sum_i (y_i - \bar{y})\cos(k\delta_i), \qquad
  B_k = \frac{2}{n}\sum_i (y_i - \bar{y})\sin(k\delta_i),
  \label{eq:s6_fourier}
\end{equation}
which follow from the orthogonality of the basis functions on these points, where
$\sum_i \cos(k\delta_i)\cos(k'\delta_i)$ equals $n/2$ when $k = k'$ and vanishes otherwise for
$1 \leq k,k' < n/2$. The measurements are independent with common variance $\sigma^2$, so
\begin{equation}
  \mathrm{Var}(A_k) = \frac{4}{n^2}\sum_i \sigma^2\cos^2(k\delta_i)
  = \frac{4\sigma^2}{n^2}\cdot\frac{n}{2} = \frac{2\sigma^2}{n},
  \label{eq:s6_varA}
\end{equation}
and the same holds for $B_k$. With $\sigma = 0.093$ the shot noise on one point this gives
$\sigma_A = \sigma\sqrt{2/n} = 0.093\sqrt{2/24} = 0.027$. The same orthogonality leaves $A_k$
and $B_k$ uncorrelated with each other and across harmonics. An amplitude
$\sqrt{A_k^2 + B_k^2}$ therefore follows a Rayleigh distribution of scale $\sigma_A$ when no
modulation is present, and its $95$th percentile is $\sigma_A\sqrt{2\ln 20} = 2.45\,\sigma_A$.
The limit of $0.066$ in IS size units quoted in the main text is this percentile at
$\sigma_A = 0.027$, and the $0.031$ from the replication reported below is the same percentile
at $\sigma_A = 0.013$.

\paragraph{Why eight degrees of freedom.}
Each harmonic carries two free parameters, an amplitude and a phase, or equivalently the pair
$A_k$ and $B_k$. The phase of a modulation is not predicted in advance, so both have to be
fitted. Under the hypothesis of no modulation each of $A_k/\sigma_A$ and $B_k/\sigma_A$ is a
standard normal variable and the eight of them are mutually independent, so the chi-squared
statistic
\begin{equation}
  Q = \sum_{k=1}^{4}\frac{A_k^2 + B_k^2}{\sigma_A^2}
  \label{eq:s6_chi2}
\end{equation}
is a sum of eight squared standard normal variables. The four available harmonics therefore
give eight degrees of freedom and not four. Testing the amplitudes on their own would be
incorrect, since $\sqrt{A_k^2 + B_k^2}$ follows a Rayleigh distribution rather than a normal
one. Table~\ref{tab:s6_harmonics} gives the fitted amplitudes. The observed value is
$Q = 8.73$ on eight degrees of freedom, and the largest amplitude in the full decomposition
falls at $k = 9$, outside the range the mechanism can produce.

\begin{table}[!htbp]
\centering
\caption{Fitted harmonic amplitudes of the mean IS size over the sweep, in IS size units,
  against the error $\sigma_A = 0.027$ on each.}
\label{tab:s6_harmonics}
\begin{tabular}{@{}lcccc@{}}
\toprule
Harmonic $k$ & 1 & 2 & 3 & 4 \\ \midrule
Amplitude & 0.020 & 0.052 & 0.013 & 0.055 \\
Amplitude $/\,\sigma_A$ & 0.7 & 1.9 & 0.5 & 2.0 \\
\bottomrule
\end{tabular}
\end{table}

An independent sweep at $100\,000$ shots per step, taken at a separately converged parameter
set, agrees. It gives a chi-squared statistic of $25.1$ on 23 degrees of freedom, amplitudes of $0.011$, $0.012$,
$0.004$ and $0.019$ at the four available harmonics against $\sigma_A = 0.013$, and a bound of
$0.031$ in IS size units.

The paired comparison of the interference and stochastic arms reported under \emph{Ancilla
Interference, Experimental Evidence} in the main text was repeated with a two-seed composition
over eight parameter initialisations. The interference arm
recovered $5, 5, 4, 5, 3, 1, 5, 5$ distinct near-optimal independent sets and the stochastic arm
$2, 4, 5, 2, 3, 6, 2, 4$, with means of 4.1 and 3.5. A paired $t$-test gives $p = 0.54$. The
first four initialisations taken alone give a difference of $+1.5$ and the following four give
$-0.25$, which sets the scale of variation between runs against which single comparisons should
be read.

\paragraph{Null calibration of the channel-resolved statistic.}
The statistic $T$ and the multinomial null it is compared against are defined under
\emph{Channel-Resolved Measurement of the Cross Terms} in the main text. For the
$400\,000$ shot run on the 180-node instance, the $200$ replicates give a null mean of
$8705$ with standard deviation $85$ against an observed $T$ of $8677$, which places the
measurement $0.33$ standard deviations below the mean and yields the $95\%$ upper limit of
$0.011$ on the root mean square contrast quoted in the main text.

\end{document}